\PassOptionsToPackage{dvipsnames}{xcolor}
\documentclass[longbibliography,nofootinbib,aps,prx,superscriptaddress,twocolumn]{revtex4-2}

\usepackage{amsmath,amsfonts,amssymb,amsthm,dcolumn,bm,qcircuit,appendix,mathtools,thmtools,thm-restate,mathrsfs,pgfplots,soul,lipsum,graphicx,braket,enumitem,ragged2e}
\usepackage[caption=false]{subfig}
\usepackage{booktabs}
\usepackage{titlesec}
\usepackage{float}

\usepackage{tikz,tkz-euclide,tikz-3dplot}
\usetikzlibrary{positioning}

\usetikzlibrary{matrix,fit,backgrounds,3d,arrows, decorations.pathreplacing,shapes.geometric,3d, calc}
\usetikzlibrary{arrows.meta,calc}
\tikzset{
yellowv/.style={circle,draw,fill=yellow!80!orange,inner sep=2.2pt},
whitev/.style={circle,draw,fill=white,inner sep=1.9pt},
e/.style={line width=0.8pt},
}
\def\BibTeX{{\rm B\kern-.05em{\sc i\kern-.025em b}\kern-.08em
    T\kern-.1667em\lower.7ex\hbox{E}\kern-.125emX}}

\usepackage[ruled,vlined]{algorithm2e}
\usepackage{algorithmic}

\newtheorem{proposition}{Proposition}
\newtheorem{theorem}{Theorem}
\newtheorem{lemma}{Lemma}
\newtheorem{corollary}{Corollary}
\newtheorem{remark}{Remark}
\newtheorem{definition}{Definition}
\newtheorem{method}{Algorithm}

\newtheorem{pipeline}{Pipeline}

\newcommand{\Ibb}{\mathbb{I}}

\newcommand{\Rbb}{\mathbb{R}}

\newcommand{\Zbb}{\mathbb{Z}}

\newcommand{\norm}[1]{\left\lVert#1\right\rVert}

\newcommand{\Minc}[3]{[#1:#2]_{\xi}^{(#3)}}
\titleformat{\paragraph}[runin]
{\normalfont\bfseries}
{}
{0em}
{}
\usepackage{xcolor}
\usepackage{hyperref}

\hypersetup{
	colorlinks = true,
	linkcolor = [rgb]{0.70,0.13,0.13},
	citecolor = [rgb]{0.13,0.55,0.13},
	urlcolor = [rgb]{0.25, 0.41, 0.88}}
\usepackage[capitalize,nameinlink]{cleveref}

\begin{document}
\title{Quantum Advantage in Topological Data Analysis via Mayer Homology}
\author{Nhat A. Nghiem}
\email{{nhatanh.nghiemvu@stonybrook.edu}}
\affiliation{Google Quantum AI, Santa Barbara, United States}
\affiliation{C. N. Yang Institute for Theoretical Physics, \\ State University of New York at Stony Brook, Stony Brook, NY 11794-3840, United States}

\author{Ryan Babbush}
\affiliation{Google Quantum AI, Santa Barbara, United States}
\author{Adam Zalcman}
\affiliation{Google Quantum AI, Santa Barbara, United States}
\author{Dominic~W.~Berry}
\affiliation{School of Mathematical and Physical Sciences, Macquarie University, Sydney, NSW 2109, Australia}

\author{Trung~V.~Phan}
\affiliation{Department of Natural Sciences, Scripps and Pitzer Colleges, \\ Claremont Colleges Consortium, Claremont, CA 91711, United States}

\author{Guo-Wei Wei}
\affiliation{Departments of Mathematics, Biochemistry \& Molecular Biology,\\
  University of Georgia, Athens, GA 30602, United States}
\author{Ryu Hayakawa}
\email{ryu.hayakawa@yukawa.kyoto-u.ac.jp}
\affiliation{Yukawa Institute for Theoretical Physics \& The Hakubi Center, Kyoto University, Japan}

\begin{abstract}
Prior work has explored quantum algorithms for topological data analysis (TDA), revealing the possibility of exponential quantum speedups in estimating the ratios of Betti numbers to the dimension of the combinatorial Laplacian. However, this quantity is only non-vanishing and efficient-to-quantumly-estimate when Betti numbers are exponentially large, a case for which concrete examples are rarely known. Furthermore, certain randomized classical algorithms are sometimes efficient in this regime. Thus, the prospect of achieving quantum advantage in conventional TDA appears fairly narrow. Here, we address these challenges to the quantum advantage in TDA by developing quantum algorithms for Mayer homology, which generalize simplicial homology to $N$-nilpotent boundary operators ($\partial^N =0$) and have recently been successfully applied to real-world TDA contexts. We introduce an efficient quantum algorithm for estimating Mayer Betti numbers and their persistent counterparts. We then prove that for high-order simplices, Mayer Betti numbers are often exponentially large in the dense regime, which ameliorates the normalization bottleneck of conventional quantum TDA. \textcolor{black}{We extend and analyze classical dequantization approaches in the Mayer setting, obtaining efficient estimators under additional sampling and spectral assumptions. These guarantees do not cover the full inverse-polynomial precision and gap regime accessible to our quantum algorithm.} \textcolor{black}{We also provide logical resource estimates suggesting that, under favorable assumptions, a quantum computer with roughly a few hundred qubits and sixty million Toffoli gates could solve Mayer homology problems beyond the capabilities of current classical methods.} Finally, we discuss real-world applications of Mayer homology in genomics, supersymmetry, drug discovery, and neuroscience, revealing the potential of our quantum algorithm to deliver real-world impacts via Mayer homology.
\end{abstract}

\maketitle

\section{INTRODUCTION}
A fundamental question in the field of quantum computation is whether a quantum computer can efficiently solve problems of practical interest that are intractable for classical computers. Integer factorization is the most well-known example of a problem for which quantum computers offer a superpolyonial speedup ~\cite{shor1999polynomial, regev2025efficient}. In recent years, topological data analysis (TDA) has emerged as another area with potential for practical quantum advantage~\cite{lloyd2016quantum, schmidhuber2022complexity, crichigno2021supersymmetry, nghiem2023quantum, hayakawa2022quantum, hayakawa2024quantum}. TDA is a powerful mathematical framework for extracting multiscale topological invariants from high-dimensional, high-order, and complex datasets. Topological deep learning (TDL), which integrates TDA with deep neural networks, has demonstrated a wide utility in molecular biology, genomics, structural biology, materials science, drug discovery, neuroscience, and machine learning~\cite{wee2025review}. Among the most compelling demonstrations of its practical impact are the repeated successes of TDA-based methods in the D3R Grand Challenges~\cite{nguyen2020review}, a worldwide competition series for computer-aided drug discovery, and the accurate forecasting of the emergence and dominance of the SARS-CoV-2 Omicron subvariants BA.2~\cite{chen2022omicron}, BA.4, and BA.5~\cite{chen2022persistent} approximately two months before they became prevalent. These achievements highlight the unique ability of approaches based on topology to extract biologically meaningful information that is often inaccessible to other mathematical, physical, and statistical methods.

In TDA, a dataset, such as a cloud of points in a metric space, is used to construct discrete representations, typically a filtration of simplicial complexes across spatial scales, of an underlying topological space. The objective is to extract multiscale topological invariants, such as persistent Betti numbers, that capture structural features, including connected components, loops, and voids.

As scientific datasets continue to grow in size and complexity, the computational cost of TDA has become a major bottleneck. The number of relevant combinatorial subcomponents, such as simplices, of a given dataset can be exponentially large in the size of the dataset, posing an inherent challenge for classical computers. At the same time, quantum computing offers an attractive avenue for accelerating TDA, and a series of quantum algorithms have been proposed for estimating both standard Betti numbers of simplicial complexes and persistent Betti numbers across filtrations, which are central problems in TDA. However, recent studies have revealed a fundamental obstacle to realizing practical quantum speedup. Existing quantum TDA algorithms estimate normalized Betti numbers, whose values are typically exceedingly small because ordinary Betti numbers occupy only a tiny fraction of the exponentially large chain spaces \cite{schmidhuber2022complexity}. Consequently, the complexity of amplitude estimation increases substantially, greatly reducing the practical advantage of quantum computation. This bottleneck carries over to persistent Laplacians, since the dimension of their harmonic subspace is exactly the persistent Betti number \cite{wang2020persistent,memoli2022persistent}. 
From a complexity-theoretic perspective, clique-homology problems are known to be hard in the worst case. Crichigno and Kohler showed that clique homology is $\mathrm{QMA}_1$-hard~\cite{crichigno2024clique}, and King and Kohler later established a gapped weighted version that is both $\mathrm{QMA}_1$-hard and contained in $\mathrm{QMA}$~\cite{king2024gapped}. 
This line of work was subsequently extended to the unweighted setting, where gapped clique homology was shown to be $\mathrm{QMA}_1$-complete~\cite{hayakawa2026unweighted}.
These results suggest that efficient Betti-number estimation should not be expected in the worst case. Recently, the work of \cite{berry2024analyzing} has pointed out a particular type of complex, built from the complete partite graph, where superpolynomial quantum speedups may remain possible. However, in the same work, the authors introduced a randomized classical algorithm which can achieve a polynomial scaling in the same regime, thus proving competitive against the best-known quantum algorithm. As such, the window for quantum advantage in estimating Betti numbers, and in TDA as a whole, appears to be very narrow. 
Recent results provide complexity-theoretic evidence for exponential quantum speedups for certain tasks in TDA
\cite{gyurik2026provable,lowe2026complexity}.
Nevertheless, the computational complexity of estimating normalized Betti numbers in practically relevant regimes remains far from fully understood.

\begingroup
\color{black}
This work explores whether a different choice of homology can expand
the regimes accessible to quantum TDA. Mayer homology
\cite{shen2023persistent} replaces the ordinary relation $\partial^2=0$
by $\partial^N=0$ and introduces several homology groups at each degree.
We give sufficient conditions under which their dimensions occupy a
constant fraction of the corresponding chain space, and an explicit
family where this occurs even when that space is exponentially large.
In a specified sector of this family, the signal condition coexists with
a provable gap lower bound and efficient membership access.
Our quantum algorithms for nonpersistent and persistent Mayer Betti estimation
have polynomial scaling under their respective input, signal, and
spectral promises. This yields superpolynomial, and in appropriate regimes
exponential, improvements over classical methods that explicitly enumerate
the chain spaces. We also examine randomized classical benchmarks
\cite{apers2023simple, berry2024analyzing}; these comparisons motivate the
search for advantages beyond enumeration, without establishing a general
classical lower bound. The resulting framework identifies Mayer homology
as a promising setting for large quantum speedups and connects this
question to descriptors already studied in data analysis
\cite{shen2023persistent, feng2025mayer}.
\par
\endgroup

Our work is organized as follows. In Section \ref{sec: overview}, we first provide an overview of homology theory, including both conventional simplicial homology ($N=2$) and Mayer homology ($N\geq2$). Then we provide a statement of our main results, including the complexity for estimating Mayer Betti numbers plus their persistent variants, and discussion of its regime for large quantum speed-up. In the same section, we also provide a specific type of a complex, namely, cone-like complex, that provably has large Mayer Betti numbers and a polynomially bounded spectral gap, thus providing concrete evidence for polynomial quantum running time. Some connections to physics, involving Abelian anyons and fractional supersymmetry, are also discussed. Section \ref{sec: estimatingmayerbetti} is devoted to the outline of the quantum algorithm for estimating Mayer Betti numbers. In this section, we first introduce the key technical recipes, then construct the main quantum algorithm, followed by a discussion on extending such an algorithm to estimating persistent Mayer Betti numbers. We then make several arguments to show that Mayer Betti numbers can be large for many kinds of complexes. In Section \ref{sec: dequantization}, we discuss the classical randomized algorithms introduced in \cite{berry2024analyzing, apers2023simple}, \textcolor{black}{examining the conditions needed for efficient classical estimation in the Mayer homology setting.} In Section \ref{sec: application}, we describe current real-world applications of Mayer homology, revealing possibilities for our quantum algorithms to prove useful in real-world scenarios. Section \ref{sec: discussion} discusses a circuit-to-Mayer-homology perspective and related open questions. The conclusion is given in Section \ref{sec: conclusion}. In Appendix \ref{sec: mayerhomology}, we provide a more formal and detailed introduction to Mayer homology theory, as well as a few examples calculating Mayer Betti numbers for common spaces. Appendix \ref{sec: persistencemayer} covers the theory of persistent Mayer homology. Appendix \ref{sec: anyonic} provides a connection between Mayer homology and anyonic operators, as well as its relation to fractional supersymmetry. In Appendix \ref{sec: Nboundaryoperator}, we provide a detailed construction of the Mayer boundary operator, which was left in the main text. An analysis of our main quantum algorithm (Algorithm \ref{algo: quantumMayerestimation}) is given in Appendix \ref{sec: comlexityanalysis}. Appendices \ref{sec: blockencodingpersistentMayerLaplacian} and \ref{sec: complexitypersistentBettinumber} are devoted to the construction of the persistent Mayer Laplacian and the corresponding complexity analysis. In Appendix \ref{sec: analyzingMayerBetti}, we provide rigorous proofs and numerical studies of the Mayer Betti numbers and spectral gap, thus verifying our arguments in the main text. Appendix \ref{sec: mayercombinatorics} develops general combinatorial formulas for powers of the Mayer boundary and for matrix elements, support, and norm bounds of Mayer Laplacians.

\section{OVERVIEW: HOMOLOGY, QUANTUM ALGORITHMS FOR TDA AND OUR CONTRIBUTION}
\label{sec: overview}
We first give a general background and essential concepts plus the notation required to understand simplicial homology, Mayer homology, and their applications in TDA. Then we provide an overview of existing works involving quantum computation and TDA, followed by a statement of our main contribution regarding a quantum algorithm for Mayer homology.

\begin{figure*}
    \centering
    \includegraphics[width=0.85\linewidth]{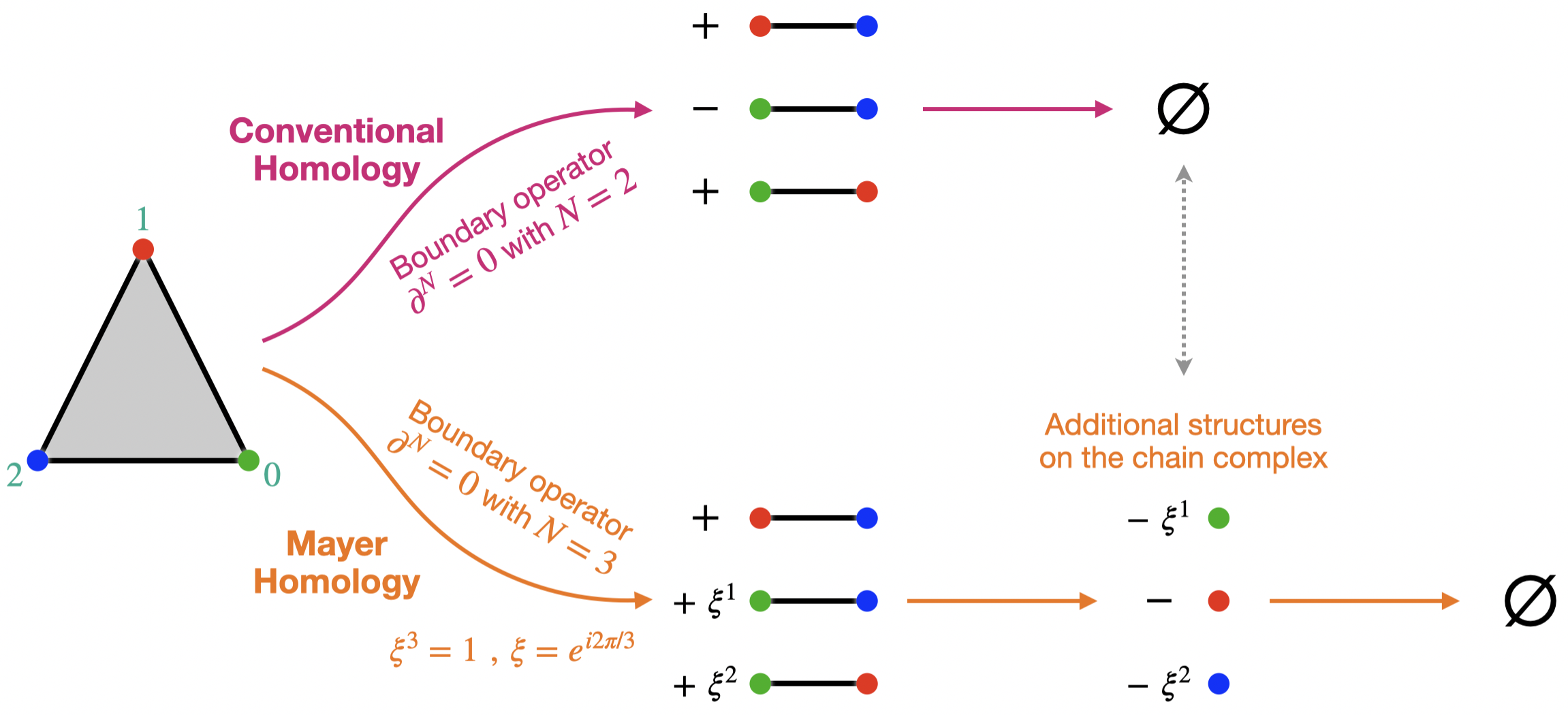}
    \caption{\textbf{The difference between Conventional Homology and Mayer Homology.} In Mayer homology, the chain complex -- specifically, the boundary sequence -- can retain deeper lower-dimensional incidence structure through successive boundary operations. Here, we consider $N=2$ (conventional homology) and $N=3$ (the simplest nontrivial Mayer homology).}
    \label{fig: homology_sum}
\end{figure*}

\subsection*{Background}
In the language of algebraic topology, a point is a 0-simplex, a line segment is a 1-simplex, a triangle is a 2-simplex, and so on. Consider an ordered list of $n$ points $X_0 \equiv [v_0,v_1,...,v_{n-1}]$ and let $\sigma_d$ denote a $d$-simplex. A simplicial complex, denoted by $K$, built out of $X_0$ is a collection of simplices that is closed under taking faces (if $\sigma\in K$ and $\tau$ is a face of $\sigma$, then $\tau\in K$) and where non-empty intersection of any two simplices is their shared face.
Let $S_d^K$ denote the set of $d$-simplices in $K$. The $d$-th chain space $C_d^K$ is constructed as the vector space of formal linear combinations called chains $C_d^K = \mathrm{span}_\mathbb{C}(S_d^K)$ where we choose to use complex coefficients. Geometric relationships between simplices are encoded by linear maps $\partial_d:C_d^K\to C_{d-1}^K$, called boundary operators or differentials. These maps give special significance to certain subspaces of $C_d^K$: some chains vanish under the action of the differential while others arise as differential images of higher-dimensional chains. Homology provides a formal description of the relationship between these subspaces.

The vector spaces $C_d^K$ can be used to assemble a graded vector space $C_*^K$ and the linear maps $\partial_d$ give rise to a degree $(-1)$ linear operator $\partial$ on $C_*^K$. This allows us to express the conditions $\partial_d^N=0$ for all $d$ and fixed $N$ succinctly as $\partial^N=0$.

\paragraph{Conventional homology.} In conventional simplicial homology, the differential is defined by its action on the simplex basis as
\begin{equation}
    \partial_d [v_0v_1,...v_d] = \sum_{k=0}^d (-1)^k [v_0... \hat{v}_k ... v_d]
\end{equation}
where $\hat{v}_k$ indicates that the $k$-th entry is absent from the simplex. It can be shown that $\partial_d^2 \equiv \partial_{d-1} \partial_d = 0$. At each dimension $d$ there are two special types of chains to consider. First, the space $Z_d := \mathrm{Ker}\,\partial_d$ of $d$-dimensional chains which vanish under the action of $\partial$. Second, the space $B_d := \mathrm{Im}\,\partial_{d+1}$ of differential images of $(d+1)$-dimensional chains. The former are called $d$-cycles and the latter $d$-boundaries. Every boundary is a cycle.
Homology measures the failure of the converse using the quotient $H_d := Z_d / B_d$. The dimensions of these quotient spaces are called Betti numbers $\beta_d := \dim H_d$.

\paragraph{Mayer homology.} In Mayer homology, the differential is defined by its action on the simplex basis as
\begin{equation}
    \partial_d [v_0v_1...v_d] = \sum_{k=0}^d \xi^k [v_0...\hat{v}_k...v_d]
\end{equation}
where $\xi = \exp(2\pi i/N)$ for some integer $N\geq 2$. This operator satisfies $\partial_d^N \equiv \partial_{d-N+1} \partial_{d-N+2} ... \partial_d = 0$. As in conventional homology, there are two broad classes of chains to consider: $d$-dimensional chains that vanish under the action of an iterated differential and $d$-dimensional images of higher-dimensional chains. However, unlike in conventional homology, in Mayer homology non-adjacent dimensions may interact non-trivially via a non-vanishing iterated differential. For example, a $(d+2)$-dimensional chain may have a non-zero $d$-dimensional image under $\partial_d^2$. Consequently, the repertoire of relevant chain subspaces is richer and each of the two broad classes of chains splits into $N-1$ subspaces. For any $p=1,\ldots,N-1$, the subspace
\begin{equation}
    Z_{d,p} = \mathrm{Ker}(\partial_d^p: C_d^K \to C_{d-p}^K)
\end{equation}
consists of $d$-dimensional chains that are annihilated by $\partial_d^p\equiv\partial_{d-p+1}\ldots\partial_d$. Similarly, for any $p$, the subspace
\begin{equation}
    B_{d,p} = \mathrm{Im}(\partial_{d+N-p}^{N-p}:C_{d+N-p}^K\to C_d^K)
\end{equation}
consists of $d$-dimensional chains that arise as images of $(d+N-p)$-dimensional chains under the action of 
$\partial_{d+N-p}^{N-p}\equiv\partial_{d+1}\ldots\partial_{d+N-p}$. Clearly, for any dimension $d$, we have
\begin{align}
&Z_{d,1} \subseteq \ldots \subseteq Z_{d,N-1}\\
&B_{d,1} \subseteq \ldots \subseteq B_{d,N-1}.
\end{align}
Furthermore, at every stage $p=1,\ldots,N-1$, we also have $B_{d,p}\subseteq Z_{d,p}$, because $\partial_d^N=0$. Mayer homology measures the failure of the converse using the quotient $H_{d,p} := Z_{d,p} / B_{d,p}$ with the corresponding $p$-th Mayer Betti number $\beta_{d,p} := \dim H_{d,p}$.

\paragraph{Topological data analysis (TDA).} In real-world settings, for example, in network analysis, data are usually given as graphs, and the goal is to extract insight into their structure. A major challenge arises when the data are high-dimensional and large-scale. TDA offers a systematic way to reveal the structure of this kind of dataset. By treating points as 0-simplices, edges as 1-simplices, triangles as 2-simplices, etc, one builds a complex and utilizes simplicial homology for analysis. Betti numbers, e.g., $\beta_d$, reflect the number of $d$-dimensional holes in the complex. Therefore, computing Betti numbers allows us to gain insight into the topological structure, revealing the shape of the dataset. Recently, Mayer homology has been applied in real-world settings \cite{feng2025mayer} and has shown certain strengths compared to conventional simplicial homology. We will discuss in greater detail about the application as well as strength of Mayer homology in Section \ref{sec: application}.

\begin{figure*}[htbp]
    \centering
    \includegraphics[width=0.95\linewidth]{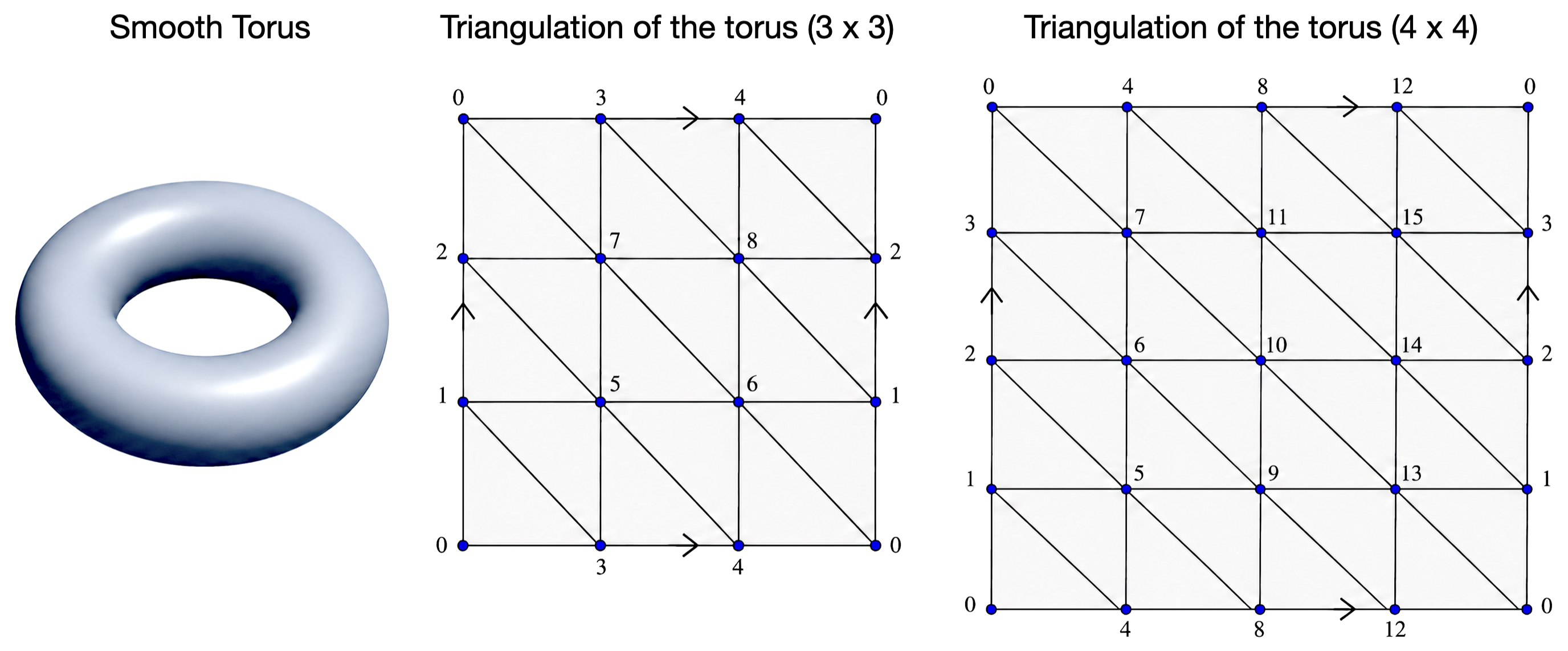}
    \caption{\textbf{A torus and its triangulations.} The $3 \times 3$ triangulation has number of vertices, edges and triangles to be $|S_0| = 9, |S_1| = 27, |S_2| = 18$. The $4 \times 4$ triangulation has number of vertices, edges and triangles to be $|S_0| = 16, |S_1| = 48, |S_2| = 32$.}
    \label{fig: torustriangulation}
\end{figure*}
\paragraph{Remarks.} A more formal and thorough covering of Mayer homology can be found in the Appendix \ref{sec: mayerhomology}. Here we make a few comments on the structure of Mayer homology. First, conventional simplicial homology is the special case of Mayer homology when $N=2$. Therefore, in this work, we only consider $N>2$. Second, the subtle difference between conventional simplicial homology and Mayer homology is that for a given $d$, there are $N-1$ different Mayer Betti numbers $\beta_{d,p}$ ($p=1,2,...,N-1$). So the Mayer Betti numbers are characterized by two parameters $d$ and $p$. As we shall see, the more general $N$-nilpotency property of the boundary map $\partial^N = 0$ results in a change in the algebraic structure, which affects the performance of the quantum algorithm. Third, the boundary operator now has complex entries, whereas in conventional homology the boundary operator has real entries, which are $\{-1,0,1\}$. To gain some insight and illustrate the difference between Mayer homology and conventional homology, we provide an example on the torus (see Fig.~\ref{fig: torustriangulation}), in which we directly calculate the Mayer Betti numbers for two different triangulations. The result is shown in Table \ref{tab:torusMayerBettinumber}. As we can see from this example, Mayer Betti numbers are typically large, comparable to the number of corresponding simplices. In the Appendix \ref{sec: mayerhomology}, we also provide more examples of Mayer Betti numbers for a tetrahedron and a (triangulated) annulus, revealing the same message that Mayer Betti numbers are mostly as large as the number of corresponding simplices. Besides, Mayer Betti numbers are sensitive to triangulations, where a larger number of simplices can lead to larger Mayer Betti numbers. This is opposite to conventional homology, where Betti numbers are unchanged under triangulation, i.e., an invariant. In addition, Mayer Betti numbers also depend on the vertex ordering. So throughout this work, we will implicitly understand that the order of vertices are already specified. Due to this vertex-ordering-dependence (also triangulation-dependence), Mayer Betti numbers are not topological invariants, meanwhile Betti numbers are invariants as they do not depend on the triangulation plus the ordering of vertices. Rather, Mayer homology theory, as a whole, are more of an abstract extension to the $N$-chain complex. As such, Mayer Betti numbers capture the ``generalized'' features of the complex, which is not as intuitive as Betti numbers, which capture the number of loops, holes, voids, etc, in the complex. Instead, Mayer Betti numbers (and their persistent counterparts) can be thought of as \textit{space descriptors}. What eventually matters is that these seemingly abstract descriptors have been successfully applied in real TDA contexts \cite{feng2025mayer, shen2023persistent}, yielding even better results than conventional homology (see also Figure \ref{fig:Barcodes}). It means that, despite the space descriptors as captured by (persistent) Mayer Betti numbers are more abstract and non-intuitive, they can still function as faithful and useful representations of the data points. We emphasize that our main objective is to seek a useful quantum advantage in TDA, rather than a particular invariant. In fact, we will discuss later that even if a quantum computer can efficiently estimate Betti numbers (which are topological invariants) in conventional homology, as shown in \cite{berry2024analyzing}, it will not guarantee a practical advantage in TDA as a whole, partly due to the exponential normalization bottleneck. As we show below, (persistent) Mayer homology offers a promising route as the exponential normalization bottleneck present in conventional homology vanishes, and particularly the application of (persistent) Mayer homology in real-world scenarios has witnessed great success \cite{feng2025mayer, shen2023persistent}. The most important aspect is that the advantage in estimating Mayer Betti numbers plus their persistent counterparts can translate into the practical advantage in TDA, as we will discuss below.

\begin{table}[htbp]
    \centering
    \subfloat[Mayer Betti numbers of a $3 \times 3$ triangulation of a torus for $N=5$.]{
        \label{tab:torus3x3}
        \begin{minipage}{0.45 \textwidth}
        \begin{tabular}{|c|c|c|c|}
            \hline
            & $\beta_{0,p}$ & $\beta_{1,p}$ & $\beta_{2,p}$ \\
            \hline
            $p=1$ & 9 & 19 & 1 \\
            $p=2$ & 9 & 27 & 18 \\
            $p=3$ & 9 & 27 & 18 \\
            $p=4$ & 1 & 10 & 18 \\
            \hline
        \end{tabular}
        \end{minipage}
    }%
    \\
    \subfloat[Mayer Betti numbers of a $4 \times 4$ triangulation of a torus for $N=5$.]{
        \label{tab:torus4x4}
        \begin{minipage}{0.45 \textwidth}
        \begin{tabular}{|c|c|c|c|}
            \hline
            & $\beta_{0,p}$ & $\beta_{1,p}$ & $\beta_{2,p}$ \\
            \hline
            $p=1$ & 16 & 32 & 1 \\
            $p=2$ & 16 & 48 & 16 \\
            $p=3$ & 0  & 48 & 32 \\
            $p=4$ & 0  & 17 & 32 \\
            \hline
        \end{tabular}
        \end{minipage}
    }
    \caption{Mayer Betti numbers of different triangulations of a torus.}
    \label{tab:torusMayerBettinumber}
\end{table}

\subsection*{Quantum computation and TDA}
Existing work connecting quantum computation and TDA falls into the simplicial homology setting, where the coefficients $\{c_i\}$ belong to $\mathbb{R}, \mathbb{Z}_2$. As the literature is extensive, we provide a concise summary of results and progress, and we refer the interested readers to the original works for more details.

Lloyd et al \cite{lloyd2016quantum} first provided a quantum algorithm to estimate the Betti number $\beta_d$ (for all $d$), or more precisely, the normalized Betti numbers $\frac{\beta_d}{|S_d^K|}$ for a given complex. Their method relies on encoding simplices of a complex with $n$ vertices into the computational basis states of $n$-qubit system as follows: the position of bit ``1'' in the computational basis state corresponds to the vertices in the simplex. For example, for a complex having $4$ vertices $v_0,v_1,v_2,v_3$, the simplex $[v_0,v_1,v_2], [v_1,v_3]$ is encoded in $\ket{1110}, \ket{0101}$, respectively. Generally, a $d$-simplex $\sigma_d$ is encoded into $\ket{\sigma_d}$ which is a basis state of Hamming weight $d+1$. The $d$-th chain space $C_d^K$ associated to a complex $K$ is thus the Hilbert space $\mathcal{H}_d := \text{span}\{ \ket{\sigma_d} : \sigma_d \in K  \}$. Based on this, the algorithmic technique of \cite{lloyd2016quantum} combines quantum search (to prepare the uniform superposition $\frac{1}{\sqrt{|S_d^K|}} \sum_{ \sigma_d \in K} \ket{\sigma_d}$), quantum simulation (to simulate $e^{-i\Delta_d}$), and quantum phase estimation (to estimate the dimension of the kernel of $\Delta_d$). This result has motivated many subsequent developments \cite{ubaru2021quantum, berry2024analyzing, mcardle2022streamlined, hayakawa2022quantum, hayakawa2024quantum, nghiem2023quantum, nghiem2025hybrid, lee2025new, schmidhuber2022complexity, schmidhuber2025quantum, crichigno2024clique, gyurik2020towards, gyurik2026provable}. For example, the work \cite{berry2024analyzing} introduced a more efficient amplitude estimation via Kaiser windows, and eigenvalue projectors based on Chebyshev polynomials that replace the phase estimation approach in \cite{lloyd2016quantum}. In parallel, the work \cite{hayakawa2022quantum} outlined a quantum algorithm for estimating persistent Betti numbers, which quantify those ``holes'' that persist across different filtration. The work \cite{nghiem2023quantum} introduced a cohomology approach for estimating Betti numbers of triangulated manifolds, thus providing a dual approach to the standard approach based on homology.

However, the problem of estimating Betti numbers turns out to be more difficult. On the complexity-theoretic side, the work \cite{crichigno2024clique} showed that determining if Betti numbers $> 0$ is $\rm QMA_1$-hard. This result is even strengthen in the work \cite{schmidhuber2022complexity}, where it was shown that estimating Betti numbers is $\rm NP$-hard, while counting it exactly is $\# \rm P$-hard. These results have implied that quantum computer cannot generically achieve efficient scaling in estimating Betti numbers, under widely believed complexity-theoretic assumptions.

Some recent works \cite{gyurik2020towards, gyurik2026provable, schmidhuber2025quantum, lowe2026complexity}, on the other hand, have opened new directions for quantum speed-up in TDA. Specifically, in \cite{gyurik2020towards}, it was shown that a more general spectral-density estimation problem, which includes normalized Betti number estimation as a special case, is $\mathrm{DQC}_1$-hard. The work \cite{gyurik2026provable} studied persistent homology from a quantum-complexity perspective, showing that determining whether a hole persists across a filtration is $\mathrm{BQP}_1$-hard and in $\mathrm{BQP}$, providing complexity-theoretic evidence for exponential quantum speed-ups for certain TDA tasks. Related complexity-theoretic evidence was also obtained in \cite{lowe2026complexity}. The work \cite{schmidhuber2025quantum} has proposed an efficient quantum algorithm for Khovanov homology and proved that estimating the rank of Khovanov homology is $\mathrm{DQC}_1$-hard.

\subsection*{Our contribution }
We first consider Mayer homology and target the problem of estimating Mayer Betti numbers of a complex with a specified order of vertices. Subsequently, we extend the algorithm to the persistent realm, where we provide a quantum algorithm for estimating persistent Mayer Betti numbers. Similar to previous works \cite{lloyd2016quantum, berry2024analyzing, gyurik2020towards, ubaru2021quantum}, the input model to us is the oracle that can verify the existence of simplices in the complex $K$ of interest. More concretely, for any order $\textcolor{black}{0} \leq d \leq n-1$, the oracle $O_d^K$ acts on the encoded simplex state as follows: $O_d^K\ket{\sigma_d}\ket{0}= \ket{\sigma_d}\ket{1}$ when $\sigma_d\in K$ and otherwise $O_d^K\ket{\sigma_d}\ket{0}= \ket{\sigma_d}\ket{0}$. The explicit construction of this oracle is given in \cite{lloyd2016quantum} (based on QRAM) and \cite{berry2024analyzing} (based on Toffoli gates plus an explicit graph description of $K$; see also Section \ref{sec: application} where we recapitulate the construction of this oracle above the Pipeline \ref{algo: practicalpipeline}). From this oracle, we first show how to build the block-encoding (see Section \ref{sec: estimatingmayerbetti} for definition) of the $N$-boundary operator. Then by utilizing block-encoding/QSVT recipes (from \cite{gilyen2019quantum}), we construct the block-encoding of the Mayer Laplacian. Then we estimate the (normalized) kernel dimension of this Mayer Laplacian by using QSVT plus amplitude estimation, which reveals the (normalized) Mayer Betti numbers. A more detailed description of the quantum algorithm is given in Section \ref{sec: estimatingmayerbetti}, with the full analysis given in the Appendix \ref{sec: comlexityanalysis}. Here, we summarize the result in the following theorem.
\textcolor{black}{For a positive quantity $x$, relative error $\delta$
means that the estimate $\widehat{x}$ satisfies
$|\widehat{x}-x|\leq\delta x$. The notation $\widetilde O$ hides
polylogarithmic factors.}

\begingroup\color{black}
For a complex on $n$ vertices, our circuit represents each boundary map
$\partial_k$ after division by the normalization factor
\[
 \alpha_k=\sqrt{(k+1)(n-k)},\qquad 1\leq k\leq n-1.
\]
The factors $k+1$ and $n-k$ count the possible vertex deletions
and additions, respectively. For a target $(d,p)$, the lower
boundary power uses $k=d-p+1,\ldots,d$ when $d\geq p$, and the
upper boundary power uses $k=d+1,\ldots,d+N-p$ when
$d+N-p\leq n-1$. A power outside these degree ranges is zero
and is omitted. We write $\alpha_{\min}$ and $\alpha_{\max}$ for
the minimum and maximum of $\alpha_k$ over the boundary maps
used in these powers. They depend only on $n,d,p,N$, not on
additional input data, and summarize the normalization factors
in the complexity bounds below.
In the two theorems below, we count each elementary gate and each
membership-oracle call as one operation (the unit-cost oracle model).
Controlled and inverse oracle calls are counted in the same way.
\par\endgroup

\begin{theorem}[Quantum algorithm for estimating Mayer Betti numbers of order-specified complex]
\label{thm: mainresult}
    Let $K$ be the complex having $n$ order-specified vertices. Let $|S_d^K|$ denotes the number of $d$-simplices.
    \textcolor{black}{For a target degree $1\leq d\leq n-1$ and
    $1\leq p\leq N-1$, assume access to the membership oracles}
    {\color{black}
    \[
        O_d^K,\qquad O_{d+N-p}^K\quad(d+N-p\leq n-1).
    \]
    }
    \textcolor{black}{Assume $|S_d^K|>0$ and $\beta_{d,p}>0$.
    The normalized value $\beta_{d,p}/|S_d^K|$ can be estimated
    to relative error $\delta$, with constant success probability,
    using the following number of operations:}
    \begin{itemize}
        \item For $d \geq p:$ \\ 
        \begin{equation}
            \mathcal{\tilde{O}}\left( \left[n^2  N^2 \frac{\alpha_{\max}^{2p} + \alpha_{\max}^{2(N-p)}}{\gamma_d^p \alpha_{\min}}  +   \sqrt{\frac{\binom{n}{d+1}}{|S_d^K|} } dn\right]\sqrt{ \frac{|S_d^K|}{ \beta_{d,p}}} \frac{1}{\delta} \right)\notag
        \end{equation}
        \item For $d < p:$  \\ 
        \begin{equation}
            \mathcal{\tilde{O}}\left( \left[n^2  N^2 \frac{\alpha_{\max}^{2(N-p)}}{\gamma_d^p \alpha_{\min}}  +   \sqrt{\frac{\binom{n}{d+1}}{|S_d^K|} } dn\right]\sqrt{ \frac{|S_d^K|}{ \beta_{d,p}}} \frac{1}{\delta} \right)\notag
        \end{equation}
    \end{itemize}
    \textcolor{black}{The supplied $\gamma_d^p>0$ is a
    lower bound on the smallest positive eigenvalue of
    $\Delta_{d,p}$.}
\end{theorem}
\begingroup\color{black}
Separate query counts and the gate cost of implementing the oracles
are accounted for in Appendix~\ref{sec: comlexityanalysis}.
If $|S_d^K|$ is known, multiplying by it gives an estimate of
$\beta_{d,p}$ with the same relative error. Otherwise, its estimation
cost and error must also be included.
\par\endgroup

\paragraph{Estimating persistent Mayer Betti numbers. } Our quantum algorithm above can be extended to estimate the persistent Mayer Betti numbers of a filtration $K_a \subseteq K_b$. Briefly speaking, while Mayer Betti numbers $\beta_{d,p}$ indicate the number of generators in the Mayer homology group $H_{d,p}^{K_a}$ of the given complex $K_a$, persistent Mayer Betti numbers $\beta_{d,p}^{a,b}$ quantify how many of them would maintain the generators for $H_{d,p}^{K_b}$ of $K_b$. This numbers are more relevant to practical application of Mayer homology, as they are the main components for building the persistence barcode, or diagram. In Section \ref{sec: estimatingmayerbetti}, we will show that by utilizing the same block-encoding recipes as above, we can build the so-called persistent Mayer Laplacian.
\begingroup
\color{black}
\begin{theorem}[Quantum algorithm for estimating persistent Mayer Betti numbers]
\label{thm: persistent-mainresult}
Let $K_a\subseteq K_b$ be simplicial complexes on a common register of
$n$ vertices with a fixed global order. For $1\leq d\leq n-1$ and
$1\leq p\leq N-1$, assume membership-oracle access to
\[
\begin{aligned}
&O_d^{K_a},\quad O_d^{K_b},\\
&O_{d+N-p}^{K_b}\quad(d+N-p\leq n-1),
\end{aligned}
\]
Assume $|S_d^{K_a}|>0$ and $\beta_{d,p}^{a,b}>0$, and supply
two positive spectral-gap lower bounds:
$\gamma_\Delta$ for the normalized block whose pseudoinverse is used
in the construction, and $\gamma_{d,p}^{a,b}$ for the final,
unnormalized persistent Mayer Laplacian. The precise normalization
and assumptions on $\gamma_\Delta$ are given in
Lemma~\ref{lemma: blockencodingpersistentLaplacian}.
For $0<\delta<1$, the normalized value
$\beta_{d,p}^{a,b}/|S_d^{K_a}|$ can be estimated to relative error
$\delta$, with constant success probability, using the following
number of operations:
\begin{itemize}
 \item For $d\geq p$:
 \begin{equation}
 \begin{aligned}
 \widetilde O\Bigg(
 \Bigg[
 \frac{n^2N^3(\alpha_{\max}^{2p}+
              \alpha_{\max}^{2(N-p)})}
      {\gamma_\Delta^2\gamma_{d,p}^{a,b}\alpha_{\min}}
 +dn\sqrt{\frac{\binom n{d+1}}{|S_d^{K_a}|}}
 \Bigg]\frac1\delta
 \sqrt{\frac{|S_d^{K_a}|}{\beta_{d,p}^{a,b}}}
 \Bigg).
 \end{aligned}
 \notag
 \end{equation}
 \item For $d<p$:
 \begin{equation}
 \begin{aligned}
 \widetilde O\Bigg(
 \Bigg[
 \frac{n^2N^3\alpha_{\max}^{2(N-p)}}
      {\gamma_\Delta^2\gamma_{d,p}^{a,b}\alpha_{\min}}
 +dn\sqrt{\frac{\binom n{d+1}}{|S_d^{K_a}|}}
 \Bigg]\frac1\delta
 \sqrt{\frac{|S_d^{K_a}|}{\beta_{d,p}^{a,b}}}
 \Bigg).
 \end{aligned}
 \notag
 \end{equation}
\end{itemize}
Here $\alpha_{\max}$ and $\alpha_{\min}$ range over the
nonvanishing boundary factors, as defined above.
If the block to be inverted is known to be zero, no inverse-gap
promise is required and the inverse-gap factors are omitted as
described in Lemma~\ref{lemma: blockencodingpersistentLaplacian}.
\end{theorem}
The proof and a more precise composition bound are given in
Appendix~\ref{sec: complexitypersistentBettinumber} and
Eq.~\eqref{eq:persistent-main-cost}.
That appendix also gives separate query counts and accounts for
the gate cost of implementing the oracles.
If $|S_d^{K_a}|$ is known, the same estimate yields
$\beta_{d,p}^{a,b}$ to relative error $\delta$; otherwise the cost
of estimating this normalization must also be included.
\par
\endgroup

\paragraph{Regime for polynomial complexity.} Given that $\alpha_{\max}$ is upper bounded by $n+1$, $p < N =\mathcal{O}(1)$, then if $\gamma_d^p$ grows at worst $\Omega\left(\frac{1}{\text{poly } n}\right)$ and the ratios $\frac{\binom{n}{d+1}}{|S_d^K|},  \frac{|S_d^K|}{ \beta_{d,p} } \in \mathcal{O}(\text{poly } n)$, our quantum algorithm achieves polynomial scaling in the number of vertices $n$. The best-possible complexity of our quantum algorithm is achievable when $\frac{\binom{n}{d+1}}{|S_d^K|},  \frac{|S_d^K|}{ \beta_{d,p} } = \mathcal{O}(1)$. We remark two things. First, the complexities above do not include the gate cost for the oracle $O_d^K$. As mentioned, the direct construction of this oracle appeared in \cite{lloyd2016quantum} and \cite{berry2024analyzing}, which has complexity $\mathcal{O}(d^2)$ and $\mathcal{O}(|E| + \log d)$ ($|E|$ is the number of edges in $K$, which is $\mathcal{O}(n^2)$), respectively. \textcolor{black}{Polynomial-cost membership oracles preserve polynomial-time scaling under the stated promises, but can change its polynomial degree. Their cost must be multiplied by the full query count, as in Eqs.~\eqref{eq:mayer-total-queries} and \eqref{eq:persistent-total-queries}.} For curious readers, we recapitulate the construction of $O_d^K$ in \cite{berry2024analyzing} in Section \ref{sec: application} (see paragraph above Pipeline \ref{algo: practicalpipeline}). Second, the polynomial complexity fundamentally relies on the premise that the spectral gap $\gamma_d^p$ of $\Delta_{d,p}$ is lower-bounded by an inverse polynomial in $n$. Alternatively, it can be said that our quantum algorithm achieves polynomial scaling for any family of complex that at worst has inverse-polynomial growing of the gap. Although deriving a rigorous analytical lower bound for the Mayer Laplacian gap of random complexes remains a highly non-trivial open mathematical challenge, we provide empirical justification for this premise in Figure \ref{fig: MayerBettivsSimplices} and Figure \ref{fig: packfigure1} (in Section \ref{sec: estimatingmayerbetti}). Our finite-size numerical simulations the Costa-Farber random complexes show no evidence of exponential gap closure, thus providing a heuristic evidence for the premise on the gap. Below, we will provide a specific type of complex that has a large normalized Mayer Betti number and a provably inverse-polynomial growing of the gap.

\paragraph{Bottleneck of prior quantum homology algorithms.} The regime of efficiency above is similar to previous quantum homology works \cite{lloyd2016quantum, berry2024analyzing, hayakawa2022quantum, ubaru2021quantum, nghiem2023quantum, lee2025new}, where the polynomial quantum complexity is achieved when the simplices are dense and the Betti numbers $\beta_d$ is as large as the corresponding number of simplices (plus the inverse-polynomial lower bound premise on the growth of the spectral gap). However, this turns out to be a bottleneck of prior quantum homology algorithms, as the families of complexes having large Betti numbers in the dense complex regime are seemingly rare. In \cite{berry2024analyzing}, the authors have pointed out that a complex built from the so-called partite complete graph would have large Betti numbers. So far, this is only the known case with provably large Betti numbers, which then implies that the prospect for achieving quantum speedup in TDA is rather narrow. Furthermore, in the same work \cite{berry2024analyzing} (and also \cite{apers2023simple}), the authors proposed a dequantization algorithm, which achieves a polynomial scaling in the number of vertices under appropriate promises. Due to this, as also emphasized in \cite{berry2024analyzing}, having a large number of simplices and large Betti numbers are not sufficient for superpolynomial speed-up. The challenges above thus naturally transfer to our Mayer homology setting. Two important aspects are how Mayer Betti numbers $\beta_{d,p}$ behave in general and how the best-known dequantization algorithms of \cite{apers2023simple, berry2024analyzing} would perform in the Mayer homology setting, especially in the dense regime, where our quantum algorithm achieves the best performance.

\begingroup
\color{black}
\paragraph{Sufficient conditions for large normalized Mayer Betti numbers.}
The Mayer Betti numbers computed for the torus in
Table~\ref{tab:torusMayerBettinumber} and for the tetrahedron and annulus
in \hyperref[sec:appendix-examples]{``Examples of Mayer Homology''}
in Appendix~\ref{sec: mayerhomology} motivate the following general bound.
\begin{proposition}[A sufficient condition for large normalized Mayer Betti numbers]
\label{prop:large-normalized-mayer-betti}
Let $K$ be an ordered simplicial complex with $n$ vertices, and write
$f_k=\dim C_k^K=|S_k^K|$, with $f_k=0$ outside $0\leq k\leq n-1$.
For $f_d>0$ and $1\leq p<N$,
\[
    \beta_{d,p}\geq f_d-f_{d-p}-f_{d+N-p}.
\]
Consequently, if $f_{d-p}+f_{d+N-p}\leq(1-c)f_d$ for some
$0<c\leq1$, then $\beta_{d,p}/f_d\geq c$.
This bound holds for every vertex order.
\end{proposition}
\begin{proof}
Consider the two maps defining the degree-$d$ Mayer homology:
\[
 \begin{aligned}
 A&=\partial^p:C_d^K\longrightarrow C_{d-p}^K,\\
 B&=\partial^{N-p}:C_{d+N-p}^K\longrightarrow C_d^K.
 \end{aligned}
\]
Since $AB=\partial^N=0$, the boundary space $\operatorname{im}B$
is contained in the cycle space $\ker A$. Hence
\[
 \begin{aligned}
 \beta_{d,p}
 &=\dim\ker A-\dim\operatorname{im}B\\
 &=f_d-\operatorname{rank}A-\operatorname{rank}B.
 \end{aligned}
\]
The second equality is rank--nullity for $A$. The rank of $A$ is
at most its target dimension $f_{d-p}$, while the rank of $B$ is
at most its source dimension $f_{d+N-p}$. Substituting these two
bounds proves the inequality, and dividing by $f_d$ gives the
stated consequence. Only chain-space dimensions enter the bound,
so it does not depend on the vertex order.
\end{proof}

For densities $f_d=p_1\binom{n}{d+1}$,
$f_{d-p}=p_2\binom{n}{d-p+1}$, and
$f_{d+N-p}=p_3\binom{n}{d+N-p+1}$, a sufficient pair of conditions is
\begin{align}
    \frac{p_3}{p_1}
    &\leq\frac14\left(\frac{d+2}{n-d-1}\right)^{N-p},\notag\\
    \frac{p_2}{p_1}\left(\frac{d+1}{n-d}\right)^p
    &\leq\frac34-c,\qquad 0<c<\frac34.
    \label{eq:main-mayer-density-conditions}
\end{align}
Here $p\leq d$, $d+N-p\leq n-1$, and $N,p$ are fixed.
When $p_2/p_1=O(1)$, the second condition holds eventually for some
$c>0$ at $d=\Theta(n^{1-\mu})$, $0<\mu<1$, or
$d=\Theta(n/\log n)$. At $d/n\to1/\zeta$ with fixed $\zeta>2$,
it instead requires a quantitative margin; a sufficient condition is
$\limsup_{n\to\infty}(p_2/p_1)/(\zeta-1)^p<3/4$.
The first condition must be retained with its explicit constant in all
three regimes. Appendix~\ref{sec: analyzingMayerBetti} gives the details.
These conclusions apply to each $p$ for which the conditions hold.
Whenever $f_d$ is exponentially large, they also give an exponentially
large Mayer Betti number. The bound is independent of vertex order.
Figure~\ref{fig: MayerBettivsSimplices} provides complementary numerical
examples, rather than a guarantee for arbitrary dense complexes.
\par
\endgroup

\begin{figure*}
    \centering
    \includegraphics[width=0.8\linewidth]{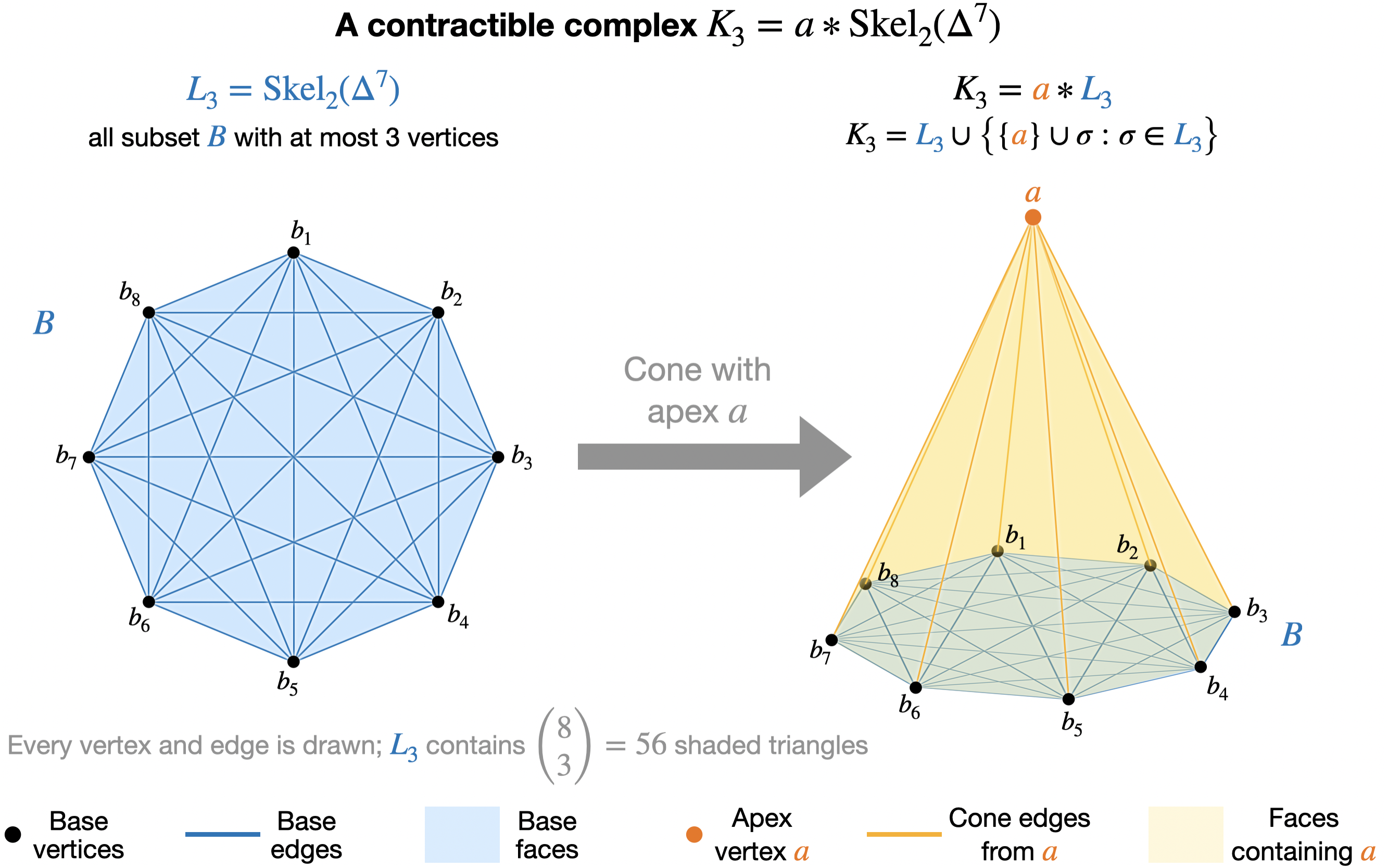}
    \caption{ Illustration for an example of the cone-complex $ K = a*\text{Skel}_{m-1} (\Delta^{3m-2}) $, in which we choose $m=3$. }
    \label{fig: conecomplex}
\end{figure*}

\begingroup\color{black}
\paragraph{A specific family of complexes with large normalized Mayer Betti numbers and an inverse-polynomial spectral gap.}
We now give a family for which a large normalized Mayer Betti number,
an inverse-polynomial gap, and efficient membership access coexist.
Fix $N\geq3$ and $m\geq N-1$, and let $n=3m$.
Choose an apex vertex $a$ and a set $B$ of $3m-1$ other vertices.
Let $L_m=\operatorname{Skel}_{m-1}(\Delta^{3m-2})$ consist of all
subsets of $B$ with at most $m$ vertices. Here
$\operatorname{Skel}_r(K)$ denotes the subcomplex consisting of all
simplices of dimension at most $r$, and $\Delta^{3m-2}$ denotes
the full simplex on the $3m-1$ vertices of $B$, including all its faces.
Define
\[
 K_m=a*L_m
 =L_m\cup\{\{a\}\cup\sigma:\sigma\in L_m\}.
\]
Figure~\ref{fig: conecomplex} illustrates the construction for $m=3$.

\begin{theorem}[A cone family with large normalized Mayer Betti numbers and a spectral gap]
\label{thm:cone-family}
For the complex $K_m$ above, fix any global vertex order and set
$d=m-1$. Then:
\begin{enumerate}
 \item Every $(d+1)$-vertex subset is a simplex, so
 \[
 |S_d^{K_m}|=\binom{3m}{m}=\exp(\Theta(n)).
 \]
 For fixed $N$ and every $1\leq p<N$,
 \[
 \frac{\beta_{d,p}}{|S_d^{K_m}|}=\Omega(1).
 \]
 \item If $p=N-1$ and $m\equiv1\pmod N$, the smallest positive
 eigenvalue of $\Delta_{m-1,N-1}$ is at least
 \[
 (n+1)^{-(2N-2)}.
 \]
 \item Membership oracles for $K_m$ can be implemented using
 polynomially many elementary gates.
\end{enumerate}
\end{theorem}
\begin{proof}
The simplex-count and normalized-signal bounds, the gap estimate,
and the membership construction are proved in
Appendix~\ref{sec: provingconecomplex}.
\end{proof}

At $d=m-1\sim n/3$, the simplex fraction is one.
Consequently, Theorem~\ref{thm: mainresult} gives polynomial quantum
running time for fixed $N$, $p=N-1$, $m\equiv1\pmod N$, and
inverse-polynomial relative precision, including oracle implementation.
The gap lower bound does not imply that the actual gap is small,
and by itself gives no conclusion about the efficiency of the
classical algorithms in \cite{berry2024analyzing,apers2023simple}.
\par\endgroup

\begingroup
\color{black}
Since $K$ is a cone, it is contractible and its ordinary Betti numbers
vanish in positive degrees. In contrast, at $d=m-1$ its Mayer Betti
numbers satisfy
\[
    \beta_{d,p}=\Theta\!\left(|S_d^K|\right)
    =\Theta\!\left(\binom{3m}{m}\right)=\exp(\Theta(n))
\]
for fixed $N\geq3$ and $1\leq p\leq N-1$.
Thus, Mayer homology occupies a constant fraction of the exponentially
large chain space even though ordinary homology is trivial at this degree.
Together with the gap bound in the sector $p=N-1$, $m\equiv1\pmod N$,
and the efficient membership test, this gives an explicit family where
the density, signal, and spectral conditions for efficient quantum
estimation coexist. Its role is analogous to that of the complete
multipartite graph example in \cite[Sec.~IV.1]{berry2024analyzing},
which exhibits large ordinary Betti numbers and a large spectral gap.
This comparison concerns favorable algorithmic parameters, not a
quantum--classical separation: the cone construction alone does not
establish classical computational difficulty.
\par
\endgroup

\begingroup
\color{black}
\paragraph{Comparison with classical methods.}
The methods of \cite{apers2023simple, berry2024analyzing} provide important
benchmarks beyond enumeration. Their guarantees depend on the normalized
spectral gap and precision, and, for the sampling constructions considered
here, on variance and mixing estimates. Our discussion in
Section~\ref{sec: dequantization} examines these parameters for Mayer
Laplacians; it does not establish a matching polynomial-time guarantee
throughout the quantum regime. This absence of a matching guarantee is
not a proof that the classical methods are inefficient on every such
instance. Nor must classical access begin by constructing the entire
Laplacian: local face and coface queries can supply its entries for fixed
$N$, and sufficiently dense simplex sets admit rejection sampling from
uniform vertex subsets. Comparisons must therefore use matched input
access, precision, and structural assumptions.

A straightforward exact method, by contrast, enumerates the relevant
simplices, constructs the boundary matrices explicitly, and computes
their ranks using exact arithmetic over the cyclotomic field. The Mayer
Betti number is then obtained by subtracting the ranks of the two
boundary powers from $|S_d^K|$, as in
Proposition~\ref{prop:large-normalized-mayer-betti}.
This method requires at least $\Omega(|S_d^K|)$ time just to enumerate
the degree-$d$ simplices. When $|S_d^K|=\exp(\Theta(n))$, enumeration
alone therefore takes exponential time, even before the rank computations.
This describes the cost of the explicit enumeration-based method, not
a lower bound for all exact algorithms or approximate estimators.

\paragraph{Prospects for large quantum speedups.}
Mayer homology provides a concrete route around the normalization
bottleneck of LGZ-type estimation. Our sufficient conditions and explicit
cone family show that an exponentially large chain space can contain a
constant fraction of harmonic states; in the specified cone sector,
efficient membership access and an inverse-polynomial gap lower bound
also coexist. Together with the improvement over enumeration-based
methods, these results make Mayer homology a promising setting for
investigating superpolynomial quantum advantages beyond those benchmarks.
The cone is an example of compatible algorithmic conditions, not a
classically hard instance. Establishing an advantage over the best known
classical methods requires an instance family and parameter regime for
which those methods can be compared under the same promises.
\par
\endgroup

\paragraph{Practical advantage in TDA via Mayer homology.} \textcolor{black}{Estimating normalized Mayer Betti numbers and their persistent counterparts provides multiscale information that could be used in TDA, without necessarily reconstructing a barcode.} In TDA, the standard procedure starts with a given data points, then build a nested sequence of complexes $K_1 \subseteq K_2 \subseteq ... \subseteq K_m$, followed by building the persistent barcode/diagram to obtain the topological and geometrical summary of such data points. A feature vector is then constructed from the barcode (see Section \ref{sec: application} for a discussion on vectorization strategies), which can then be fed into regression or machine learning models \cite{feng2025mayer, shen2023persistent}. As illustrated in Figure \ref{fig:Barcodes}, Mayer homology barcodes are typically more informative compared to conventional simplical homology barcodes, suggesting that Mayer homology may be more useful in reality. This has been experimentally demonstrated in \cite{shen2023persistent, feng2025mayer}, where the authors applied persistent Mayer homology in real-world contexts and achieved better results than state-of-the-art methods, including persistent simplicial homology.

\textcolor{black}{Our quantum algorithm estimates normalized Mayer Betti
numbers and their persistent counterparts at selected filtration levels
and pairs, under the stated algorithmic assumptions. For efficient
relative estimation in the persistent case, the normalized persistent
rank must itself be sufficiently large.}\footnote{\textcolor{black}{Large
Mayer Betti numbers at the two endpoints do not ensure a large persistent
rank: $\beta_{d,p}^{a,b}\leq\min\{\beta_{d,p}^{a},\beta_{d,p}^{b}\}$
provides only an upper bound.}}
At the same time, the same thing is not likely to occur in conventional TDA. To show this, we again consider a nested sequence $K_1 \subseteq K_2 \subseteq ... \subseteq K_m$. Assuming that at some filtration level $j$, the complex $K_j$ has the same structure as the one provided in \cite{berry2024analyzing}, which means the Betti numbers of $K_j$ are large. However, at other levels of filtration $i > j$, $K_j$ is not guaranteed to have this structure, thus losing the guaranty of having large Betti numbers. In fact, in  \cite{linial2006homological, kahle2009topology, meshulam2009homological}, it was proved that beyond a relatively small probability threshold, there is a high probability that homology vanishes at all orders. Because the probability threshold is relatively small, it means that the homology vanishes even at a relatively sparse regime, and more certainly in a denser regime. Furthermore, in \cite{schmidhuber2022complexity}, see their Section V, it was also mentioned that for Vietoris-Rips complex (which frequently emerges in realistic TDA context), in both sparse and dense regime, the averaged number of Betti numbers (at a given order $r$) is $\mathbb{E}(\beta_r)$ grows at most linearly in the number of vertices $n$. Whereas Mayer Betti numbers are large for dense complexes, simplicial Betti numbers are very likely to be small in the same regime, except when the complex has the structure described in \cite{berry2024analyzing}. Therefore, even if the numbers are large at $K_i$, the Betti numbers are likely to be (much) smaller at $K_j$ (for $j > i$ and $j < i$), especially when $K_j$ is considerably denser than $K_i$, which then implies the small values of persistent Betti numbers. As such, even if there is an advantage in estimating the Betti numbers at the specific level $K_i$, it does not guaranty efficiency at other levels. This implies that the construction of the (approximate) persistent barcode can be costly, and thus, as a whole, there is no practical advantage in TDA.

\begingroup
\color{black}
The estimated normalized persistent Mayer Betti numbers provide a
multiscale summary of the filtration and can be considered directly as
candidate features for data analysis. They also connect to barcode
representations: exact persistent ranks determine barcode multiplicities
through M\"obius inversion (inclusion--exclusion)
\cite{kim2021generalized}. Extending this reconstruction to estimated
ranks requires appropriate error control.\footnote{\textcolor{black}{Recovering
unnormalized ranks introduces chain-dimension factors, and barcode
multiplicities are differences of ranks. Thus, accurate barcode
reconstruction does not follow automatically from relative-error
estimates of normalized persistent ranks.}}
\par
\endgroup

\paragraph{General combinatorial formulas for Mayer boundaries and Laplacians.}
Alongside the quantum algorithms, we develop a general combinatorial calculus for Mayer boundaries and Laplacians in Appendix~\ref{sec: mayercombinatorics}. This includes a closed form for arbitrary powers of the Mayer boundary and explicit formulas for Mayer-Laplacian matrix elements, diagonal degrees, support, and operator-norm bounds. These results provide a reusable toolkit for the algorithmic and dequantization analyzes considered here, as well as analyzing the related complexity-theoretic hardness related to Mayer homology. This is one of the potential future direction that can stem from our work.

\begingroup\color{black}
\paragraph{An operator representation of the Mayer boundary.}
Previous work has connected ordinary homology with supersymmetric
quantum mechanics through fermionic representations of boundary
operators~\cite{crichigno2021supersymmetry,crichigno2024clique}.
We give an explicit representation of the ordered Mayer boundary in terms
of hard-core anyonic annihilation operators. A generalized Jordan--Wigner
transformation expresses these operators on qubits, replacing fermionic
parity strings with strings of cyclotomic phase gates. This representation
connects the position-dependent coefficients of the Mayer boundary to
operator exchange relations and recovers the fermionic case at $N=2$.
It also provides an algebraic starting point for exploring connections
with quantum many-body models. Its $N$-nilpotency and grading by degree
modulo $N$ also suggest a structural analogy with fractional supersymmetry.
Developing this analogy into a concrete physical model is a direction for
future work. The operator construction and these algebraic connections
are discussed in Appendix~\ref{sec: anyonic}.
\par\endgroup

\section{QUANTUM ALGORITHM FOR ESTIMATING MAYER BETTI NUMBERS}
\label{sec: estimatingmayerbetti}
In this section, we first outline the quantum algorithm for estimating Mayer Betti numbers of an order-specified complex. Our algorithm is built on several technical lemmas and recipes, which are shown below.
\subsection*{Key recipes}
\paragraph{Block-encoding and quantum singular value transformation (QSVT) \cite{gilyen2019quantum}.} A unitary $U$ is said to be a $(\alpha, a, \epsilon)$-encoding of a matrix $A$ if 
$$\left[\bra{0}^{\otimes a} \otimes \Ibb\right] U \left[\ket{0}^{\otimes a}\otimes \Ibb\right] = \frac{\Tilde{A}}{\alpha}$$ 
and $\|\Tilde{A}-A \|_{o} \leq \epsilon$ ($\|.\|_o$ denotes operator norm), so that $U = \begin{pmatrix}
    \frac{\tilde{A}}{\alpha} & \cdot \\
    \cdot & \cdot
\end{pmatrix}$. Given the block-encoding $U_1,U_2$ of operators $A_1, A_2$, there are arithmetic operations involving them. For example, the block-encoding of $A_1A_2$ (product Lemma \ref{lemma: product}), $\alpha A_1 +\beta A_2$ (linear combination Lemma \ref{lemma: linearcombination}), $A_1 \otimes A_2$ (tensor product \cite{camps2020approximate}) can be constructed with a single use of $U_1,U_2$ plus their adjoints and controlled versions. A more formal statements of these block-encoding recipes are summarized in the Appendix \ref{sec: blockencodingrecipes}.


\paragraph{Input access model and encoding of simplices. } First, similar to \cite{lloyd2016quantum, berry2024analyzing, hayakawa2022quantum}, we encode the simplices $\{\sigma_d\}$ into the computational basis $\{\ket{\sigma_d}\}$ that has appropriate Hamming weight (as described earlier). Second, for a fixed $d$, it is a known property that $\beta_{d,p} = \dim \ker\Delta_{d,p}$ where $\Delta_{d,p} :=  (\partial_d^p )^\dagger \partial_d^p + \partial_{d+N-p}^{N-p} (\partial_{d+N-p}^{N-p})^\dagger$ is the $p$-th Mayer Laplacian. Therefore, we will estimate the (normalized) Mayer Betti number $\beta_{d,p}$ by estimating the kernel dimension of $\Delta_{d,p}$. Third, our access to the order-specified complex of interest, denoted by $K$, is given by an oracle $O^K_d$ that can verify the existence of $\sigma_d$ in the complex $K$, i.e., $O^K_d\ket{\sigma_d}\ket{0} = \ket{\sigma_d}\ket{1/0}$ where $1/0$ depends on whether $\sigma_d \in K$ or not. This input model is similar to most previous quantum TDA works \cite{lloyd2016quantum, berry2024analyzing, ubaru2021quantum, mcardle2022streamlined}. In \cite{lloyd2016quantum}, the authors suggested that this oracle can be built from QRAM which holds the pairwise distance of data points. An explicit construction of this oracle, based on Toffoli gates plus the classical database of the edges within the complex, can be found in \cite{berry2024analyzing}. In Section \ref{sec: application}, we will recapitulate the procedure for constructing the oracle introduced in \cite{berry2024analyzing}, which is more near-term  friendly.\\

\begingroup
\color{black}
\paragraph{Preparing a uniform superposition of simplices.}
The following lemma provides the state preparation used in our
algorithm, combining Dicke-state preparation, membership marking,
and fixed-point amplitude amplification.
\begin{lemma}[Uniform superposition of simplices]
\label{lemma: uniformsuperposition}
Let $f=|S_d^K|>0$ and
$|\Phi_d^K\rangle=f^{-1/2}\sum_{\sigma_d\in K}|\sigma_d\rangle$.
There is a coherent preparation unitary $\widetilde U_d$
whose output satisfies
\[
 \left\|\widetilde U_d|0\cdots0\rangle
       -|\Phi_d^K\rangle|0\rangle_{\rm work}\right\|
 \leq\epsilon.
\]
For $d\geq1$, it uses
$\widetilde O(dn\sqrt{\binom n{d+1}/f})$ non-oracle gates
and $\widetilde O(\sqrt{\binom n{d+1}/f})$ membership calls,
with logarithmic dependence on inverse preparation error.
For $d=0$, replace $dn$ by $(d+1)n$.
A constant-factor lower bound on $f/\binom n{d+1}$ suffices
to set the amplification length. If it is not supplied,
an initial approximate-counting stage obtains such a bound
with the same asymptotic expected query cost and a chosen bounded
failure probability.
\end{lemma}
\begin{proof}
We follow the simplex-state preparation construction of
\cite[Section~5]{hayakawa2022quantum}.
Write $\rho=f/\binom n{d+1}$ and let $U_{\mathrm D}$ be the
Dicke-state preparation unitary satisfying
\[
 U_{\mathrm D}|0^n\rangle_{\mathrm{data}}
 =\binom n{d+1}^{-1/2}\sum_{|x|=d+1}|x\rangle_{\mathrm{data}},
\]
where $|x|$ denotes Hamming weight. It uses $O((d+1)n)$ gates
and no ancillas \cite{bartschi2019deterministic}.
Membership marking via $O_d^K$ gives the desired component
amplitude $\sqrt{\rho}$.
Fixed-point amplification \cite[Lemma~25]{gilyen2019quantum}
uses $O(\rho^{-1/2}\log(1/\epsilon))$ calls to
$U_{\mathrm D}$, $U_{\mathrm D}^{\dagger}$, and $O_d^K$,
with $O(n)$ gates for each reflection about the all-zero input.
Resetting the success flag yields the stated target state and costs;
rotation synthesis contributes only polylogarithmic overhead.
If $\rho$ is unknown, constant-relative-accuracy approximate counting
\cite{brassard2002quantum} supplies a constant-factor lower bound
with $\widetilde O(\rho^{-1/2})$ expected queries.
The preparation guarantee is conditional on this classical
preprocessing succeeding; its failure probability is accounted for
separately from $\epsilon$.
\end{proof}
The norm bound includes the success flag and all work registers;
they need not be exactly clean in the implemented state.
\par
\endgroup

\begingroup
\color{black}
\paragraph{Block-encoding the Mayer boundary.}
We encode $\partial_d$ directly from its face--simplex
incidences. Preparing a uniform deleted-vertex state for each
column and a uniform added-vertex state for each row gives
their overlap equal to the boundary coefficient divided by
$\alpha_d=\sqrt{(d+1)(n-d)}$.
Degree and membership tests restrict this construction to $K$.

\begin{lemma}[Block-encoding the Mayer boundary]
\label{lemma: Nboundaryoperator}
For $1\leq d\leq n-1$, there is an
$(\alpha_d,a,\epsilon)$-block-encoding $U$ of $\partial_d$,
where $\alpha_d=\sqrt{(d+1)(n-d)}$ and
$a=n+O(\log n)$. It uses $\widetilde O(n)$ elementary gates
outside the membership oracles and $O(1)$ calls to
$O_d^K$. Logarithms of inverse precision are
suppressed; with exact arbitrary rotations the encoded block
can be exact. The adjoint $U^\dagger$ is an
$(\alpha_d,a,\epsilon)$-block-encoding of $\partial_d^\dagger$.
\end{lemma}
\begin{proof}
See Appendix~\ref{sec: Nboundaryoperator} for the circuits
and their error and resource analysis.
\end{proof}
Since $\alpha_d\leq n$, the non-oracle gate complexity is also
$\widetilde O(n^2/\alpha_d)$.
The error $\epsilon$ is absolute operator error; the corresponding
normalized block error is $\epsilon/\alpha_d$.
\par
\endgroup


\subsection*{Algorithm }
\begingroup
\color{black}
We now combine these ingredients to estimate the normalized Mayer
Betti number. The algorithm constructs a block-encoding of the Mayer
Laplacian, uses QSVT to approximately project onto its kernel, and
estimates the normalized kernel dimension by amplitude estimation.
Use the boundary normalizations $\alpha_k$ from
Lemma~\ref{lemma: Nboundaryoperator}. Let $\gamma_d^p>0$
be a supplied lower bound on the smallest positive eigenvalue
of $\Delta_{d,p}$.
Let $D_k$ denote the Mayer boundary acting on all $(k+1)$-vertex
subsets of the given ordered vertex set (i.e., all $n$-bit strings
of Hamming weight $k+1$), whether or not they belong to $K$.
It acts as zero outside the degree-$k$ input subspace.
Membership in $K$ is imposed separately below.
\begin{method}
\label{algo: quantumMayerestimation}
The quantum algorithm for estimating the normalized Mayer Betti
number $\beta_{d,p}/|S_d^K|$ of an order-specified complex $K$
proceeds as follows. Write $f=|S_d^K|>0$ and $w=\beta_{d,p}/f$.
\begin{enumerate}
 \item Using the boundary-encoding construction in
 Appendix~\ref{sec: Nboundaryoperator}, for each boundary factor
 in the lower and upper terms, construct an
 $(\alpha_k,a,\alpha_k\nu)$-block-encoding of
 $D_k$, with degree selection but no
 membership test. Its adjoint encodes $D_k^\dagger$ with the same
 normalization and error. Here $\nu$ is a normalized
 block-error tolerance.

 \item Using Lemma~\ref{lemma: product}, compose these encodings
 and apply membership selection only
 at the source of each boundary power:
 \[
 \begin{aligned}
 A&=D_{d-p+1}\cdots D_d\Pi_d^K,\\
 B&=D_{d+1}\cdots D_{d+N-p}\Pi_{d+N-p}^K.
 \end{aligned}
 \]
 Here $\Pi_s^K$ selects the degree-$s$ simplices of $K$;
 powers outside the chain degrees are zero and are omitted.
 Face closure gives $A=\partial_d^p$ and
 $B=\partial_{d+N-p}^{N-p}$, padded by zeros outside $K$.

 \item Take the adjoints of the encoding unitaries for $A$ and $B$
 to encode $A^\dagger$ and $B^\dagger$. Using
 Lemma~\ref{lemma: product}, compose the encodings to obtain
 $H_-=A^\dagger A$ and $H_+=BB^\dagger$.
 These are the lower and upper terms of the
 Mayer Laplacian, with normalizations $a_-$ and $a_+$:
 \[
 \begin{aligned}
 H_-&=(\partial_d^p)^\dagger\partial_d^p,&
 a_-&=\prod_{i=0}^{p-1}\alpha_{d-i}^2,\\
 H_+&=\partial_{d+N-p}^{N-p}
       (\partial_{d+N-p}^{N-p})^\dagger,&
 a_+&=\prod_{i=1}^{N-p}\alpha_{d+i}^2 .
 \end{aligned}
 \]
 Their absolute errors are at most $2p\,a_-\nu$
 and $2(N-p)a_+\nu$, respectively. Set $a_-=0$ and
 omit $H_-$ when $d<p$; likewise omit an upper term
 outside the chain degrees.

 \item Using the linear combination of block-encoded operators
 (Lemma~\ref{lemma: linearcombination}), combine the two terms to encode
 $\Delta_{d,p}=H_-+H_+$ with normalization
 $\alpha=a_-+a_+$ and $O(Na)$ ancillas.
 Its normalized block error is at most $h=2N\nu$.
 If the Laplacian is known to vanish, return $w=1$.

 \item Using Lemma~\ref{lemma: qsvt} (QSVT), apply a bounded
 polynomial kernel filter $P$ with
 $|P(x)|\leq1/2$ on $[-1,1]$, approximating $1/2$ at
 zero and zero on the positive spectrum of
 $\Delta_{d,p}/\alpha$.
 Choose its approximation error and the input error $h$
 so that the resulting encoded block $Q$ satisfies
 $\|Q-\Pi_0/2\|\leq\eta$, where $\Pi_0$ projects onto
 $\ker\Delta_{d,p}$ within $C_d^K$ and is zero outside.
 Sandwich the filter by the degree-$d$ membership projector
 $\Pi_d^K$ to enforce this restriction. Its polynomial degree is
 $O((\alpha/\gamma_d^p)\log(2/\eta))$.
 Let $U_Q$ denote the resulting block-encoding unitary,
 with encoded block $Q$. Thus $P$ specifies the polynomial
 approximation, while $U_Q$ implements the block $Q$.

 \item Coherently prepare a state close to
 \begin{equation}
 |\Omega\rangle=\frac1{\sqrt f}
       \sum_{\sigma_d\in K}|\sigma_d\rangle_{\mathrm{data}}
       |\sigma_d\rangle_{\mathrm{ref}}
 \label{eq:mayer-correlated-simplex-state}
 \end{equation}
 using Lemma~\ref{lemma: uniformsuperposition} and $n$
 CNOT gates into an initially zero reference register.
 For $1\leq i\leq n$, the $i$-th CNOT uses the $i$-th qubit
 of the first register as control and the $i$-th qubit of the
 reference register as target.
 Keep the state-preparation ancilla registers without discarding them.

 \item Apply $U_Q$ with its block-encoding ancillas initialized to
 zero. For the ideal preparation $|\Omega\rangle$, this gives
 \[
 \begin{aligned}
 &(U_Q\otimes I_{\mathrm{ref}})
   (|0\rangle_{\mathrm{anc}}|\Omega\rangle)\\
 &\quad=|0\rangle_{\mathrm{anc}}
   (Q\otimes I_{\mathrm{ref}})|\Omega\rangle
   +|\mathrm{rest}\rangle,
 \end{aligned}
 \]
 where $(\langle0|_{\mathrm{anc}}\otimes I)|\mathrm{rest}\rangle=0$.
 Here $\mathrm{anc}$ denotes the block-encoding ancillas;
 the state-preparation ancillas are omitted, and $U_Q$ acts as
 the identity on them.
 Use amplitude estimation on the event $\mathrm{anc}=0$, whose
 probability for this ideal preparation is
 \[
 \begin{aligned}
 p_{\mathrm{succ}}
 &=\|(Q\otimes I_{\mathrm{ref}})|\Omega\rangle\|^2\\
 &\approx\frac{\operatorname{Tr}\Pi_0}{4f}=\frac w4.
 \end{aligned}
 \]
 Multiply the estimated probability by four.
 The approximation-error tolerances, relative-error estimation
 for $w>0$, and treatment of $\beta_{d,p}=0$ are given in
 Appendix~\ref{sec: comlexityanalysis}.
\end{enumerate}
\end{method}

\paragraph{Some comments.}
The filter in Step~5 approximately selects the zero-eigenvalue
subspace of the Mayer Laplacian, with amplitude $1/2$.
The state $|\Omega\rangle$ prepared in Step~6 has reduced density matrix $I_{C_d^K}/f$
on its first register. This is why the success probability
measures a normalized kernel dimension, rather than the
overlap of the kernel with a single uniform superposition.
The filter only approximates the projector; its errors are
included in the error analysis of the final estimation step.

For $\beta_{d,p}>0$, the non-oracle gate cost is bounded by
\begin{equation}
 \begin{aligned}
 \widetilde O\Bigg(
 \Bigg[&
 \frac{n^2\alpha}{\gamma_d^p}
 \left(\sum_{i=0}^{p-1}\frac1{\alpha_{d-i}}
       +\sum_{i=1}^{N-p}\frac1{\alpha_{d+i}}\right)\\
 &+dn\sqrt{\frac{\binom n{d+1}}{f}}
 \Bigg]\frac1\delta\sqrt{\frac f{\beta_{d,p}}}
 \Bigg).
 \end{aligned}
 \label{eq:ordinary-encoding-cost}
\end{equation}
Only nonvanishing terms contribute to the sums and products.
If $\alpha_{\max}$ and $\alpha_{\min}$ are the largest
and smallest normalizations of the included boundary factors,
then
\[
 \alpha\sum_k\frac1{\alpha_k}
 \leq
 N\frac{\alpha_{\max}^{2p}+\alpha_{\max}^{2(N-p)}}
        {\alpha_{\min}}.
\]
Omit the lower summand when $d<p$ and the upper summand
when the upper term vanishes. This implies the looser upper bound
with an $N^2$ factor in Theorem~\ref{thm: mainresult}.
Membership queries are counted separately in
Appendix~\ref{sec: comlexityanalysis};
their gate cost must be added when a concrete oracle is used.
\par
\endgroup

\subsection*{Extension to estimating persistent Mayer Betti numbers}
\textcolor{black}{We now explain how Theorem~\ref{thm: persistent-mainresult}
follows from the preceding algorithm. The additional step is to
construct a block-encoding of the persistent Mayer Laplacian using
a Schur complement. We then apply the same kernel filtering and
amplitude estimation, with the normalization and cost of this new
block-encoding.}
Here, we show how the algorithm introduced above can be extended to the context of persistent Mayer homology. To understand the persistence version, we first need to understand a filtration. A filtration is a sequence of complexes $K_1,K_2,...,K_m$ such that $K_1 \subseteq K_2 \subseteq ... \subseteq K_m$.  This inclusion would induce a sequence on Mayer homology. In essence, persistent Mayer Betti numbers quantify those ``generalized'' holes surviving over filtration. As, for example, let $K_a \subseteq K_b$. Then the persistent  Mayer Betti number $\beta_{d,p}^{a,b}$ (where $d,p$ are defined similarly to above) quantifies those ``generalized'' holes that exist at $K_a$ and persist at $K_b$. Alternatively, it can be thought that $\beta_{d,p}^{a,b}$ quantifies the generators of the homology group $H_{d,p}^{K_a}$, which maintain the generators of $H_{d,p}^{K_b}$. Due to this, persistent Mayer Betti numbers are more relevant to application compared to Mayer Betti numbers, as they are the key ingredients to build the persistence diagram. This is a topological $\&$ geometric summary of the data and can be useful for many other purposes, for example, as input to machine learning models.

A more formal and detailed introduction to persistent Mayer homology shall be left in the Appendix \ref{sec: persistencemayer}. Here, for the purpose of demonstrating the result, we point out that persistent Mayer Betti numbers $\beta_{d,p}^{a,b}$ can be obtained via the so-called \textit{persistent Mayer Laplacian} $\Delta_{d,p}^{a,b}$. This operator is defined as follows:
\begin{align}
    \Delta_{d,p}^{a,b} = \left(\partial_{d}^{p,(a)} \right)^\dagger \partial_{d}^{p,(a)} + \partial_{d+N-p}^{N-p,(a,b)} \left(\partial_{d+N-p}^{N-p,(a,b)}\right)^\dagger
\end{align}
where $\partial_{d}^{p,(a)} := \partial_{d-p+1}\cdots\partial_d: C_d^{K_a} \to C_{d-p}^{K_a}$ and $ \partial_{d+N-p}^{N-p,(a,b)} $ is the restriction of $\partial_{d+N-p}^{N-p,(b)}$ to a subspace of $C_{d+N-p}^{K_b}$ that $\partial_{d+N-p}^{N-p,(b)}$ sends into $C_d^{K_a} \subseteq C_d^{K_b}$. A known result \cite{shen2023persistent} shows that the dimension of the kernel of $\Delta_{d,p}^{a,b}$ equals $\beta_{d,p}^{a,b}$. This is very similar to the normal context above where the dimension of the kernel of the Mayer Laplacian $\Delta_{d,p}$ is equal to the Mayer Betti number $\beta_{d,p}$.

It was also shown in \cite{shen2023persistent} that $\partial_{d+N-p}^{N-p,(a,b)} \left(\partial_{d+N-p}^{N-p,(a,b)}\right)^\dagger $ can be obtained by taking the Schur complement of $\partial_{d+N-p}^{N-p,(b)} \left(\partial_{d+N-p}^{N-p,(b)}\right)^\dagger$, see the Appendix \ref{sec: alternativeform}.


\textcolor{black}{Using the Schur complement and the first three steps of
Algorithm~\ref{algo: quantumMayerestimation}, we obtain the following
block-encoding. We first specify the notation.}
\begingroup
\color{black}
Let $a_-=\prod_{i=0}^{p-1}\alpha_{d-i}^2$ and
$a_+=\prod_{i=1}^{N-p}\alpha_{d+i}^2$, with $a_-=0$ when
$d<p$. Write $\Delta_4$ for the principal block of
$\Gamma=\partial_{d+N-p}^{N-p,(b)}
(\partial_{d+N-p}^{N-p,(b)})^\dagger$ on
$C_d^{K_b}\ominus C_d^{K_a}$.
Here $C_d^{K_b}\ominus C_d^{K_a}$ is the subspace spanned by
the $d$-simplices in $K_b$ that are not in $K_a$.
Let $T_-$ and $T_+$ be the costs of encoding the lower term and
$\Gamma$ at the required precision, including block selection and
controlled operations. As above, $a$ bounds the ancilla count
for each boundary encoding.

\begin{lemma}[Block-encoding the persistent Mayer Laplacian]
\label{lemma: blockencodingpersistentLaplacian}
Suppose $\Delta_4\ne0$ and
\[
 0<\gamma_\Delta\leq
 \min\{1/2,\lambda_{\min}^{+}(\Delta_4/a_+)\}.
\]
Then an $(\alpha',O(Na+1),\epsilon)$-block-encoding of
$\Delta_{d,p}^{a,b}$ can be constructed with
\[
 \alpha'=a_-+a_++\frac{2a_+}{\gamma_\Delta}
\]
and cost $\widetilde O(T_-+T_+/\gamma_\Delta)$.
\end{lemma}
\begin{proof}
See Appendix~\ref{sec: blockencodingpersistentMayerLaplacian}.
\end{proof}
For $d<p$ omit the lower term and set $T_-=0$.
If $\Delta_4$ is known to be zero, omit the inverse and use
$\alpha'=a_-+a_+$ with cost $\widetilde O(T_-+T_+)$.
Terms that vanish by degree are omitted.
The inverse-gap parameter $\gamma_\Delta$ concerns the normalized
block being eliminated. It is separate from
$\gamma_{d,p}^{a,b}$, the gap of the final, unnormalized persistent
Laplacian.
\par
\endgroup

\begingroup
\color{black}
From this block-encoding, we can execute the same procedure from
Step~5 of Algorithm~\ref{algo: quantumMayerestimation}, replacing
the Mayer Laplacian by $\Delta_{d,p}^{a,b}$ and preparing the
correlated two-register state over $S_d^{K_a}$ as in Step~6.
If $\Pi_0^{a,b}$ denotes
the kernel projector of $\Delta_{d,p}^{a,b}$, the ideal
zero block-ancilla probability is
\[
 \frac{1}{4|S_d^{K_a}|}\operatorname{Tr}\Pi_0^{a,b}
 =\frac{\beta_{d,p}^{a,b}}{4|S_d^{K_a}|}.
\]
Multiplication by four recovers the normalized
persistent Mayer Betti number. The error analysis is given in
Appendix~\ref{sec: complexitypersistentBettinumber}.
For $\beta_{d,p}^{a,b}>0$ and relative error $\delta$, the preceding
per-call estimates give the following total complexity.
The factor $\alpha'/\gamma_{d,p}^{a,b}$ comes from the
polynomial degree in Step~5, while
$\delta^{-1}\sqrt{|S_d^{K_a}|/\beta_{d,p}^{a,b}}$ comes from
amplitude estimation in Step~7.
\par
\endgroup
\begingroup
\color{black}
\begin{equation}
 \widetilde O\left[
 \left(
 \frac{\alpha'}{\gamma_{d,p}^{a,b}}
 \left(T_-+\frac{T_+}{\gamma_\Delta}\right)
 +T_{\rm prep}\right)
 \frac1\delta\sqrt{\frac{|S_d^{K_a}|}{\beta_{d,p}^{a,b}}}
 \right].
 \label{eq:persistent-main-cost}
\end{equation}
Here $T_{\rm prep}$ is the coherent simplex-state preparation
cost, including the reflections required for amplitude estimation.
The formula assumes
$\Delta_4\ne0$; otherwise use the zero-block case of the lemma.
Using the boundary-cost estimates above and bounding
$\alpha'\leq3(a_-+a_+)/\gamma_\Delta$ yields the coarser
bound in Theorem~\ref{thm: persistent-mainresult} after applying
the same $\alpha_{\min},\alpha_{\max}$ bounds as in the ordinary case
and counting each membership query as one operation.
The two inverse-gap factors arise from constructing the
pseudoinverse and from the normalization $\alpha'$, respectively.
When $d<p$, omitting the lower term gives the second case of the
theorem. The derivation and precision choices appear in
Appendices~\ref{sec: blockencodingpersistentMayerLaplacian}
and \ref{sec: complexitypersistentBettinumber}.
\par
\endgroup

\section{ PERFORMANCE OF EXISTING DEQUANTIZATION ALGORITHM}
\label{sec: dequantization}
\textcolor{black}{We examine possible extensions of the classical algorithms
of Refs.~\cite{apers2023simple, berry2024analyzing} to Mayer homology,
distinguishing efficient matrix access from the additional sampling
guarantees needed for a polynomial-time algorithm.
Heat-trace estimation via Monte Carlo~\cite{berry2024analyzing}
can accommodate inverse-polynomial additive
error if its sampling variance, inverse mixing gap, and other cost
parameters remain polynomially bounded; whether the required
conditions hold for the Mayer families of interest remains unresolved.
For polynomial trace estimation~\cite{apers2023simple}, we analyze
a Mayer extension based on an explicit factor-path estimator and
establish polynomial-time
estimation at constant additive error and constant normalized gap,
under the access and normalization assumptions stated below.
This latter analysis does not establish a polynomial-time guarantee
at inverse-polynomial error; it is not a limitation proved for
the original algorithm of Ref.~\cite{apers2023simple}.}

\subsection*{Heat-trace estimation via Monte Carlo~\cite{berry2024analyzing}}
\begingroup\color{black}
Berry et al.~\cite[Appendix G]{berry2024analyzing} use path-integral
Monte Carlo to estimate a heat trace associated with the ordinary
combinatorial Laplacian.\footnote{In their notation, the operator
being decomposed is $\widetilde{B_G^2}=B_G^2+\gamma_{\min}(I-P)$,
where $B_G$ is their Dirac operator, $P$ projects onto valid simplex
states, and $\gamma_{\min}$ is the spectral-gap parameter.
The added term shifts the spectrum outside the valid subspace
without changing the operator within it.}
We consider the corresponding approach for the Mayer Laplacian
$\Delta_{d,p}$. The cost of the classical method depends on the
one-sparse decomposition, the estimator variance $\sigma^2$, the
condition number $\kappa$, and the sampling-chain gap $\gamma_M$.
These are separate from the dimension of the chain space.

For reference, their ordinary-homology algorithm uses
\cite[Appendix G.3]{berry2024analyzing}
\begin{equation}
\begin{aligned}
 \widetilde O\Bigg(\frac{\sigma^2}{\epsilon^2}
 \Bigg[&\frac{|E|M}{f}
 +\frac{D^4\kappa^3}{\gamma_M\epsilon}
 \left(\frac{\log M}{D^2}+\frac{\kappa^3}{\epsilon}\right)
 \Bigg]\Bigg)
\end{aligned}
\label{eq:berry-classical-cost}
\end{equation}
arithmetic operations for additive error $\epsilon$.
Here $M=\binom n{d+1}$ and $f=|S_d^K|$ for the clique complex
of the input graph, whose edge set is $E$; their notation is
$d_{k-1}=M$ and $|\mathrm{Cl}_k(G)|=f$, with $k=d+1$.
The parameters $D$, $\sigma^2$, and $\gamma_M$ are the number of
one-sparse terms, the sample variance, and the sampling-chain gap,
respectively; $\kappa=\gamma_{\max}/\gamma_{\min}$ is their
spectral ratio excluding the kernel.
Thus polynomial bounds on $M/f$, $D$, $\sigma^2$, $\kappa$,
$\gamma_M^{-1}$, and $\epsilon^{-1}$ suffice for polynomial cost.
Equation~\eqref{eq:berry-classical-cost} is the guarantee for their
original construction, not a proved runtime for its Mayer extension.

\paragraph{A Mayer counterpart.}
On $C_d^K$, let $f_d=|S_d^K|>0$, and let $g>0$ lower-bound the
positive spectrum of $\Delta_{d,p}$. Spectral decomposition gives
\begin{equation}
  0\leq
  \frac{\operatorname{Tr}_{C_d^K}e^{-t\Delta_{d,p}}}{f_d}
  -\frac{\beta_{d,p}}{f_d}
  \leq e^{-tg}.
  \label{eq:mayer-heat-trace}
\end{equation}
Thus $t\geq g^{-1}\log(1/\epsilon)$ controls the filtering error.
If implemented on a larger subset register, one may extend
$\Delta_{d,p}$ by zero and use
$\Delta_{d,p}+g(I-\Pi_d^K)$, where $\Pi_d^K$ projects onto $C_d^K$.
The trace in Eq.~\ref{eq:mayer-heat-trace} is still restricted to
valid degree-$d$ simplices: sampling uniformly from these gives
normalization $f_d$, not the dimension of the padded register.
Rejection sampling from all $(d+1)$-vertex subsets costs
$\binom{n}{d+1}/f_d$ membership tests on average.

\paragraph{Sparsity and decomposition.}
Write
\begin{equation}
  \Delta_{d,p}=L_-+L_+,\qquad
  \begin{aligned}
    L_-&=(\partial_d^p)^\dagger\partial_d^p,\\
    L_+&=\partial_{d+N-p}^{N-p}
              (\partial_{d+N-p}^{N-p})^\dagger.
  \end{aligned}
  \label{eq: laplacian_expanded}
\end{equation}
Let $s_-$ and $s_+$ bound the number of nonzero entries per row
of these two terms. A nonzero entry of $L_-$ requires a shared
$(d-p)$-face; a nonzero entry of $L_+$ requires a shared
$(d+N-p)$-coface. These are necessary, not sufficient, conditions,
since phase contributions can cancel.
For $L_-$, a fixed $d$-simplex has $\binom{d+1}{p}$ faces
of dimension $d-p$. Each such face is contained in at most
$\binom{n-d+p-1}{p}$ candidate $d$-simplices, obtained by adding
$p$ vertices. For $L_+$, a fixed $d$-simplex has at most
$\binom{n-d-1}{N-p}$ cofaces of dimension $d+N-p$,
each containing $\binom{d+N-p+1}{N-p}$ faces of dimension $d$.
Thus the same counting argument for the two terms gives
\begin{align}
 s_-&\leq\binom{d+1}{p}\binom{n-d+p-1}{p},\\
 s_+&\leq\binom{n-d-1}{N-p}\binom{d+N-p+1}{N-p}.
\end{align}
Terms outside the chain complex are omitted. These bounds can
overcount a pair of simplices and remain valid upper bounds.
Consequently the row sparsity $s$ of $\Delta_{d,p}$ satisfies
\begin{equation}
  s\leq s_-+s_+
   =O\bigl(n^{2p}+n^{2(N-p)}\bigr)
  \quad\text{for fixed }N.
\end{equation}
The decomposition size $D$ is not $s$: the graph-coloring
construction discussed in Ref.~\cite[Appendix G]{berry2024analyzing}
uses $O(s^2)$ one-sparse terms.
In our setting, candidate entries and their values can be generated
by enumerating the bounded-length face and coface moves, with
membership queries; this has polynomial cost for fixed $N$.
Neither this access nor the sparsity bound alone controls the
variance or mixing of a resulting path sampler.

\paragraph{Classical sampling requirements.}
The potentially exponential variance upper bound in
Ref.~\cite[Appendix G.3]{berry2024analyzing} does not by itself
establish classical inefficiency. Nevertheless, polynomial sparsity
alone is insufficient to guarantee efficient classical estimation:
a heat-trace Monte Carlo extension also requires control of sampling variance
and mixing at the requested
precision.\footnote{For comparison with relative error $\delta$,
the additive estimation error should be $O(\delta w)$, where
$w=\beta_{d,p}/f_d>0$.}
Establishing such control under the same promises as our quantum
algorithm remains an unresolved challenge.
\par\endgroup

\subsection*{Polynomial trace estimation~\cite{apers2023simple}}
\begingroup\color{black}
Apers et al.~\cite[Theorem 1.1]{apers2023simple} assume uniform
face sampling and membership access, not an explicitly stored
Laplacian. For ordinary homology their algorithm estimates
normalized Betti numbers to additive error $\epsilon$ in time
\begin{equation}
 n^{O(\gamma^{-1/2}\log(1/\epsilon))},
 \label{eq:apers-classical-cost}
\end{equation}
where $\gamma$ is a
normalized spectral-gap lower bound. Their clique-complex bound
also permits efficient regimes beyond constant gap and precision.
The square-root dependence on $\gamma^{-1}$ uses a Chebyshev
improvement; it should not be confused with the power estimator below.

\paragraph{Filtering with the normalized gap.}
For the Mayer Laplacian on $C_d^K$, choose
$\hat{\lambda}\geq\|\Delta_{d,p}\|>0$, let $g>0$ lower-bound its
positive spectrum, and put
\begin{equation}
 H_{d,p}=I-\Delta_{d,p}/\hat{\lambda},
 \qquad \gamma=g/\hat{\lambda}\in(0,1].
\end{equation}
Then, with $f_d=|S_d^K|>0$, spectral decomposition gives
\begin{equation}
 0\leq\frac{\operatorname{Tr}(H_{d,p}^z)}{f_d}
       -\frac{\beta_{d,p}}{f_d}
 \leq(1-\gamma)^z\leq e^{-\gamma z}.
 \label{eq:mayer-power-filter}
\end{equation}
Hence $z=\lceil\gamma^{-1}\log(2/\epsilon)\rceil$ suffices for
filtering error at most $\epsilon/2$, for $0<\epsilon<1$.
A constant unnormalized gap $g$ does not imply constant $z$
if $\hat{\lambda}$ grows with $n$. If $\Delta_{d,p}=0$ is known,
the normalized Betti number is one and no filtering is needed.

\paragraph{An explicit factor-path estimator.}
We give a direct Mayer extension of the power-estimation idea,
without importing the ordinary-homology variance guarantees.
Use the factors $L_-,L_+$ defined above and write
\begin{equation}
 \begin{aligned}
 H_{d,p}&=\sum_{\nu\in\{0,-,+\}}c_\nu A_\nu,\\
 (A_0,A_-,A_+)&=(I,L_-,L_+),\\
 (c_0,c_-,c_+)&=(1,-\hat{\lambda}^{-1},-\hat{\lambda}^{-1}).
 \end{aligned}
\end{equation}
For $A_\nu=F_{\nu,h_\nu}\cdots F_{\nu,1}$, take the elementary
boundary or adjoint factors, in application order, with
$h_-=2p$ and $h_+=2(N-p)$. The identity uses an empty product.
Given a current simplex $x$, an elementary move with factor $F$
samples $y$ with probability
$|F_{yx}|/\|F_{\cdot,x}\|_1$ and contributes the weight
\begin{equation}
 \|F_{\cdot,x}\|_1\,\frac{F_{yx}}{|F_{yx}|}.
 \label{eq:mayer-factor-weight}
\end{equation}
A zero column terminates the sample with weight zero. Adjoint
moves use the conjugated boundary phases.

Each factor column has at most $n$ entries of modulus one.
Thus we may use the path-weight bounds
$B_0=1$, $B_-=n^{2p}$, and $B_+=n^{2(N-p)}$.
Omit a term if its degree lies outside the chain complex. Define
\begin{equation}
 B_H=\sum_\nu |c_\nu|B_\nu,\qquad
 \pi_\nu=\frac{|c_\nu|B_\nu}{B_H}.
 \label{eq:mayer-term-probabilities}
\end{equation}
Choosing term $\nu$ contributes the additional weight
$c_\nu/\pi_\nu$; this compensates for its selection probability.

\begin{algorithm}[H]
\color{black}
\caption{Factor-path estimate of $\operatorname{Tr}(H_{d,p}^z)/f_d$}
\KwIn{Uniform sampling from $S_d^K$, membership access at the
factor degrees, $z$, and sample count $n_{\rm samp}$.}
\KwOut{A real trace estimate $\widehat T_z$.}
\For{$t=1,\ldots,n_{\rm samp}$}{
 Sample $\sigma_0\in S_d^K$ uniformly; set $x=\sigma_0$, $W=1$\;
 \For{$s=1,\ldots,z$}{
  Choose $\nu$ with probability $\pi_\nu$;
  set $W\gets Wc_\nu/\pi_\nu$\;
  Apply the elementary moves for $A_\nu$, updating $x$ and
  multiplying $W$ by Eq.~\ref{eq:mayer-factor-weight};
  if a column is zero, set $W=0$ and terminate this sample\;
 }
 Set $Z_t=\operatorname{Re}W$ if $x=\sigma_0$, and $Z_t=0$ otherwise\;
}
Return $\widehat T_z=n_{\rm samp}^{-1}\sum_t Z_t$\;
\end{algorithm}

Summing over each sampled move cancels its probability against
the corresponding weight. Summing over $\nu$ similarly recovers
$\sum_\nu c_\nu A_\nu=H_{d,p}$. Closing the path and averaging
the starting simplex therefore gives
\begin{equation}
 \mathbb{E}Z_t=\operatorname{Tr}(H_{d,p}^z)/f_d.
\end{equation}
Taking the real part is valid because $H_{d,p}$ is Hermitian.
For the probabilities in Eq.~\ref{eq:mayer-term-probabilities},
every full step has absolute weight at most $B_H$, so
\begin{equation}
 |Z_t|\leq B_H^z,\qquad
 \mathbb{E}Z_t^2\leq B_H^{2z}.
\end{equation}
With independent samples,
\begin{equation}
 n_{\rm samp}=O(B_H^{2z}/\epsilon^2)
 \label{eq:mayer-classical-samples}
\end{equation}
samples give sampling error at most $\epsilon/2$ with constant success
probability. Increasing this probability only costs the usual
logarithmic factor. Face moves can be enumerated directly;
coface moves can be enumerated by testing at most $n$ candidate
vertex additions against the membership oracle.
Each sample requires one uniform simplex sample and at most
$2Nz$ such moves, each with polynomial cost under the stated
access assumptions.

For fixed $N$, polynomial-cost access, and polynomially bounded
$\hat{\lambda}^{-1}$, this estimator runs in polynomial time at
constant normalized gap and constant additive error.
Extending this guarantee to smaller gaps or inverse-polynomial
error requires sharper control of the sampling variance.\footnote{The
upper bound above is not a lower bound on the actual sampling cost,
and does not establish a Mayer analogue of the sharper
clique-complex guarantees in Ref.~\cite{apers2023simple}.}
\par\endgroup

\section{PROSPECT OF APPLICATION}
\label{sec: application}
In this section, we first give an example regarding protein-ligand binding affinity prediction to show how persistent Mayer homology has been applied to real-world context. Through this example, we would understand the general procedure on how complex and sequence of nested complexes are typically built from initial data. Then we describe a general pipeline for applying the quantum algorithm \ref{algo: quantumMayerestimation} and its persistence variant in practice. \textcolor{black}{Following this, we give conditional resource estimates on the cost of implementing the quantum algorithm in practice and discuss the prospects for practical quantum advantage.} Then we discuss state-of-the-art progress on the application of Mayer homology, including its limitation, and propose some future direction for quantum-accelerated Mayer homology. Through the limitation of current classical methods as well as the proposal for using quantum algorithms, we will see how quantum computers can potentially advance scientific discovery and open up new avenues for many scientific and engineering areas.

\subsection*{Example: Predicting protein-ligand binding affinity}
Protein-ligand binding affinity is crucial in drug design, which has been significantly advanced by machine learning (ML) and artificial intelligence (AI). In order to apply ML and AI method, an important ingredient is to have a proper feature vector that can well-represent the object of interest. As such, obtaining the proper feature representation vector, e.g., of molecules, is of fundamental importance. In \cite{feng2025mayer}, the authors have proposed to use persistent Mayer homology to construct the multiscale topological representation of protein-ligand complex. This can then be integrated with machine learning models to predict the binding affinity with an outstanding accuracy compared to state-of-the-art methods, following the analysis pipeline in Figure \ref{fig: proteinligand} \cite{feng2025mayer}.
\begin{figure*}
    \centering
    \includegraphics[width=\linewidth]{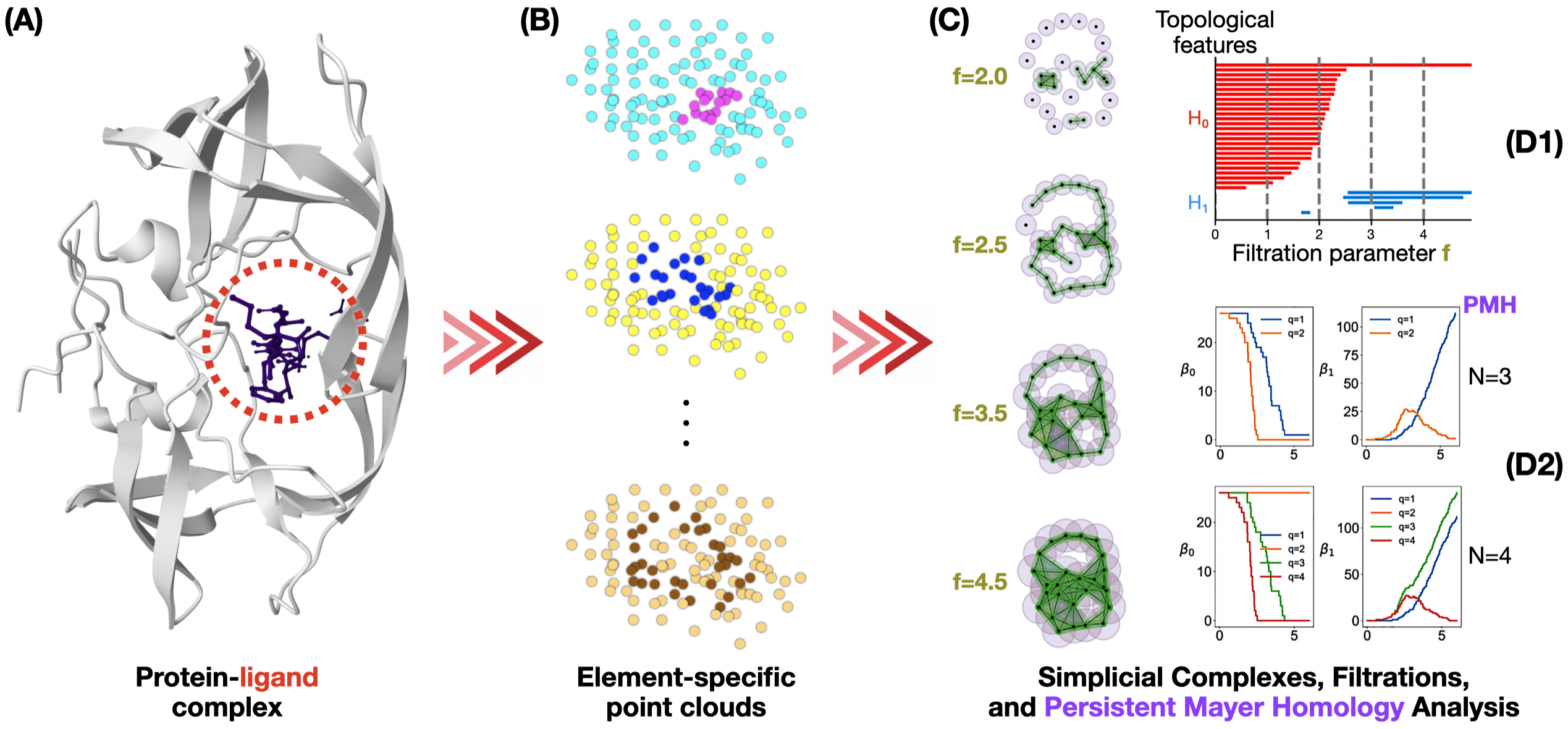}
    \caption{\textbf{Mayer Homology Analysis Pipeline for Protein–Ligand Complexes.} \textbf{(A)} A protein-ligand complex (RCSB PDB 1A94, rendered with Mol*). \textbf{(B)} Point clouds constructed directly from the physical positions of the atoms that make up the protein and ligand. \textbf{(C)} Simplicial complexes built from either Euclidean distances or a correlation matrix at different values of the filtration parameter $f$. \textbf{(D1)} The persistent barcode plot. \textbf{(D2)} Characterization of persistent Mayer homology (PMH) at homological degrees $N=3$ and $N=4$. Images adapted from \cite{feng2025mayer}. }
    \label{fig: proteinligand}
\end{figure*}

The way in which the simplicial complex is built from the protein-ligand complex is as follows. First, treat the atoms within the protein and ligand as points, using 3D atomic coordinates. Second, define the distance between a pair of points, which then controls the connectivity. In \cite{feng2025mayer}, there are two choice of distance. The first one is $d(i,j) = r_{ij}$ where $r_{ij}$ is the Euclidean distance between the $i$-th atom and $j$-th atom. The second one is 
$$d(i,j) = 1 - \exp\left(  - \left(\frac{r_{ij}}{ \tau (r_i + r_j)} \right)^k \right)$$
where $k, \tau$ are some constants. From this distance, the simplicial complex can be constructed by first choosing a filtration parameter, or radius $\epsilon$, then building a radius ball $\frac{\epsilon}{2}$ around each data point, then connecting two points if their distance satisfies $d(i,j) \leq \epsilon$. This type of complex is called Vietoris-Rips (VP) complex. Another approach for constructing simplicial complex, as also used in \cite{feng2025mayer}, is via the so-called alpha complex, which is based on the Voronoi diagram and Delauney triangulation, instead of ball of certain radius as in VP complex.

A filtration can be obtained by varying the filtration radius $\epsilon$. Consider a sequence $\epsilon_1 < \epsilon_2 < ... < \epsilon_m$, then we have a sequence of nested complexes $K_1 \subseteq K_2 \subseteq ... \subseteq K_m$. From this sequence, the Mayer Betti numbers and persistent ones are computed for each $K_1,K_2,...,K_m$. Then they are used to build the so-called persistent barcodes, or the Mayer Betti curves. Either of these data, barcodes or Mayer Betti curves, can be treated as a topological $\&$ geometrical summary of data. 
We illustrate filtrations and how the sequence of nested complexes are built in Figure \ref{fig: proteinligand}C. Figure \ref{fig: proteinligand}D1 shows an instance of persistent barcodes, where the $x$-axis represents the filter parameter/radius, and the $y$-axis represents the homology groups. Essentially, a barcode contains information about some homology generator, when it appears and when it dies, so the length of any bar is the lifespan of such the generator. Figure \ref{fig: proteinligand}D2 shows what the Betti curve looks like. Again the $x$-axis indicates the filtration radius, and $y$-axis are the value of corresponding Mayer Betti numbers. Essentially, this plot shows how the Mayer Betti numbers behave as a function of filtration radius. 

Provided that the persistent barcodes have been constructed, we now describe a few strategies to obtain the multiscale topological $\&$ geometrical representation of the dataset, as indicated in \cite{feng2025mayer}:
\begin{itemize}
    \item Length sum of all bars corresponding to 0-th Betti numbers. 
    \item Length sum and birth sum of all Mayer Betti numbers. 
    \item Length sum of all Mayer Betti bars with birth values within a certain range of interval.
    \item Mayer Betti numbers $\beta_{d,q}$ over filtration parameter/radius range from 0 to some number, with a certain step size. 
\end{itemize}
These features are obtained from persistent barcodes, which is a topological and geometrical summary of the given dataset. They can be concatenated into a large vector, serving as a feature vector within a regression or machine learning model to subsequent usage, for example, predicting binding affinity or generally any property of interest. 

We note that, although the above example is specific for protein-ligand affinity, the procedure for applying persistent Mayer homology (in fact, persistent homology as well) is the same. As long as a point cloud is defined, one can build the VR complex (or other kind of complex like $\alpha$-complex or $\Delta$-complex) based on the distance between a pair of data points (or the Delauney triangulation). Then one can choose a sequence of filtration radiuses and build a sequence of nested complexes. Mayer Betti numbers and persistent ones are then computed from these complexes, followed by a construction of persistent barcodes and eventually the vectorization of the initial dataset is obtained via the strategies mentioned above.

\subsection*{Pipeline for quantum algorithm in practice}
The main ingredient for our quantum algorithm to be able to execute in practice is the oracle $O_d^K$ (for all $d$). The original work \cite{lloyd2016quantum} builds the oracle based on QRAM access. This approach is not feasible for near-term implementation due to the requirement of QRAM. Fortunately, a more friendly approach has been provided in \cite{berry2024analyzing}, for which we recapitulate as follows. 

\paragraph{Practical approach for constructing the oracle $O_d^K$.}
\begingroup\color{black}
Following the edge-list approach of Ref.~\cite{berry2024analyzing},
first, we need a classical database containing the edges $E$ of a
simple graph whose flag complex is $K$. Let $s_i$ indicate whether
vertex $i$ belongs to the input subset $S$.

Second, introduce one ancilla qubit for each edge. For each
$\{i,j\}\in E$, use a Toffoli gate controlled by input qubits
$i$ and $j$ to store $s_i s_j$ in its edge ancilla. For example,
the edge $\{1,3\}$ uses the first and third input qubits as controls.
Reversibly sum the edge bits into a counter to obtain
\begin{equation}
 e(S)=\sum_{\{i,j\}\in E}s_i s_j.
\end{equation}
Also compute the subset size $w=\sum_i s_i$ in a separate counter.
For unrestricted input subsets, the edge counter requires
$\lceil\log_2(|E|+1)\rceil$ bits and the size counter requires
$\lceil\log_2(n+1)\rceil$ bits. If $E$ is empty, the edge count
is the known constant zero and needs no counter.

Finally, check both the subset size and edge count. The degree-$d$
membership predicate is
\begin{equation}
 S\in S_d^K
 \quad\Longleftrightarrow\quad
 \begin{cases}
 w=d+1,\\
 e(S)=\binom{d+1}{2}=\dfrac{d(d+1)}{2}.
 \end{cases}
\end{equation}
Compute the conjunction of the two comparisons, XOR it into the
oracle output bit, and uncompute all work registers.
For $d=0$, the test correctly accepts single vertices with no induced
edges. This gives a polynomial-size membership circuit for flag
(clique) complexes, not for arbitrary simplicial complexes.
With the classical edge list fixed in the circuit, reversible accumulation
of the $|E|$ edge bits and $n$ vertex bits into binary counters gives
the elementary-gate upper bound
\begin{equation}
 C_K=O\bigl((|E|+n)\log(n+1)\bigr)
     =\widetilde O(|E|+n).
 \label{eq:flag-membership-cost}
\end{equation}
This includes the comparisons, output operation, and uncomputation;
controlled use and inverse have the same asymptotic bound. The bound
is uniform over the required degrees and excludes the classical cost
of constructing the edge list.
\par\endgroup
\begingroup
\color{black}
For fixed vertex labels, the membership oracle tests whether a vertex
subset belongs to $K$ and is independent of the chosen total order.
The order dependence of Mayer homology instead enters through the phase
coefficients of the boundary operator. If the qubit positions are permuted
 to represent a different order, the membership circuit and the boundary
phase circuit must both be relabeled consistently.
\par
\endgroup

Below is the pipeline for executing our quantum algorithm in practice. 
\begin{pipeline}[Pipeline for practical quantum algorithm for persistent Mayer homology]
\label{algo: practicalpipeline}
Given the procedure for constructing the oracle as above, the pipeline for the quantum algorithm for persistent Mayer homology in practice is as follows:
\begin{itemize}
    \item Choose a sequence of filtration radius $\epsilon_1 < \epsilon_2 < ... < \epsilon_m$.
    \item For each $\epsilon_i$ ($i =1,2,...,m$), classically build the graph-description of the complex $K_i$.
    \item For each $i=1,2, ..., m$, use the procedure above to construct the corresponding oracle $O_d^{K_i}$ for $d=1,2,..,n$.
    \item Run the quantum algorithm Algo.~\ref{algo: quantumMayerestimation} (and its persistent variants) to estimate the Mayer Betti numbers at the filtration radius $\epsilon_1,\epsilon_2, ..., \epsilon_m$, and also the persistent Mayer Betti numbers along the nested sequence $K_1 \subseteq K_2 \subseteq ... \subseteq K_m$. That being said, we would obtain $\beta_{d,p}$ (for $d=1,2,...,n$ and $p = 1,2,..., N-1$) and $\beta_{d,p}^{K_i, K_j}$ (for $i,j \in [1,2,...,m]$ and $i < j$).
    \item \textcolor{black}{Collect the estimated normalized Mayer Betti numbers and persistent ranks at the selected filtration levels and pairs into a feature vector.}
    \item \textcolor{black}{If barcode-based features are desired, reconstruct the barcode and form the feature vector, accounting for the estimation errors.\footnote{Exact persistent ranks determine barcode multiplicities by M\"obius inversion \cite{kim2021generalized}. Reconstruction from estimated normalized ranks additionally requires accounting for normalization factors and error propagation through reconstruction and vectorization.}}
\end{itemize}
\end{pipeline}
\begingroup
\color{black}
The resulting feature vectors can be used as inputs to regression or
classification models. When the algorithmic promises hold, access to
additional homological degrees may provide a richer representation
of the data, whose predictive value can be assessed empirically.
\par
\endgroup

\begingroup\color{black}
\subsection*{Regime for practical quantum advantage}
\label{sec:resource-estimates}
The resource estimates below use the alternative LCU construction in
Sections~\ref{sec:boundary-alternative-lcu}
and~\ref{sec:boundary-alternative-amplification}, rather than the direct
incidence-state construction used in Lemma~\ref{lemma: Nboundaryoperator}
and Algorithm~\ref{algo: quantumMayerestimation}.
In the previous subsection, we have summarized the procedure for building
$O_d^K$, which is drawn from \cite{berry2024analyzing}.
In Appendix~\ref{sec: realisticresourceestimation}, we give resource
estimates based on this alternative construction, including the cost
of $O_d^K$, which results in:
\begin{widetext}
    \begin{align}
\begin{split}
      \Big(  \frac{n}{\gamma_{d,p} }    \left( 6|E| + n + 2\log ( d^{2p} ) +  2  \log ( (d+N-p)^{2(N-p)} ) \right) N^2 \frac{ (\alpha_{\max})^{2p} + (\alpha_{\max})^{2(N-p)}  }{\alpha_{\min}} + & 
      \\
      2    \sqrt{\frac{\binom{n}{d+1}}{|S_d^K|}} ( 3|E| + 2 \log d) \Big) \frac{1}{\delta} \sqrt{\frac{|S_d^K|}{ |\beta_{d,p}|}} & \ .
\end{split}
\end{align}
\end{widetext}
With this complexity, we now try to estimate what should be the reasonably minimal threshold for the number of Toffoli gates required for our quantum algorithm to be useful in practice and surpass the classical capabilities. \\

\noindent
\textbf{State-of-the-art performance of classical computers.} We point out that the current upper limit on the performance of classical algorithm in Mayer homology is $n = 1\text{,}000$ vertices, $|E| = 10\text{,}000$ edges, for $N=3$ (for $N=5$ the upper limit of vertices are $n=800$) \cite{feng2025mayer}. With this number of vertices, the time recorded for computing 1-st Mayer Betti numbers (for a certain value of $p$) for $N =3$ is 20 hours, which is considered to be practically inefficient. We can treat this 20 hours time as the minimum time required to diagonalize the corresponding Mayer Laplacian and compute the 1-st orders, then for all other orders it would take at least this much time, if the number of simplices of higher orders are greater or equal than the number of 1-st simplices (which is always true in the dense complex regime). Since 20 hours is the lower bound computing the first Mayer Betti number at an order, one can infer that it will take $n \times 20$ hours to estimate all the Mayer Betti numbers after enumerating the whole chain space. For a complex with $1\text{,}000$ vertices, it will be $20\text{,}000$ hours, which is more than 2 years. 

\noindent
\textbf{Minimal quantum resource to surpass classical capabilities.} Since classical algorithms already hit a computational wall at $n=1\text{,}000$ vertices and $|E| = 10\text{,}000$ edges, at first it seems reasonable to evaluate the quantum resources at this specific classical limit, which can reveal the absolute minimum hardware threshold required for a quantum computer to surpass current classical capabilities. However, we point out one subtle point that our quantum algorithm achieves a provable asymptotic advantage in the dense complex regime. For a graph with $n =1\text{,}000$ vertices, the maximum number of edges is $\frac{1}{2} n (n-1) = 499\text{,}500$. The number of actual edges above is $10\text{,}000$, which is only $2\%$ of total edges, implying that this complex is rather sparse. This is the regime in which our quantum algorithm does not achieve the best performance, so any estimation at this limit would lead to a considerably larger value compared to the minimal one. Furthermore, we point out that the above classical limit is at $d=1$, which is low order and is considered to be the efficient regime for classical computers (in theory). The high orders regime, e.g., $d \longrightarrow \frac{n}{2}$, is where classical algorithms fail, and it is where quantum computers prove most useful. In order to ``reverse'' the dense regime, we consider the following alternative. The computational limit of classical computers is at $10\text{,}000$ edges for $n=1\text{,}000 $ vertices and for computing $\beta_{1,p}$, then for a complex with $n=18 $ vertices, at $d= n/3 = 6$, in the dense complex regime, the maximum number of $d$-simplices can be $\binom{18}{6} = 31824$, which is considerably larger than $10\text{,}000$. Beside, we point out that the Mayer boundary operator $\partial_{d,p}$ for higher $d$ is denser than $\partial_{1,p}$, thus computing $\beta_{d,p}$ in this case $d=n/3$ is practically inefficient for classical computers, as it would possibly take much longer time than $20$ hours. The reason we choose $d= \frac{n}{3}$ in this case is to make it consistent with the Theorem \ref{thm: mainresult}, where we are able to prove that at this order, the ratio $\frac{\beta_{n/3,p}}{|S_{n/3}^K|}$ is bounded by $\Omega(1) \ \forall p$ and thus implying the exponentially large value of $\beta_{n/3,p} $ when $|S_{n/3}^K|$ approaches $\binom{n}{n/3}$. For any complex with the number of vertices beyond $n=18$, the calculation of $\beta_{d,p}$ for $d$ approaching $\frac{n}{3}$ is beyond the reasonable reach of classical computing. 

Now we estimate the minimal quantum resource at this limit. For $n = 18$, the maximum number of edges is $|E| = \frac{1}{2}n(n-1) = 153$. For $d =\frac{n}{3} = 6, N= 3, p =1$, so $\alpha_{\max} \approx 9.3, \alpha_{\min} \approx 9.1$. In the most optimistic scenario, the factor $\gamma_{d,p} = \mathcal{O}(1)$, and $\frac{\binom{n}{d+1}}{|S_d^K|}, \sqrt{\frac{|\beta_{d,p}|}{|S_d^K|}} \in \mathcal{O}(1)$, we find that the minimal number of Toffoli gates is approximately $63\times 10^6$, and the number of qubits required is approximately $n + |E| = 153$, which is of the order $\sim 10^2$. This bar still holds even when $ \frac{\binom{n}{d+1}}{|S_d^K|} = \mathcal{O}(n^2)$, which corresponds to the less dense regime. Therefore, as the best-case theoretical lower bound, a quantum computer with hundreds of qubits and millions of Toffoli gates can already reach the current classical limit. Any number beyond this threshold can make quantum computers surpass the classical limit, potentially achieving what the current classical computer has not done so far. For example, with more than $153$ qubits and $63\times 10^6$ gates, we can estimate Mayer Betti numbers and the persistent ones at all orders $d=1,2,..., n$ (for all $p$) for a complex with  $n \geq 18$ vertices. This allows us to build a more completed persistent barcode that includes all orders $d=1,2,...,\frac{n}{3}, \frac{2n}{3}+1, ..., n$. Since for $n=18$ and $d=6$, calculating $\beta_{d,p} \equiv \beta_{8,p}$ already exceeds the limit (in the dense case), it implies that calculating $\{\beta_{d,p}\}$ for a complex with $n > 18$ vertices would be practical inefficient. This further means that classical computers cannot build a persistent barcode that includes high-orders, even for a complex with hundreds of vertices. In fact, in \cite{feng2025mayer}, the authors only calculate $\beta_{d,p}$ for $d=0,1$. Therefore, the persistent barcode only includes $H_{d,p}$ for $d=0,1$, lacking those $H_{d,p}$ for $d > 1$. Lastly, we would like to note that the estimation given above is the hardware source, not end-to-end time complexity. A quantum algorithm with the number of qubits and gates as mentioned would eventually need to repeat the circuit multiple times as required by the amplitude estimation algorithm and also take into account of the relative error $\delta$.
\par\endgroup

\subsection*{Proposal for future application of quantum algorithm in persistent Mayer homology }

\begin{figure*}[htbp]
\includegraphics[width=0.99\linewidth]{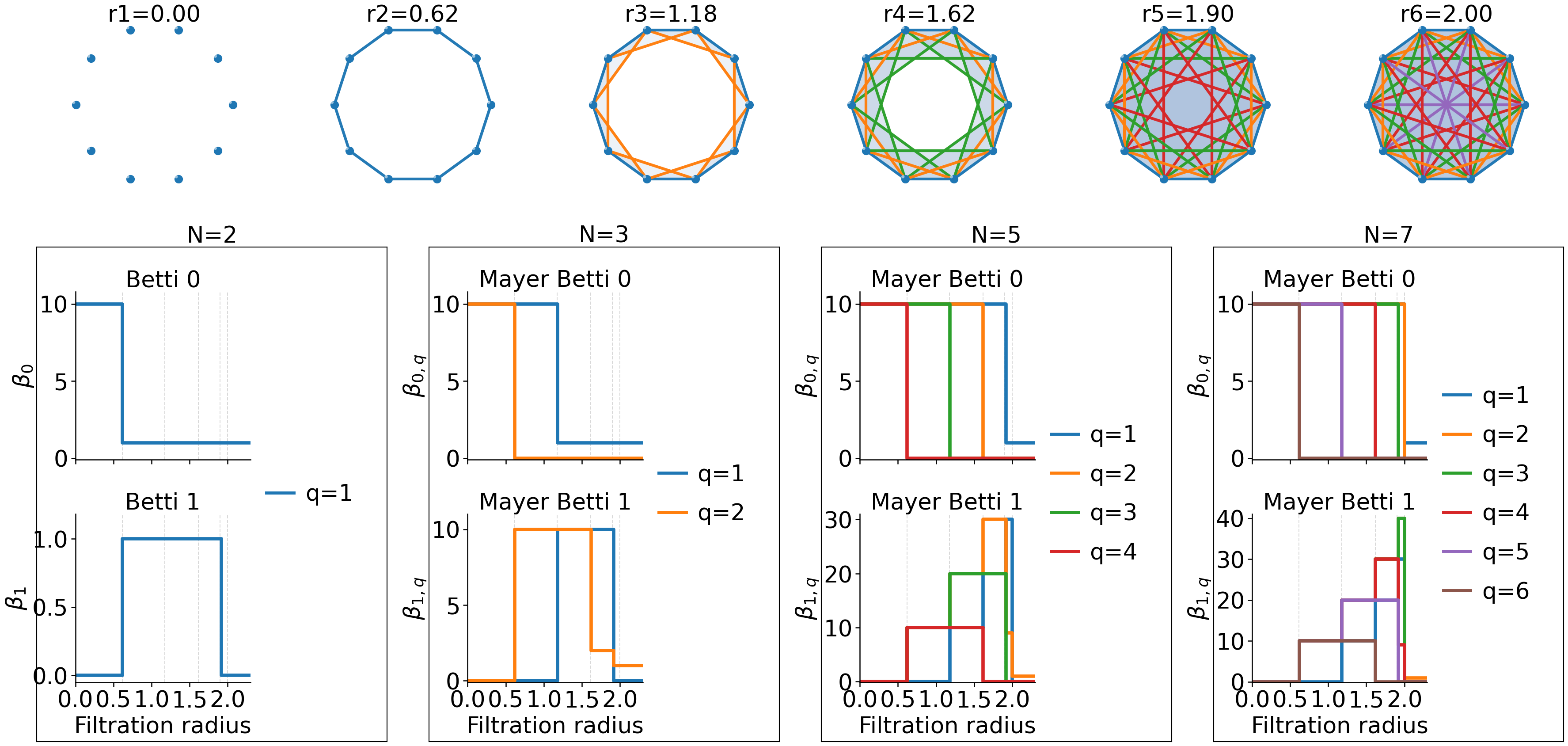}
    \caption{Comparison persistent homology ($N=2$) and persistent Mayer homology ($N=3,5,7$) Betti barcodes. For the Betti-$0$ barcodes, classical persistent homology ($N=2$) is insensitive to all geometric connectivity changes except the one occurring at $r_2$. In contrast, persistent Mayer homology with $N=7$ captures every geometric connectivity change. For the Betti-$1$ barcodes, classical persistent homology ($N=2$) fails to detect the geometric connectivity changes occurring at $r_3$, $r_4$, and $r_6$. By comparison, persistent Mayer homology at $N=5,7$ successfully resolves all of these geometric changes. Image courtesy of Caleb Khaemba.}
    \label{fig:Barcodes}
\end{figure*}

Traditional TDA relies primarily on persistent homology, which is known to exhibit several fundamental limitations \cite{su2025topological}. In particular, persistent homology is insensitive to many non-topological changes in geometric connectivity that occur during filtration as shown in FIG. \ref{fig:Barcodes}, limiting its descriptive power for complex datasets. Persistent combinatorial Laplacians \cite{wang2020persistent, memoli2022persistent} partially overcome this limitation by encoding both topological invariants and geometric connectivity. However, they remain computationally demanding for large-scale datasets, and existing quantum algorithms for combinatorial Laplacians may not provide practical speedups. Persistent Mayer homology \cite{shen2023persistent} offers an attractive alternative by introducing a two-parameter family of Mayer Betti numbers that simultaneously characterize conventional topological features and geometric connectivity throughout the filtration process as shown in FIG. \ref{fig:Barcodes}.
\begingroup\color{black}
Although the Mayer boundary has complex entries, its Betti numbers
are integer dimensions, and complex coefficients alone do not imply
greater classical computational cost. Low-degree Mayer features may
be useful for data analysis, but $\beta_{0,p}\leq n$ and
$\beta_{1,p}\leq\binom n2$ for all $N,p$.
Our motivation for large quantum speedups instead concerns higher
degrees with exponentially large chain spaces, under the stated
access, normalized Betti-number, and spectral-gap assumptions.
Identifying useful biological features in this regime is a promising
direction for future work.
\par\endgroup

\paragraph{Comparative Genomics. }

Comparative genomics seeks to identify structural and functional relationships among genomes across species, populations, or disease states. Genomic variation often involves complex structural rearrangements, including inversions, duplications, translocations, and higher-order chromatin organization, which are not adequately characterized by existing topological sequence analysis (TSA) methods. \cite{rabadan2019topological, liu2026topological}

A recently developed $\Delta$-complex framework for TSA \cite{liu2025delta} represents sequence segments as simplices in a $\Delta$-complex, a generalization of simplicial complexes that offers greater flexibility. 
Real-valued functions defined on sequence segments, together with the associated $\Delta$-closure, generate filtrations of $\Delta$-complexes from which persistent homology is computed. \textcolor{black}{In the reported bacterial-genome experiments, this framework extracts topological features substantially faster than the earlier $k$-mer topology approach \cite{hozumi2026revealing}, although its clustering performance is somewhat inferior \cite[Section~4]{liu2025delta}.}\footnote{\textcolor{black}{A $k$-mer is a contiguous subsequence of length $k$, for example, a string of $k$ consecutive nucleotides in a DNA sequence. Here $k$ denotes sequence length, not simplex dimension.}} \textcolor{black}{This trade-off motivates exploring richer topological features to improve accuracy while retaining computational efficiency.}

\begingroup\color{black}
Persistent Mayer homology may provide complementary features for
topological sequence analysis, to be evaluated against ordinary
persistence and $k$-mer methods on genomic benchmarks.
Our current encoding represents each simplex by a subset of distinct
vertices, whereas a $\Delta$-complex can have distinct simplices with
the same vertex set and vertex identifications within a simplex.
Extending the quantum algorithm to this setting is an open question
for future work, requiring a compatible encoding, Mayer boundary
maps, and efficient input access.\footnote{Such an encoding must
retain simplex identities and face identifications, support ordered
face maps satisfying $\partial^N=0$, and allow efficient membership
queries, boundary access, and state preparation.}
Under the required algorithmic assumptions, such an extension could
enable quantum estimation of persistent Mayer ranks for genomic analysis.
\par\endgroup

\paragraph{Drug Discovery. }

Modern drug discovery increasingly depends on large-scale virtual screening, molecular docking, ranking, ADMET analysis, and structure-based machine learning. Although traditional TDA has had much success in this field \cite{nguyen2020review}, 
drug binding frequently involves small geometric changes in molecular interactions that are difficult to distinguish using conventional persistent homology.  Persistent Mayer homology provides richer multiscale descriptors by incorporating both topology and evolving geometric connectivity throughout molecular filtration. Quantum acceleration could dramatically improve the scalability of analyzing billions of molecular conformations, protein-ligand complexes, and simulation trajectories, thereby accelerating target identification, hit discovery, lead optimization, and drug repurposing.



\paragraph{Neuroscience. }

Brain networks exhibit highly complex structural and functional organization across multiple spatial and temporal scales. Classical TDA has been successfully applied to characterize brain connectivity but often limits to small datasets and overlooks important changes in connection strength and local geometric organization \cite{curto2025topological}. Persistent Mayer homology naturally captures both topological organization and connectivity evolution over scales, offering a more comprehensive description of neural circuits and dynamic brain networks. Quantum-accelerated computation would make it feasible to analyze increasingly large neuron datasets, supporting studies of brain evolution, cognitive function, neurological disorders, and human-machine interactions.

As a sum up, Mayer features have shown preliminary predictive value in one low-dimensional molecular-learning setting \cite{feng2025mayer}; whether useful high-dimensional Mayer features occur in applications remains open. This is where we believe our quantum algorithm can unlock its ability and deliver potentially useful real-world applications.

\section{DISCUSSION AND OPEN QUESTIONS}
\label{sec: discussion}

\paragraph{Circuit-to-Mayer-homology perspective.}
Another possible perspective on Mayer homology comes from circuit-to-Hamiltonian and Hamiltonian-to-homology type constructions. Existing hardness constructions for clique homology \cite{crichigno2024clique,king2024gapped} are based on the ordinary boundary operator, whose coefficients are real signs. In contrast, the Mayer boundary operator carries cyclotomic phases $\xi^i$, and therefore naturally gives rise to operators over $\mathbb{Q}(\zeta_N)$. This suggests that Mayer homology may encode complex gate structures more directly than ordinary simplicial homology. For example, small Mayer gadgets already realize complex projectors associated with cyclotomic phases, and higher cyclotomic phases, including $T$-like phases when $N=8$, are in principle available from the same mechanism. These observations indicate an additional expressivity of Mayer homology beyond the ordinary real boundary formalism.

The filled triangle in Fig.~\ref{fig:filled-triangle-projector} gives a minimal illustration of this phenomenon. In the ordinary case $N=2$, the boundary coefficients are real signs. For $N=3$, however, the same filled triangle already carries the phases $1,\xi,\xi^2$ in its boundary map. As a result, the corresponding Mayer Laplacian term on $C_0$ is proportional to a rank-one projector onto a genuinely complex cyclotomic state. Higher-dimensional examples, such as the tetrahedral examples discussed in the appendix, can also be viewed as producing projectors onto Mayer image or harmonic subspaces, but the triangle captures the basic complex phase mechanism in the most transparent form.

At the same time, we do not claim here a full circuit-to-homology reduction for Mayer Betti numbers. A first basic question is what should play the role of the $2^n$-dimensional hole space for $n$-qubits in such a construction. In ordinary circuit-to-homology reductions, this space is built from tensor products of small homological gadgets, but the corresponding Mayer construction is not immediate because for each degree $d$, Mayer homology contains several groups $H_{d,p}$ indexed by $p=1,\ldots,N-1$, and does not inherit the ordinary K\"unneth-type behavior in a straightforward way. Thus, even the choice of the base family of holes representing computational basis states remains part of the open problem. Turning the local complex-gate/projector gadgets above into a complete worst-case complexity result would further require constructing a full history-state gadget, isolating the relevant propagation projectors, controlling the appropriate Mayer sector, and proving the resulting spectral gap. We leave such a circuit-to-Mayer-homology construction, and the corresponding worst-case complexity of Mayer Betti-number estimation for $N>2$, as an open direction.

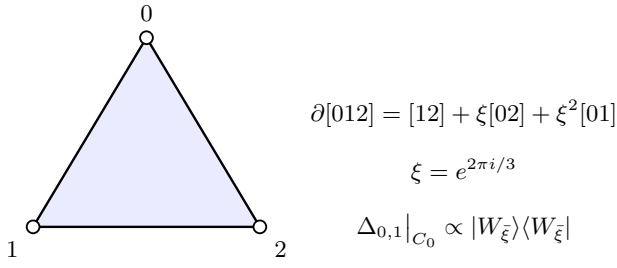
\begin{figure}[htbp]
\centering
\begin{tikzpicture}[scale=1.0]
    \coordinate (v0) at (0,1.7);
    \coordinate (v1) at (-1.5,-0.8);
    \coordinate (v2) at (1.5,-0.8);

    \fill[blue!8] (v0)--(v1)--(v2)--cycle;
    \draw[line width=0.9pt] (v0)--(v1)--(v2)--cycle;
    \foreach \v/\lab/\pos in {v0/0/above,v1/1/below left,v2/2/below right}{
        \filldraw[fill=white,draw=black,line width=0.7pt] (\v) circle (2.4pt) node[\pos=3pt] {$\lab$};
    }
    \node at (4.2,0.7) {$\partial[012]=[12]+\xi[02]+\xi^2[01]$};
    \node at (4.2,-0.05) {$\xi=e^{2\pi i/3}$};
    \node at (4.2,-0.85) {$\Delta_{0,1}\big|_{C_0}\propto \ket{W_{\bar\xi}}\!\bra{W_{\bar\xi}}$};
\end{tikzpicture}
\caption{A filled triangle as a minimal Mayer gadget. For $N=3$, applying the Mayer boundary twice to the filled triangle gives, up to an overall phase, the vertex-chain vector $\partial^2[012]\propto\ket{W_{\bar\xi}}=(\ket{0}+\bar\xi\ket{1}+\bar\xi^2\ket{2})/\sqrt{3}$, where $\ket{0},\ket{1},\ket{2}$ are the three vertex basis states. Since the cycle condition on $C_0$ is trivial, all vertex states are initially cycles, and the term $\Delta_{0,1}|_{C_0}$ is the image-boundary contribution $\partial^2(\partial^2)^\dagger$. Thus the filled triangle penalizes the one complex vertex combination $\ket{W_{\bar\xi}}$, while the orthogonal vertex combinations remain harmonic.}
\label{fig:filled-triangle-projector}
\end{figure}

\section{CONCLUSION}
\label{sec: conclusion}
\begingroup
\color{black}
We have developed quantum algorithms for estimating Mayer Betti numbers
and their persistent counterparts, with polynomial scaling under the
stated input-access, normalized Betti-number, and spectral assumptions.
Our structural results address the normalization bottleneck: explicit
simplex-count conditions guarantee large normalized Mayer Betti numbers,
and the cone family realizes a constant normalized Mayer Betti number in
exponentially large chain spaces, together with a gap lower bound in a
specified sector. Our numerical examples also exhibit large normalized
Mayer Betti numbers in the dense regimes examined, complementing the
analytical results. These results give explicit settings in which the
conditions for efficient quantum estimation coexist. The combinatorial
formulas and operator representations developed here also provide tools
for further study of Mayer homology.

In the corresponding regimes, the quantum algorithms improve
superpolynomially over explicit enumeration, with an exponential
improvement when the number of simplices is exponential in the input size.
Our analysis of classical randomized methods identifies the additional
sampling conditions needed for efficient extensions to Mayer homology.
Establishing their performance at inverse-polynomial precision under the
quantum algorithm's assumptions, and identifying families with an advantage
over the best known classical estimation methods under matched input and
precision assumptions, remain important directions for future work.

\textcolor{black}{Our conditional resource estimates identify the normalized gap and target
precision as key parameters for implementing the quantum algorithm in
practice and for future comparisons with classical methods.}
Existing applications of persistent Mayer homology motivate exploring
these algorithms for data analysis. An important next step is to identify
useful high-degree Mayer features that also satisfy the conditions for
efficient quantum estimation. Together, these directions make Mayer
homology a promising setting for investigating quantum advantages in TDA
and their potential value for real-world applications.
\par
\endgroup

\section*{Acknowledgement}
We thank Nathan Wiebe for insightful discussion regarding dequantization algorithm. We thank Alexander Schmidhuber for the insightful comment regarding vertex ordering. We acknowledge the use of ChatGPT, developed by OpenAI, and Gemini, developed by Google, in our work. These tools were used to supply proof ideas, as well as to help proofreading the work for better clarity and presentation. 
RH was supported by JST PRESTO Grant Number JPMJPR23F9 and JST ASPIRE Grant Number JPMJAP26A4, Japan.
DWB worked on this project under a sponsored research agreement with Google Quantum AI.
This project is supported by Australian Research Council Discovery Projects DP220101602 and DP260102543.

\bibliography{main.bib}
\bibliographystyle{unsrtnat}


\clearpage
\newpage
\appendix
\onecolumngrid

In this appendix, we provide full details on what we left in the main text. Specifically, we provide key recipes within block-encoding and quantum singular value transformation framework. We then provide a thorough introduction to Mayer homology, including core definitions and concepts. We also discuss in detail the construction of the block-encoding of the Mayer $N$-boundary operator, and also provide a completed analysis of our quantum algorithm. 

\section{Block-encoding recipes}
\label{sec: blockencodingrecipes}
Given that the block-encoding was defined in the main text (see Section \ref{sec: estimatingmayerbetti}), they key arithmetic recipes that we need are the following. 

\begin{lemma}[Product of block-encoded operators; Lemma 53 of \cite{gilyen2019quantum}]
\label{lemma: product}
If $U$ is an $(\alpha, a, \delta)$-block-encoding of a $s$-qubit operator $A$, and $V$ is a $(\beta, b, \epsilon)$-block-encoding of a $s$-qubit operator $B$, then $(I_b \otimes U) (I_a \otimes V)$ is an $(\alpha \beta, a+b, \alpha \epsilon+ \beta \delta)$-block-encoding of $AB$.
\end{lemma}

\begin{lemma}[Linear combination of block-encoded operators; Lemma 52 of \cite{gilyen2019quantum}]
    \label{lemma: linearcombination}
    Let $A = \sum_{j=1}^m y_j A_j$ be an $s$-qubit operator and $\epsilon\in \Rbb_+$. Suppose that $(P_L,P_r)$ is a $(\beta, b, \epsilon_1)$-state preparation pair for $y$, $W = \sum_{j=0}^{m-1} \ket{j}\bra{j} \otimes U_j + \big( (I - \sum_{j=0}^{m-1}  \ket{j}\bra{j}) \otimes I_a \otimes I_b \big) $ is an $s+a+b$ qubit unitary such that for all $j \in 0,1,...,m$ we have that $U_j$ is an $(\alpha, a, \epsilon_2)$-block-encoding of $A_j$. Then we can implement a $(\alpha\beta, a+b, \alpha \epsilon_1 + \alpha \beta \epsilon_2)$-block-encoding of $A$, with a single use of $W, P_R, P_L^\dagger$.
\end{lemma}

A standard result called QSVT allows the transformation of block-encoded operators:
\begin{lemma}
\label{lemma: qsvt}[\cite{gilyen2019quantum} Theorem 56]
Suppose that $U$ is an
$(\alpha, a, \epsilon)$-encoding of a Hermitian matrix $A$. (See Definition 43 of~\cite{gilyen2019quantum} for the definition.)
If $P \in \mathbb{R}[x]$ is a degree-$d$ polynomial satisfying that
\begin{itemize}
\item for all $x \in [-1,1]$: $|P(x)| \leq \frac{1}{2}$,
\end{itemize}
then, there is a quantum circuit $\tilde{U}$, which is an $(1,a+2,4d \sqrt{\frac{\epsilon}{\alpha}})$-encoding of $P(A/\alpha)$ and
consists of $d$ applications of $U$ and $U^\dagger$ gates, a single application of controlled-$U$ and $\mathcal{O}((a+1)d)$
other one- and two-qubit gates.
\end{lemma}

\section{Mayer homology}
\label{sec: mayerhomology}
Following the previous appendix, this appendix contains an introduction to Mayer homology. We first outline the essential concepts and provide concrete examples to show the subtlety of Mayer homology, especially in comparison to simplicial homology.

\subsection*{Essential concepts and properties}
Let $K$ be a simplicial complex with $n = \lvert X_0 \rvert$ vertices, and let $N$ be a prime number. Write $X_d$ for the set of $d$-simplices of $X$ and set $\xi = e^{2\pi i/N}$.

\paragraph{$N$-Chain Complex.}

Define the chain group $C_d^K := \mathrm{Span}(\{\ket{\sigma}\}_{\sigma \in S_d^K})$ with inner product $\braket{\sigma' \mid \sigma} = \delta_{\sigma,\sigma'}$. The \emph{$N$-boundary operator} $\partial_d \colon C_d^K \to C_{d-1}^K$ is defined by
\begin{equation}
    \partial_d \ket{v_0 v_1 \cdots v_d} = \sum_{i=0}^{d} \xi^i \ket{v_0 \cdots \hat{v}_i \cdots v_d}.
\end{equation}
One can verify that $\partial_d^N = 0$; thus $(C_*, \partial_*)$ forms an $N$-chain complex.

\paragraph{Mayer Homology Groups.}

For $1 \leq q \leq N-1$, define the \emph{$q$-cycle group} and \emph{$q$-boundary group} by
\begin{equation}
    Z_d^p := \mathrm{Ker}(\partial^p \colon C_d^K \to C_{d-p}^K), \qquad B_d^p := \mathrm{Im}(\partial^{N-p} \colon C_{d+N-p}^K \to C_d^K).
\end{equation}
Since $\partial^N = 0$ one has $B_d^p \subseteq Z_d^p$, and the \emph{$p$-th Mayer homology group} is defined by
\begin{equation}
    H_{d,p} := Z_d^p / B_d^p.
\end{equation}
The corresponding \emph{Mayer Betti number} is $\beta_{d,p} := \dim H_{d,p}$.

\paragraph{Mayer Laplacian.}

The \emph{Mayer Laplacian} $\Delta_{d,p} \colon C_d^K \to C_d^K$ is defined by
\begin{equation}
    \Delta_{d,p} := (\partial_d^p )^\dagger \partial_d^p + \partial_{d+N-p}^{N-p} (\partial_{d+N-p}^{N-p})^\dagger,
\end{equation}
where $\dagger$ denotes the Hermitian adjoint. This operator is positive semidefinite, and the following analogue of the Hodge theorem holds \cite{shen2023persistent}:
\begin{equation}
    \dim \ker \Delta_{d,p} = \beta_{d,p}.
\end{equation}

\paragraph{Harmonic Mayer Homology.}

We define \emph{harmonic Mayer homology group} by
\begin{equation}
    \mathcal{H}_d^p := Z_d^p \cap (B_d^p)^{\perp}.
\end{equation}
By the Hodge theorem, $\ker \Delta_{d,p} = \mathcal{H}_d^p$, and each homology class in $H_{d,p}$ has a unique harmonic representative in $\mathcal{H}_{d,q}$.

\paragraph{Reduction to simplicial homology.}
\textcolor{black}{For $N=2$, the phase is $\xi=-1$ and the Mayer
boundary becomes the usual alternating simplicial boundary.
With the complex chain spaces used here, $H_{d,1}$ therefore
recovers ordinary simplicial homology with complex coefficients.
Ordinary homology can also be defined with other coefficient choices.
For $N>2$, the change is in the boundary phases and the relation
$\partial^N=0$, leading to the family of groups $H_{d,p}$;
the use of complex coefficients is not itself what distinguishes
Mayer homology from ordinary homology.}

\subsection*{Examples of Mayer Homology}
\label{sec:appendix-examples}

\subsubsection*{Example: $\Delta[3]$ with $N=3$}

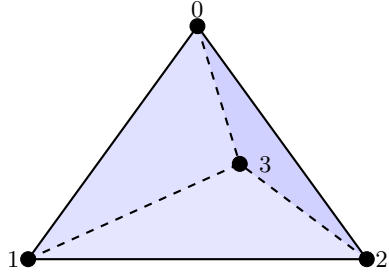
\begin{figure}[H]
\centering
\begin{tikzpicture}[scale=1.4]
  \coordinate (v0) at (0, 2.2);
  \coordinate (v1) at (-1.6, 0);
  \coordinate (v2) at (1.6, 0);
  \coordinate (v3) at (0.4, 0.9);

  \fill[blue!10] (v0) -- (v1) -- (v2) -- cycle;
  \fill[blue!18] (v0) -- (v2) -- (v3) -- cycle;
  \fill[blue!12] (v0) -- (v1) -- (v3) -- cycle;

  \draw[thick] (v0) -- (v1);
  \draw[thick] (v0) -- (v2);
  \draw[thick] (v1) -- (v2);

  \draw[thick, dashed] (v0) -- (v3);
  \draw[thick, dashed] (v2) -- (v3);
  \draw[thick, dashed] (v1) -- (v3);

  \foreach \v in {v0, v1, v2, v3}
    \filldraw[black] (\v) circle (2pt);

  \node[above]       at (v0) {$0$};
  \node[left]        at (v1) {$1$};
  \node[right]       at (v2) {$2$};
  \node[right=4pt]   at (v3) {$3$};
\end{tikzpicture}
\caption{The filled tetrahedron $\Delta[3]$.}
\end{figure}

We compute the Mayer homology of the filled tetrahedron $\Delta[3]$ with $N=3$. The simplicial complex $\Delta[3]$ has vertex set $\{0,1,2,3\}$ and contains all simplices of dimension $0$ through $3$ (see figure above). Set $\xi = e^{2\pi i/3}$.

\paragraph{Boundary Operators}

The $3$-boundary operator acts on the unique $3$-simplex by
\begin{equation}
    \partial_3 \ket{0123} = \ket{123} + \xi\ket{023} + \xi^2\ket{013} + \ket{012},
\end{equation}
and on $2$-simplices, e.g.,
\begin{equation}
    \partial_2 \ket{012} = \ket{12} + \xi\ket{02} + \xi^2\ket{01}.
\end{equation}
On $1$-simplices: $\partial_1\ket{vw} = \ket{w} + \xi\ket{v}$ for $0 \leq v < w \leq 3$. One can verify $\partial^3 = 0$.

\paragraph{Mayer Homology Groups}

Direct computation yields:
\begin{equation}
    H_{3}^1 = H_{3}^2 = H_{2}^1 = H_{2}^2 = H_{0}^2 = 0,
\end{equation}
\begin{equation}
    H_{1}^1(\Delta[3];\mathbb{C}) \cong \mathbb{C}, \qquad
    H_{1}^2(\Delta[3];\mathbb{C}) \cong \mathbb{C}^2, \qquad
    H_{0}^1(\Delta[3];\mathbb{C}) \cong \mathbb{C}.
\end{equation}

\begin{table}[H]
    \centering
    \begin{tabular}{|c|c|c|c|}
    \hline
    & $\beta_0^{p}$ & $\beta_1^p$ & $\beta_2^p$ \\
\hline
   $p=1$ & 1 & 1   & 0 \\
   $p=2$ & 0   & 2  & 0  \\
   \hline
\end{tabular}
\end{table}

\paragraph{Harmonic Representatives}

We compute the harmonic Mayer homology groups $\mathcal{H}_{d,q} = Z_{d,q} \cap B_{d,q}^\perp$ for each non-trivial class.

\textbf{$\mathcal{H}_{0}^1$:} Since $Z_{0}^1 = C_0$, we have $\mathcal{H}_{0}^1 = (B_{0}^1)^\perp$. The four generators of $B_{0}^1 = \mathrm{im}(\partial^2: C_2 \to C_0)$ span a 3-dimensional subspace (the $4\times 4$ matrix of generators has determinant zero by $1+\xi+\xi^2=0$). Solving $(B_{0}^1)^\perp$ forces all coefficients equal, giving:
\begin{equation}
    \mathcal{H}_{0}^1 = \mathrm{span}\{\ket{0}+\ket{1}+\ket{2}+\ket{3}\}.
\end{equation}

\textbf{$\mathcal{H}_{1}^1$:} $\dim Z_{1,1} = 2$ with basis $e_1 = \ket{01}-\ket{03}-\ket{12}+(1-\xi^2)\ket{13}+\xi^2\ket{23}$, $e_2 = \ket{02}-\ket{03}-\ket{12}+\ket{13}$, and $B_{1,1} = \mathrm{span}\{e_1+\xi^2 e_2\}$. Imposing $\braket{v|b_1}=0$ with $\braket{e_1|b_1}=3(1-\xi)$ and $\braket{e_2|b_1}=-3\xi$ gives $\beta = \alpha(\xi-1)$:
\begin{equation}
    \mathcal{H}_{1}^1 = \mathrm{span}\{\ket{01}+(\xi-1)\ket{02}-\xi\ket{03}-\xi\ket{12}+(\xi-\xi^2)\ket{13}+\xi^2\ket{23}\}.
\end{equation}

\textbf{$\mathcal{H}_{1}^2$:} Since $Z_{1,2} = C_1$, we have $\mathcal{H}_{1,2} = B_{1,2}^\perp$ where $B_{1,2} = \mathrm{im}(\partial_2: C_2 \to C_1)$ has dimension 4. The orthogonality conditions reduce to a 2-parameter family with the structural constraint $c_{03}=c_{12}$ and $c_{23}=c_{01}$:
\begin{equation}
    \mathcal{H}_{1}^2 = \mathrm{span}\{\ket{u_1},\, \ket{u_2}\},
\end{equation}
\begin{equation}
    \ket{u_1} = \ket{01}-\xi\ket{03}-\xi\ket{12}+(1-\xi)\ket{13}+\ket{23}, \qquad
    \ket{u_2} = \ket{02}-\xi^2\ket{03}-\xi^2\ket{12}+\xi\ket{13}.
\end{equation}

\paragraph{Remark}

Note that $\Delta[3]$ is contractible, so its standard homology satisfies $H_0 = \mathbb{C}$ and $H_k = 0$ for $k \geq 1$. Nevertheless, its Mayer homology is non-trivial, demonstrating that Mayer homology captures combinatorial structure of the simplicial complex beyond what standard homology detects.

\subsubsection*{Example: Annulus with $N=3$ and $N=5$}
In a similar manner to the previous example, we compute the Mayer Betti numbers of the annulus, which are shown in Table \ref{tab: annulusMayerbetti}.
\begin{figure}[H]
    \centering
    \includegraphics[width=0.3\linewidth]{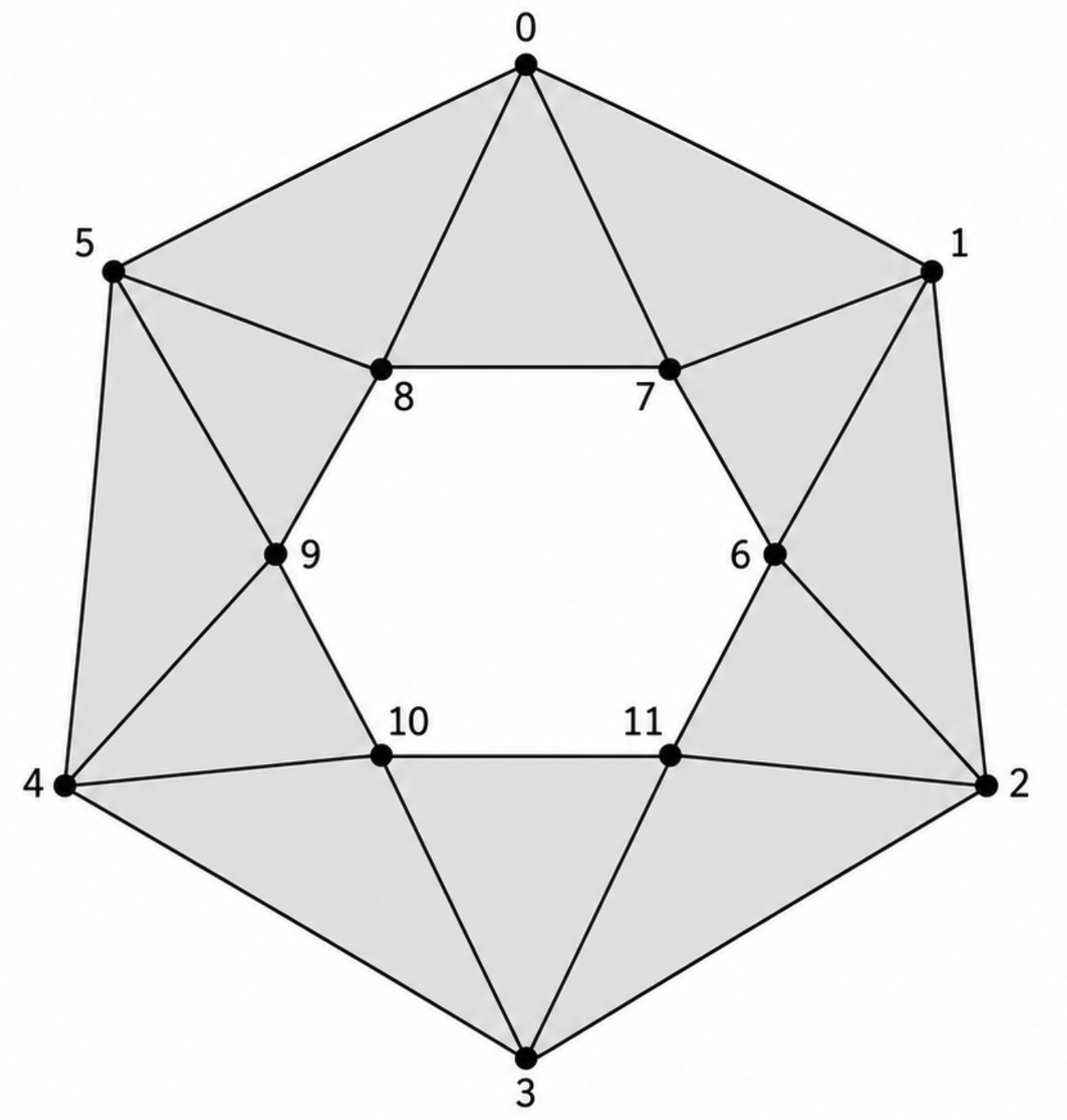}
    \caption{ A annulus having $12$ vertices, $24$ edges and $12$ faces. }
    \label{fig: annulus}
\end{figure}

\begin{table}[H]
    \centering
    \begin{tabular}{|c|c|c|c|}
    \hline
    & $\beta_0^{p}$ & $\beta_1^p$ & $\beta_2^p$ \\
\hline
   $p=1$ & 1& 12   & 0 \\
   $p=2$ & 0   &12  &1  \\
   \hline
\end{tabular}
    \caption{Mayer Betti numbers for $N=3$.}
    \begin{tabular}{|c|c|c|c|}
    \hline
    & $\beta_0^{p}$ & $\beta_1^p$ & $\beta_2^p$ \\
\hline
   $p=1$ & 12& 12   & 0 \\
   $p=2$ & 12   & 24  &1  \\
   $p=3$ & 1 & 24 & 12 \\
   $p=4$ & 0 & 12 & 12 \\
   \hline
\end{tabular}
    \caption{Mayer Betti numbers for $N=5$.}
    \label{tab: annulusMayerbetti}
\end{table}

\section{Persistent Mayer Homology }
\label{sec: persistencemayer}
In the above, we have introduced Mayer homology, in which the $N$-chain complex is defined on the simplicial complex $K$. In the context of persistent Mayer homology, our object of interest is the filtration of complexes
\begin{equation}
    K_1 \subseteq K_2 \subseteq ... \subseteq K_m.
\end{equation}
\textcolor{black}{Fix a total order on the vertices of $K_m$, and equip
each $K_i$ with the induced order. The inclusion maps then commute with
the Mayer boundaries.}
The sequence of inclusion maps between the corresponding $N$-chain complexes
\begin{align}
    C_d^{K_1} \longrightarrow C_{d}^{K_2} \longrightarrow \cdots \longrightarrow C_d^{K_m}
\end{align}
induces a sequence of linear maps on Mayer homology
\begin{align}
    H_{d,p}(K_1) \longrightarrow H_{d,p}(K_2) \longrightarrow \cdots \longrightarrow H_{d,p}(K_m)
\end{align}
for any $1 \leq p \leq N-1$.

Consider the inclusion map $j_{a,b}: K_a \to K_b$. The $(a,b)$-persistent Mayer homology is defined by
\begin{align}
    H_{d,p}^{a,b} := \text{Im } (H_{d,p}^{K_a} \longrightarrow H_{d,p}^{K_b} )
\end{align}
The rank of $H_{d,p}^{a,b}$ is the $(a,b)$-persistent Mayer Betti number, denoted by $\beta_{n,p}^{a,b}$.

To provide a convenient way to compute the persistent Mayer Betti number, we define the persistent Mayer Laplacian as follows. For any $1 \leq p \leq N-1$, let
\begin{align}
    C_{d,p}^{a,b} = \{ x \in C_d^{K_b} \,|\, \partial_d^{p,(b)} x \in C_{d-p}^{K_a}  \} \subset C_d^{K_b}
\end{align}
where $\partial_d^{p,(b)}= \partial_{d-p+1}\cdots\partial_d: C_d^{K_b} \to C_{d-p}^{K_b}$ is the iterated differential on $K_b$. By construction, the image of the restriction of $\partial_d^{p,(b)}$ to $C_{d,p}^{a,b}$ is contained in $C_{d-p}^{K_a}$. We denote this restriction with $\partial_d^{p,(a,b)} : C_{d,p}^{a,b} \to C_{d-p}^{K_a}$, so that $\partial_d^{p,(a,b)} x = \partial_d^{p,(b)}x \in C_{d-p}^{K_a}$. The $(a,b)$-persistent Mayer Laplacian $\Delta_{d,p}^{a,b}: C_d^{K_a} \to C_d^{K_a}$ is defined as
\begin{align}
    \label{eq:persistent_Mayer_Laplacian}
    \Delta_{d,p}^{a,b} := \left(\partial_d^{p,(a)} \right)^\dagger \partial_{d}^{p,(a)} + \partial_{d+N-p}^{N-p,(a,b)} \left(\partial_{d+N-p}^{N-p,(a,b)}\right)^\dagger.
\end{align}
According to \cite{shen2023persistent}, the following theorem relates the persistent Mayer Laplacian and persistent Mayer homology:
\begin{theorem}
    For any $a \leq b$, we have an isomorphism $\ker \Delta_{d,p}^{a,b} \cong H_{d,p}^{a,b}$ where $n \geq 0$, and $1 \leq q \leq N-1$.
\end{theorem}
The theorem implies that the persistent Betti number $\beta_{d,p}^{a,b}$ can be computed as the dimension of the kernel of $\Delta_{d,p}^{a,b}$.

\section{Alternative form for the persistent Mayer Laplacian}
\label{sec: alternativeform}
In the following, we give a more explicit form for the operator $\partial_{d+N-p}^{N-p,(a,b)}: C_{d+N-p, N-p}^{a,b} \to C_d^{K_a}$ which is the restriction of $\partial_{d+N-p}^{N-p,(b)}$ to the subspace $C_{d+N-p,N-p}^{a,b}$ of $C_{d+N-p}^{K_b}$. The key idea, based on \cite{memoli2022persistent}, employs Schur complement.
\begin{definition}[Schur complement]
\label{def: schurcomplemenmt}
    Let $M \in \Rbb^{n \times n}$ be a block matrix
    $$ M = \begin{pmatrix}
        A & B \\
        C& D
    \end{pmatrix}$$
    where $D \in \Rbb^{d\times d}$ is a block matrix. The Schur complement of $D$ in $M$ is:
    $$ M \backslash D := A - BD^+C$$
    where $D^+$ is the Moore-Penrose pseudo-inverse of $D$. For $\emptyset \neq I \subseteq [n], J \subseteq [n]$, we denote the submatrix of $M$ consisting of rows and columns of indices of $I$ and $J$ by $M(I,J)$. Then, the Schur complement of $M(I,I)$ in $M$ is
    $$ M/ M(I,I) \equiv M([n]\backslash I, [n] \backslash I) - M([n] \backslash I,I)M(I,I)^+ M(I, [n] \backslash  I)  $$
    where $[n]/I = \{ a \in [n]\  | \ a \in I \} $ and $M(I,J)$ is the submatrix of $M$ that consists of rows and columns of $M$ indexed by $I$ and $J$, respectively.
\end{definition}

Let $S_d^{K_a}, S_d^{K_b}$ be the set of $d$-simplices in $K_a, K_b$, and $I_a^b \equiv S_d^{K_b} \backslash S_d^{K_a}$ denotes those $d$-simplices in $K_b$ which are not in $K_a$. Then based on \cite{memoli2022persistent}, the operator $ \partial_{d+N-p}^{N-p, (a,b)} \left(\partial_{d+N-p}^{N-p,(a,b)}\right)^\dagger $ can be found as the Schur complement of the operator $\partial_{d+N-p}^{N-p, (b)} \left(\partial_{d+N-p}^{N-p,(b)}\right)^\dagger $:
\begin{align}
    \label{eq:alternative_form}
    \partial_{d+N-p}^{N-p, (a,b)} \left(\partial_{d+N-p}^{N-p,(a,b)}\right)^\dagger = \partial_{d+N-p}^{N-p, (b)} \left(\partial_{d+N-p}^{N-p,(b)}\right)^\dagger/ \left( \partial_{d+N-p}^{N-p, (b)} \left(\partial_{d+N-p}^{N-p,(b)}\right)^\dagger\right)\left( I_a^b, I_a^b \right)
\end{align}
In fact, by replacing $p=1,N=2$, we obtain the Schur complement of $\partial_{d+1}^{(a,b)} \left(\partial_{d+1}^{(a,b)}\right)^\dagger  $, which appears in \cite{memoli2022persistent} used. Substituting \eqref{eq:alternative_form} into \eqref{eq:persistent_Mayer_Laplacian}, we obtain:
\begin{align}
      \Delta_{d,p}^{a,b} :=  \left(\partial_d^{p,(a)} \right)^\dagger \partial_d^{p,(a)} + \partial_{d+N-p}^{N-p,(b)} \left(\partial_{d+N-p}^{N-p,(b)}\right)^\dagger/ \left( \partial_{d+N-p}^{N-p,(b)} \left(\partial_{d+N-p}^{N-p,(b)}\right)^\dagger\right)( I_a^b, I_a^b ).
\end{align}
Below, we will show how to block-encode this operator, and later how to block-encode the persistent Mayer Laplacian.

\section{Anyonic Operators}
\label{sec: anyonic}

\paragraph{Setup}

Let $n = |X_0|$ be the number of vertices, ordered $v_0 < v_1 < \cdots < v_{n-1}$.
We introduce \emph{anyonic annihilation operators} $\hat{a}_{v_0}, \ldots, \hat{a}_{v_{n-1}}$ satisfying the off-site generalized commutation relations
\begin{align}
    \hat{a}_j \hat{a}_k^\dagger - e^{-i\theta \operatorname{sgn}(j-k)} \hat{a}_k^\dagger \hat{a}_j &= 0, \qquad j\neq k, \label{eq:anyon-comm-1}\\
    \hat{a}_j \hat{a}_k - e^{i\theta \operatorname{sgn}(j-k)} \hat{a}_k \hat{a}_j &= 0, \qquad j\neq k, \label{eq:anyon-comm-2}
\end{align}
with statistical angle $\theta = -2\pi/N$, so that $e^{-i\theta} = \xi = e^{2\pi i/N}$.
On each site, we impose the hard-core conditions
\begin{equation}
    \hat{a}_j^2=(\hat{a}_j^\dagger)^2=0,
    \qquad
    \{\hat{a}_j,\hat{a}_j^\dagger\}=1.
\end{equation}

\begin{remark}[Hard-core anyonic exchange algebra on an ordered set of vertices]
\label{rem:abelian}
The relations above define a hard-core representation of the one-dimensional Abelian anyonic exchange algebra. The exchange phase $e^{i\theta\,\operatorname{sgn}(j-k)}$ in \eqref{eq:anyon-comm-1}--\eqref{eq:anyon-comm-2} is a scalar that commutes with all operators. The cases $\theta = \pi$ modulo $2\pi$ (fermions, $\xi = -1$, $N = 2$) and $\theta = 0$ (hard-core bosons) give the two familiar limits. For $\theta = -2\pi/N$ with prime $N \geq 3$, the off-site exchange relations carry nontrivial $\mathbb{Z}_N$ phases. Here, ``anyonic'' refers to this operator exchange algebra; we do not construct a physical braiding process.
\end{remark}

Define simplex states by ordered application of creation operators:
\begin{equation}
    \ket{v_0 v_1 \cdots v_d} := \hat{a}_{v_0}^\dagger \hat{a}_{v_1}^\dagger \cdots \hat{a}_{v_d}^\dagger \ket{\mathrm{vac}}.
\end{equation}

\paragraph{Action on simplex states. }

\begin{lemma}[Anyonic face removal]
\label{lem:anyon-face}
For a $d$-simplex $\sigma = v_0 v_1 \cdots v_d$ with $v_0 < v_1 < \cdots < v_d$,
\begin{equation}
    \hat{a}_{v_i} \ket{v_0 \cdots v_d} = \xi^i \ket{v_0 \cdots \hat{v}_i \cdots v_d}.
\end{equation}
\end{lemma}

\begin{proof}
Since $v_i > v_j$ for all $j < i$, relation \eqref{eq:anyon-comm-1} gives
$\hat{a}_{v_i} \hat{a}_{v_j}^\dagger = e^{-i\theta \operatorname{sgn}(v_i - v_j)} \hat{a}_{v_j}^\dagger \hat{a}_{v_i} = \xi \hat{a}_{v_j}^\dagger \hat{a}_{v_i}$ for each $j < i$.
Commuting $\hat{a}_{v_i}$ past $\hat{a}_{v_0}^\dagger, \ldots, \hat{a}_{v_{i-1}}^\dagger$ accumulates a factor $\xi^i$.
The remaining step $\hat{a}_{v_i}\hat{a}_{v_i}^\dagger\ket{\mathrm{vac}} = \ket{\mathrm{vac}}$ follows from $\hat{a}_{v_i}^\dagger\hat{a}_{v_i}\ket{\mathrm{vac}}=0$ via $\hat{a}_{v_i}\hat{a}_{v_i}^\dagger = 1 - \hat{a}_{v_i}^\dagger \hat{a}_{v_i}$.
\end{proof}

\paragraph{Anyonic boundary operator}

Define
\begin{equation}
    \partial := \sum_{v \in X_0} \hat{a}_v, \qquad \partial_d := \partial \cdot \Pi_d,
\end{equation}
where $\Pi_d = \mathrm{Proj}(C_d^K)$ is the orthogonal projection onto the $d$-chain group.

\begin{proposition}[Anyonic representation of the Mayer boundary operator]
\label{prop:anyon-boundary}
\begin{equation}
    \partial_d \ket{v_0 \cdots v_d} = \sum_{i=0}^{d} \xi^i \ket{v_0 \cdots \hat{v}_i \cdots v_d}.
\end{equation}
Hence $\partial_d$ coincides with the Mayer $N$-boundary operator.
\end{proposition}

\begin{proof}
Only the term $\hat{a}_{v_i}$ in $\partial$ contributes non-trivially to each face.
Apply Lemma~\ref{lem:anyon-face} to each $i \in \{0, \ldots, d\}$.
\end{proof}

\paragraph{$N$-nilpotency}

\begin{lemma}
\label{lem:nilpotent}
$\partial^N = 0$.
\end{lemma}

\begin{proof}
Expanding $\partial^N = \bigl(\sum_v \hat{a}_v\bigr)^N$: terms with any repeated index vanish by $\hat{a}_v^2 = 0$.
For $N$ distinct vertices $j_1 < \cdots < j_N$, all $N!$ orderings contribute
\begin{equation}
    \Bigl(\sum_{\pi \in S_N} \xi^{\mathrm{inv}(\pi)}\Bigr) \hat{a}_{j_1} \cdots \hat{a}_{j_N},
\end{equation}
where $\mathrm{inv}(\pi)$ counts the inversions of $\pi$ (each commutation \eqref{eq:anyon-comm-2} with $j_k < j_l$ contributes a factor $\xi$).
The prefactor is the $q$-symmetrizer at $q = \xi$:
\begin{equation}
    \sum_{\pi \in S_N} \xi^{\mathrm{inv}(\pi)} = [N]_\xi!, \qquad [N]_\xi! := \prod_{k=1}^{N} [k]_\xi, \qquad [k]_\xi := 1 + \xi + \cdots + \xi^{k-1}.
\end{equation}
Since $\xi = e^{2\pi i/N}$ is a primitive $N$-th root of unity, $[N]_\xi = 0$, hence $[N]_\xi! = 0$ and $\partial^N = 0$.
\end{proof}

\paragraph{Mayer Laplacian}

The \emph{Mayer Laplacian} of degree $(d,p)$ is
\begin{equation}
    \Delta_{d,p} := \left(\partial_d^{p}\right)^\dagger \partial_d^{p} + \partial_{d+N-p}^{N-p} \left(\partial_{d+N-p}^{N-p}\right)^\dagger = \left( \left(\partial^p\right)^\dagger \partial^p + \partial^{N-p}\left(\partial^{N-p}\right)^\dagger \right)\big|_{C_d^K}.
\end{equation}
By the Hodge theorem,
\begin{equation}
    \dim \ker \Delta_{d,p} = \beta_{d,p}.
\end{equation}

\paragraph{Quantum circuit representation}

The anyonic operators admit an explicit qubit representation via a \emph{generalized Jordan-Wigner transformation}:
\begin{equation}
\label{eq:gen-jw}
    \hat{a}_{v_k} = \underbrace{P_{v_0} \otimes \cdots \otimes P_{v_{k-1}}}_{k \text{ string sites}} \otimes\, Q^+_{v_k} \otimes I_{v_{k+1}} \otimes \cdots \otimes I_{v_{n-1}},
\end{equation}
where
\begin{equation}
    P = \begin{pmatrix} 1 & 0 \\ 0 & \xi \end{pmatrix}, \qquad Q^+ = \begin{pmatrix} 0 & 1 \\ 0 & 0 \end{pmatrix}.
\end{equation}
A direct check verifies that \eqref{eq:gen-jw} satisfies the anyonic commutation relations \eqref{eq:anyon-comm-1}--\eqref{eq:anyon-comm-2} with $\theta = -2\pi/N$.
At $N = 2$ (where $P = \sigma_z$), this reduces to the standard Jordan-Wigner transformation, recovering the fermionic representation of the standard simplicial boundary operator employed in \cite{akhalwaya2022representation}.
The string operator $P = \mathrm{diag}(1,\xi)$ closely parallels generalized Jordan-Wigner strings used in one-dimensional hard-core anyon models, up to conventions for the statistical angle \cite{wang2022exact}; we adapt it here to represent the Mayer boundary operator on qubits.
One could further develop this representation into a more explicit circuit-level implementation via the generalized Jordan-Wigner transformation, which we leave for future work.

\subsection*{Relation to Fractional Supersymmetry}
\label{sec:prelim-fsusy}

The constructions above suggest a fractional-SUSY-like extension of the usual correspondence between homology and fermionic supersymmetric quantum mechanics.
The following table summarizes the analogy between the standard fermionic ($N=2$) picture and the Mayer/order-$N$ version.

\paragraph{From fermionic SUSY to the Mayer/order-$N$ analogue}

Witten's identification of supersymmetric quantum mechanics with Hodge theory gives the usual dictionary between the fermionic supercharge and the simplicial boundary operator.
Motivated by the hard-core fermion framework of \cite{crichigno2021supersymmetry,crichigno2024clique}, the Mayer construction suggests a formal order-$N$ analogue of this dictionary at the level of $N$-complexes, replacing the $\mathbb{Z}_2$ grading and binary nilpotency $Q^2 = 0$ with a $\mathbb{Z}_N$ grading and $N$-step nilpotency $Q^N = 0$.
\begin{center}
\renewcommand{\arraystretch}{1.2}
\begin{tabular}{lll}
\toprule
Fermionic SUSY & Homology & Mayer/order-$N$ analogue \\
\midrule
Supercharge $Q^2=0$ & Boundary $\partial^2=0$ & Mayer differential $\partial^N=0$ \\
$\mathbb{Z}_2$ grading & Chain grading by $d\bmod 2$ & $\mathbb{Z}_N$ grading by $d\bmod N$ \\
Fermionic operators & Exterior algebra generators & Hard-core anyonic exchange representation \\
Hamiltonian $\{Q,Q^\dagger\}$ & Hodge Laplacian $\Delta_d$ & Positive Mayer Laplacians $\{\Delta_{d,p}\}_{p=1}^{N-1}$ \\
Ground states & Harmonic forms $\mathcal{H}_d\cong H_d$ & Mayer harmonics $\mathcal{H}_{d,p}\cong H_{d,p}$ \\
Witten index & Euler characteristic $\sum_d(-1)^d\beta_d$ & Formal $\xi$-graded Betti sums \\
\bottomrule
\end{tabular}
\end{center}

The first two columns reflect the correspondence between fermionic SUSY and homology used in recent works at the interface of quantum complexity and TDA~\cite{crichigno2021supersymmetry,crichigno2024clique}.
The third column is motivated by the algebraic framework of fractional supersymmetry, introduced via nilpotent variables of order $N$ in \cite{durand1993fractional} and developed through $\mathbb{Z}_N$-graded Lie algebras in \cite{azcarraga1996group}.
We stress that this table records a structural analogy only. We do not construct a fractional-supersymmetric supercharge satisfying a closure relation such as $Q^N=H$, nor do we identify the Mayer Laplacians with fractional-supersymmetric Hamiltonians.

This viewpoint also suggests replacing the parity weight $(-1)^F$ by the $\mathbb{Z}_N$ weight $\xi^F$ with $\xi = e^{2\pi i/N}$.
This motivates the formal $\xi$-graded Betti sums
\begin{equation}
    \chi_p(X) := \sum_{d} \xi^{p d}\, \beta_{d,p}, \qquad p \in \{1, \ldots, N-1\},
\end{equation}
which reduce to the classical Euler characteristic at $N = 2$, $p = 1$.



Motivated by the hard-core fermion framework of \cite{crichigno2021supersymmetry,crichigno2024clique}, one can view the construction above as an order-$N$ analogue in which fermionic creation/annihilation operators are replaced by the abelian anyonic operators introduced in this section.
We do not study this connection in detail here.
We point out that it gives a natural route for relating Mayer homology to quantum complexity, in parallel with the fermionic Hodge-Laplacian framework discussed above.

\section{Construction of the block-encoding of the $N$-boundary operator $\partial_d$}
\label{sec: Nboundaryoperator}
\begingroup
\color{black}
Fix $1\leq d\leq n-1$ and the total vertex order.
We retain $\alpha_d=\sqrt{(d+1)(n-d)}$ and regard
$\partial_d$ as a zero-padded operator on the $n$-qubit
subset register. Its column sums in absolute value are at most
$d+1$, and its row sums are at most $n-d$, so
$\|\partial_d\|\leq\alpha_d$.
We realize this normalization directly by state preparation,
using the Gram-matrix block-encoding principle of
\cite{gilyen2019quantum}, Lemma~47.

We first express the boundary matrix elements as overlaps between
vertex-deletion and vertex-addition states
(Section~\ref{sec:boundary-incidence}). We then construct a block-encoding
from their preparation unitaries and the degree and membership selections
(Section~\ref{sec:boundary-selection}). Finally, we describe the preparation
circuits and bound their gate complexity
(Section~\ref{sec:boundary-implementation}).

\subsection{Boundary matrix elements as state overlaps}
\label{sec:boundary-incidence}
For a $(d+1)$-vertex subset $\sigma$, let
$r_\sigma(v)=|\{u\in\sigma:u\prec v\}|$.
Define a vertex-deletion state for $\sigma$ and a vertex-addition state
for a $d$-vertex subset $\tau$, without yet restricting the subsets to $K$:
\begin{equation}
 \begin{aligned}
 |\phi_\sigma\rangle
 &=\frac1{\sqrt{d+1}}\sum_{v\in\sigma}
   \xi^{r_\sigma(v)}|\sigma\setminus\{v\}\rangle|\sigma\rangle,\\
 |\psi_\tau\rangle
 &=\frac1{\sqrt{n-d}}\sum_{v\notin\tau}
   |\tau\rangle|\tau\cup\{v\}\rangle,
 \qquad |\tau|=d.
 \end{aligned}
 \label{eq:boundary-incidence-states}
\end{equation}
Each family is orthonormal: the second register identifies
$\sigma$ in the first family, and the first identifies $\tau$
in the second. The two states share a basis vector precisely when
$\tau$ is obtained by deleting one vertex from $\sigma$. Hence
\begin{equation}
 \langle\psi_\tau|\phi_\sigma\rangle
 =\begin{cases}
 \xi^{r_\sigma(v)}/\alpha_d,&\tau=\sigma\setminus\{v\},\\
 0,&\text{otherwise}.
 \end{cases}
 \label{eq:boundary-incidence-overlap}
\end{equation}
This overlap is exactly the corresponding boundary matrix element on
all vertex subsets, divided by $\alpha_d$.

\subsection{Block-encoding and restriction to the complex}
\label{sec:boundary-selection}
Introduce preparation unitaries for the states in
Eq.~\eqref{eq:boundary-incidence-states}:
\begin{equation}
 U_R|0\rangle^{\otimes a_0}|\sigma\rangle=|\phi_\sigma\rangle,
 \qquad
 U_L|0\rangle^{\otimes a_0}|\tau\rangle=|\psi_\tau\rangle.
 \label{eq:boundary-preparation-unitaries}
\end{equation}
Clean work registers are suppressed on the right; actions on other input
weights may be arbitrary unitary extensions. The preparation circuits
are given in Section~\ref{sec:boundary-implementation}; we first show how
these unitaries yield the desired block-encoding.
By Eq.~\eqref{eq:boundary-incidence-overlap}, the relevant block of
$U_L^\dagger U_R$ is the boundary on all vertex subsets divided by
$\alpha_d$. To obtain the zero-padded boundary of $K$, we select the
appropriate degrees and source simplices as follows.

Let $\Pi_d^K$ select bit strings of weight $d+1$ that belong
to $K$. Compute the Hamming weight reversibly, test it against
$d+1$, and combine this test with $O_d^K$.
Flipping the acceptance flags makes their all-zero block equal
to $\Pi_d^K$. This uses $O(\log n)$ clean counting work bits,
a constant number of flag bits and oracle calls, and
$O(n\log n)$ elementary gates outside the oracle.
It is a block encoding with ancillas, not a unitary projector
on the data register alone.

Recall that $D_d$, defined before Algorithm~\ref{algo: quantumMayerestimation},
is the Mayer boundary on all $(d+1)$-vertex subsets, whether or not
they belong to $K$, and acts as zero outside its source degree.
Implement its degree selections
on both sides of $U_L^\dagger U_R$; these require no membership oracle.
Face closure gives
\[
 \Pi_{d-1}^K D_d\Pi_d^K=D_d\Pi_d^K=\partial_d.
\]
Thus only the source membership selection $\Pi_d^K$ is needed.
Compose it before the degree-selected incidence encoding, using
the product construction of Lemma~\ref{lemma: product} in
Appendix~\ref{sec: blockencodingrecipes}. Here $U_L$ and $U_R$ act on
the same registers to form $U_L^\dagger U_R$, while each selection
uses a separate block-encoding ancilla register, distinct from those
of the incidence encoding and the other selections. All factors share
the data register, and identity operations on unused ancillas are implicit.
The resulting unitary $U$ has
\begin{equation}
 (\langle0|^{\otimes a}\otimes I)U
 (|0\rangle^{\otimes a}\otimes I)
 =\frac{\partial_d}{\alpha_d}.
 \label{eq:direct-boundary-encoding}
\end{equation}
In particular, $U^\dagger$ encodes
$\partial_d^\dagger/\alpha_d$, including the conjugate
cyclotomic phases automatically. No Hermitian Dirac operator
or replacement of its phase gates is required.

\subsection{State-preparation circuits and gate complexity}
\label{sec:boundary-implementation}
We now analyze the gate complexity of the preparation unitaries in
Eq.~\eqref{eq:boundary-preparation-unitaries} by describing their circuits,
and combine their costs with those of the selections in
Section~\ref{sec:boundary-selection}.
Let $S$ denote the input subset register, $T$ an additional $n$-qubit
register, $R$ the rank register, and $V$ the vertex-index register.
Registers $T,R,V$ start in zero states. We display the intermediate
states for $U_R$, writing $v_r$ for the $r$-th occupied vertex of
$\sigma$ and suppressing clean scan counters. Both preparations follow
four steps, with ranks counted from zero:
\begin{enumerate}
\item \textbf{Select a vertex.}
For $U_R$, prepare a uniform superposition of ranks
$r\in\{0,\ldots,d\}$ and reversibly locate the $r$-th occupied
vertex $v$ of $\sigma$. For $U_L$, use
$r\in\{0,\ldots,n-d-1\}$ and locate the $r$-th unoccupied
vertex of $\tau$. Retain the vertex index and uncompute the scan counters.
For $U_R$, the state is
\[
 \frac{1}{\sqrt{d+1}}\sum_{r=0}^{d}
 |0^n\rangle_T|\sigma\rangle_S|r\rangle_R|v_r\rangle_V.
\]

\item \textbf{Apply the phase and erase the rank.}
For $U_R$, apply the phase $\xi^r$ using rotations controlled by the
bits of $r$; $U_L$ requires no phase. Before changing the input string,
uncompute the rank register from the input and the retained vertex index.
The state for $U_R$ becomes
\[
 \frac{1}{\sqrt{d+1}}\sum_{r=0}^{d}
 \xi^r|0^n\rangle_T|\sigma\rangle_S|0\rangle_R|v_r\rangle_V.
\]

\item \textbf{Form the two subset registers.}
Copy the input from $S$ to $T$. For $U_R$, clear bit $v$ in $T$, giving
\[
 \frac{1}{\sqrt{d+1}}\sum_{r=0}^{d}
 \xi^r|\sigma\setminus\{v_r\}\rangle_T|\sigma\rangle_S
 |0\rangle_R|v_r\rangle_V.
\]
For $U_L$, instead set bit $v$ in $S$, giving the subset registers
$|\tau\rangle_T|\tau\cup\{v\}\rangle_S$.

\item \textbf{Erase the vertex index.}
The two output strings differ at exactly the selected vertex $v$.
Use this differing position to uncompute the vertex index and return
all remaining scan work to zero. The resulting state for $U_R$ is
\[
 |\phi_\sigma\rangle_{TS}|0\rangle_R|0\rangle_V.
\]
For $U_L$, the same index-erasure step gives
\[
 \frac{1}{\sqrt{n-d}}\sum_{v\notin\tau}
 |\tau\rangle_T|\tau\cup\{v\}\rangle_S|0\rangle_R|0\rangle_V
 =|\psi_\tau\rangle_{TS}|0\rangle_R|0\rangle_V.
\]
These are the states in Eq.~\eqref{eq:boundary-incidence-states},
with all work registers clean.
\end{enumerate}

\paragraph{Gate complexity.}
The reversible scans in steps 1, 2, and 4 use a running count and a
vertex index, each of $O(\log n)$ bits, and $O(n\log n)$ gates in total.
Uniform rank preparation in step 1 has at most $n$ outcomes and can be
implemented by controlled rotations with $O(n\,\operatorname{polylog}n)$
gates. The phase in step 2 uses $O(\log n)$ controlled rotations, and
copying and modifying the subset registers in step 3 costs at most
$O(n\log n)$ gates. Adding the degree and membership selections of
Section~\ref{sec:boundary-selection} gives the following total cost.

Consequently the full block-encoding in
Eq.~\eqref{eq:direct-boundary-encoding} uses $a=n+O(\log n)$ ancillas,
$O(1)$ calls to $O_d^K$, and
$O(n\,\operatorname{polylog}n)$ gates when arbitrary rotations
are elementary. With rotations synthesized to obtain absolute
block error $\epsilon$, the non-oracle gate cost is
$\widetilde O(n)$, hiding logarithms in $n,\alpha_d/\epsilon$
(and phase-description precision).
The same bounds hold for controlled versions and adjoints.
All rotation errors are included: normalized block error
$\nu$ corresponds to absolute error $\alpha_d\nu$.

This is already an $\alpha_d$-normalized encoding, so
singular-value amplification and its endpoint-slack condition
are unnecessary. Since $\alpha_d\leq n$, it also satisfies
the looser bound $\widetilde O(n^2/\alpha_d)$ used in the
main complexity statements. If a membership call costs
$C_K$ gates, its contribution is $O(C_K)$ per boundary
encoding and must be multiplied by all subsequent call counts.
We use the unaugmented convention $\partial_0=0$;
boundary maps outside the chain degrees are omitted.
\par
\endgroup

\subsection{Alternative LCU construction for resource estimation}
\label{sec:boundary-alternative-lcu}
We present an alternative linear-combination-of-unitaries (LCU)
construction, used in the resource estimation in
Appendix~\ref{sec: realisticresourceestimation}.
Let $P=\operatorname{diag}(1,\xi)$ and
$a=\begin{pmatrix}0&1\\0&0\end{pmatrix}$.
In this alternative approach, we directly rely on the relation
between the Mayer boundary and the anyonic operators in
Appendix~\ref{sec: anyonic}:
\begin{align}
    D = \sum_{j=1}^n a_j =  \sum_{j=1}^n P_1 \otimes P_2 \otimes \cdots \otimes P_{j-1} \otimes a \otimes \Ibb \otimes \cdots \otimes \Ibb
\end{align}
Each $P$ is a fractional rotation gate, which is already a unitary, so it block-encodes itself (or it is a $(1,0,0)$-encoding of itself). The matrix $a$ can be decomposed as:
\begin{align}
    a = \frac{\sigma_x + i \sigma_y}{2}
\end{align}
Each of the gates $\sigma_x,\sigma_y$ is unitary, so we can use Lemma \ref{lemma: linearcombination} to construct the $(1,1,0)$-encoding of $a$. Then we can use the result of \cite{camps2020approximate}, which uses further $\mathcal{O}(1)$ gates to construct the $(1,1,0)$-encoding of
$$ P_1 \otimes P_2 \otimes \cdots \otimes P_{j-1} \otimes a \otimes \Ibb \otimes \cdots \otimes \Ibb$$
The complexity of this step is $\mathcal{O}(1)$. Then we can use Lemma \ref{lemma: linearcombination} to construct the $(1, 2 ,0)$-encoding of:
\begin{align}
    \frac{1}{n} \sum_{j=1}^n P_1 \otimes P_2 \otimes \cdots \otimes P_{j-1} \otimes a \otimes \Ibb \otimes \cdots \otimes \Ibb = \frac{D}{n}
\end{align}
The gate complexity of this step is $\mathcal{O}(n)$ if we count a general rotation gate as a single gate. Otherwise, we will have a factor for the Toffoli gates in implementing rotations. However, this only incurs the gate complexity by a factor, and thus the overall gate complexity is the same (asymptotically).  We note that the boundary operator $D$ is unrestricted, containing all orders, i.e., $D = \sum_{k} D_k$. In order to get the dimension $d$ of interest, we need to project it to the relevant dimension. Given that the $(1,0,0)$-encoding of $\ket{d}\bra{d} \otimes \Ibb$ can be achieved, then we can use Lemma \ref{lemma: product} to construct the $(n, 3, 0)$-encoding of:
\begin{align}
    D ( \ket{d}\bra{d} \otimes \Ibb) = D_d
\end{align}
This $D_d$ includes all the $d$-simplices, and we need to obtain the $\partial_d$ which is $D_d$ restricted to those $d$-simplices present in $K$. To achieve this, we point out that the oracle $O_d^K$ acts as follows:
\begin{align}
    O_d^K \ket{0}\ket{\sigma_d} = \ket{1}\ket{\sigma_d}
\end{align}
if $\sigma_d \in K$. Therefore, $O_d^K$ admits the following unitary representation:
\begin{align}
    O_d^K = \sum_{\sigma_d \in K} X \otimes \ket{\sigma_d}\bra{\sigma_d} + \sum_{\sigma_d \notin K } \Ibb \otimes \ket{\sigma_d}\bra{\sigma_d}
\end{align}
One can verify that:
\begin{align}
\begin{split}
    (O_d^K)^\dagger O_d^K &=  \big( \sum_{\sigma_d \in K} X \otimes \ket{\sigma_d}\bra{\sigma_d} + \sum_{\sigma_d \notin K } \Ibb \otimes \ket{\sigma_d}\bra{\sigma_d} \big)^\dagger \big( \sum_{\sigma_d \in K} X \otimes \ket{\sigma_d}\bra{\sigma_d} + \sum_{\sigma_d \notin K } \Ibb \otimes \ket{\sigma_d}\bra{\sigma_d} \big) \\
    &=  \sum_{\sigma_d \in K} X^\dagger X \otimes \ket{\sigma_d}\bra{\sigma_d} + \sum_{\sigma_d \notin K } \Ibb \otimes \ket{\sigma_d}\bra{\sigma_d} \\
    &= \Ibb_{n+1}
\end{split}
\end{align}
We have:
\begin{align}
\begin{split}
    O_d^K &= \sum_{\sigma_d \in K} X \otimes \ket{\sigma_d}\bra{\sigma_d} + \sum_{\sigma_d \notin K } \Ibb \otimes \ket{\sigma_d}\bra{\sigma_d}\\ 
    &= \sum_{\sigma_d \in K} (\ket{0}\bra{1} + \ket{1}\bra{0}) \otimes \ket{\sigma_d}\bra{\sigma_d} + \sum_{\sigma_d \notin K } (\ket{0}\bra{0} + \ket{1}\bra{1}) \otimes \ket{\sigma_d}\bra{\sigma_d} 
\end{split}
\end{align}
Applying the gate $X \otimes \Ibb$ to the above oracle, we obtain:
\begin{align}
   O_d^K  (X \otimes \Ibb)  =  \sum_{\sigma_d \in K} (\ket{0}\bra{0} + \ket{1}\bra{1}) \otimes \ket{\sigma_d}\bra{\sigma_d} + \sum_{\sigma_d \notin K } (\ket{0}\bra{1} + \ket{1}\bra{0})\otimes \ket{\sigma_d}\bra{\sigma_d}
\end{align}
which is exactly a $(1,1,0)$-encoding of $\sum_{\sigma_d \in K}   \ket{\sigma_d}\bra{\sigma_d}$. Using a similar procedure, we can obtain a $(1,1,0)$-encoding of $\sum_{\sigma_{d-1} \in K}   \ket{\sigma_{d-1}}\bra{\sigma_{d-1}}$.

As the next step, we use Lemma \ref{lemma: product} to construct the $(n, 5, 0)$-encoding of:
\begin{align}
    \sum_{\sigma_{d-1} \in K}   \ket{\sigma_{d-1}}\bra{\sigma_{d-1}} D_d\sum_{\sigma_d \in K}   \ket{\sigma_d}\bra{\sigma_d} = \partial_d
\end{align}
Given that the complexity in obtaining the $(1,3,0)$-encoding of $\partial_d/n$ is $\mathcal{O}(n)$, the complexity of obtaining the $(n,5,0)$-encoding of $\partial_d$ is $\mathcal{O}(n)$, plus $\mathcal{O}(1)$ uses of the oracle $O_d^K, O_{d-1}^K$.

\subsection{Normalization of the alternative encoding}
\label{sec:boundary-alternative-amplification}
To achieve this, we simply need to apply the following lemma:

\begin{lemma}[\cite{gilyen2019quantum} Theorem 30;\label{lemma: amp_amp}]
Let $U$, $\Pi$, $\widetilde{\Pi} \in {\rm End}(\mathcal{H}_U)$ be linear operators on $\mathcal{H}_U$ such that $U$ is a unitary, and $\Pi$, $\widetilde{\Pi}$ are orthogonal projectors.
Let $\gamma>1$ and $\delta,\epsilon \in (0,\frac{1}{2})$.
Suppose that $\widetilde{\Pi}U\Pi=W \Sigma V^\dagger=\sum_{i}\varsigma_i\ket{w_i}\bra{v_i}$ is a singular value decomposition.
Then there is an $m= \mathcal{O} \Big(\frac{\gamma}{\delta}
\log \left(\frac{\gamma}{\epsilon} \right)\Big)$ and an efficiently computable $\Phi\in\mathbb{R}^m$ such that
\begin{equation}
\left(\bra{+}\otimes\widetilde{\Pi}_{\leq\frac{1-\delta}{\gamma}}\right)U_\Phi \left(\ket{+}\otimes\Pi_{\leq\frac{1-\delta}{\gamma}}\right)=\sum_{i\colon\varsigma_i\leq \frac{1-\delta}{\gamma} }\tilde{\varsigma}_i\ket{w_i}\bra{v_i} , \text{ where } \Big|\!\Big|\frac{\tilde{\varsigma}_i}{\gamma\varsigma_i}-1 \Big|\!\Big|\leq \epsilon.
\end{equation}
Moreover, $U_\Phi$ can be implemented using a single ancilla qubit with $m$ uses of $U$ and $U^\dagger$, $m$ uses of C$_\Pi$NOT and $m$ uses of C$_{\widetilde{\Pi}}$NOT gates and $m$ single qubit gates.
Here,
\begin{itemize}
\item C$_\Pi$NOT$:=X \otimes \Pi + I \otimes (I - \Pi)$ and a similar definition for C$_{\widetilde{\Pi}}$NOT; see Definition 2 in \cite{gilyen2019quantum},
\item $U_\Phi$: alternating phase modulation sequence; see Definition 15 in \cite{gilyen2019quantum},
\item $\Pi_{\leq \delta}$, $\widetilde{\Pi}_{\leq \delta}$: singular value threshold projectors; see Definition 24 in \cite{gilyen2019quantum}.
\end{itemize}
\end{lemma}

Taking the $( n, 5, 0)$-encoding of $\partial_d$, using the Lemma above with $\gamma = \frac{1}{ \sqrt{(d+1) (n-d) }} n$, we thus obtain the $(\sqrt{(d+1) (n-d) }, 6, \epsilon)$-encoding of $\partial_d$. Since the block-encoding of $\partial_d/n$ has complexity $\mathcal{O}(n)$, the complexity for block-encoding $ \partial_d/\sqrt{(d+1)(n-d) } $ is
$$\mathcal{O}\left(n^2  \frac{1}{ \sqrt{(d+1)(n-d) }} \log \left(\frac{n}{\sqrt{(d+1)(n-d) } \epsilon}\right)\right) = \mathcal{\tilde{O}}\left(  \frac{1}{ \sqrt{(d+1)(n-d) }} n^2 \log \frac{1}{\epsilon} \right)$$
where again $\mathcal{\tilde{O}}(.)$ hides the (poly)logarithmic factor.

\section{ Complexity analysis of Algorithm \ref{algo: quantumMayerestimation}}
\label{sec: comlexityanalysis}
\begingroup
\color{black}
The total cost consists of the cost per state-preparation and filter
call multiplied by the number of amplitude-estimation repetitions.
For $f=|S_d^K|>0$ and $w=\beta_{d,p}/f>0$, the accounting is
\begin{equation}
 \begin{aligned}
 G_{\rm nonoracle}&=\widetilde O\bigl(R_{\rm AE}
                  (T_{\rm prep}+T_{\rm filt})\bigr),\\
 Q_{\rm total}&=\widetilde O\bigl(R_{\rm AE}
                  (Q_{\rm prep}+Q_{\rm filt})\bigr),\\
 R_{\rm AE}&=\widetilde O\left(\frac1{\delta\sqrt w}\right).
 \end{aligned}
 \label{eq:mayer-cost-structure}
\end{equation}
Here $T$ counts elementary gates outside membership oracles, $Q$ counts
membership-oracle calls, and $R_{\rm AE}$ counts repetitions required
for relative error $\delta$. Costs include controlled and inverse calls
and the associated ancilla reflections, up to constant factors.
We first bound the cost per call in Section~\ref{sec:mayer-per-call},
then derive the repetition count in Section~\ref{sec:mayer-repetitions},
and combine them in Section~\ref{sec:mayer-total-cost}.
Throughout this appendix, logarithms of inverse precision and
success probability are suppressed in $\widetilde O$.

\subsection{Cost per preparation and filter call}
\label{sec:mayer-per-call}
\paragraph{Boundary and Laplacian encodings.}
A boundary encoding with normalization $\alpha_k$ and normalized
error $\nu$ has absolute error $\alpha_k\nu$.
Lemma~\ref{lemma: Nboundaryoperator} supplies it with
$\widetilde O(n)$ non-oracle gates, $O(1)$ membership queries,
and $a=O(n)$ ancillas.
We use block-encodings of $D_k$, the Mayer boundary on all
$(k+1)$-vertex subsets defined before
Algorithm~\ref{algo: quantumMayerestimation}, with degree selection at
every factor and membership selection only at the source of each power.
Indeed, for $1\leq q\leq s$,
\begin{equation}
 \partial_s^q=D_{s-q+1}\cdots D_s\Pi_s^K,
 \label{eq:boundary-power-source-selection}
\end{equation}
by face closure. Thus, encoding each of the lower and upper Laplacian
terms $H_-$ and $H_+$ in step 3 of
Algorithm~\ref{algo: quantumMayerestimation} requires only $O(1)$
membership queries, using $O_d^K$ and $O_{d+N-p}^K$, respectively.

For actual normalized encoded blocks $\widetilde A_j$ and
ideal blocks $A_j$, we have $\|\widetilde A_j\|,\|A_j\|\leq1$.
Expanding the difference of the products by replacing one factor at a
time and applying the triangle inequality gives
\[
 \left\|\prod_{j=1}^{m}\widetilde A_j-\prod_{j=1}^{m}A_j\right\|
 \leq\sum_{j=1}^{m}\|\widetilde A_j-A_j\|.
\]
Thus a product with normalization $\prod_j\alpha_j$ has
absolute error at most
$(\prod_j\alpha_j)\sum_j(\epsilon_j/\alpha_j)$,
where $\epsilon_j$ is the absolute error of factor $j$.
This proves the $2p\,a_-\nu$ and $2(N-p)a_+\nu$ bounds
in Algorithm~\ref{algo: quantumMayerestimation}.
Their linear combination has normalized error at most $2N\nu$.

One Laplacian-encoding call uses $O(N)$ boundary calls,
$\widetilde O(Nn)$ non-oracle gates, $O(1)$ membership
queries, and $O(Nn)$ block ancillas. Degree and membership
selection and controlled operations are included.
For compatibility with the coarser main-text notation,
its gate cost is also bounded by
\[
 \widetilde O\left(n^2
 \left(\sum_{i=0}^{p-1}\frac1{\alpha_{d-i}}
       +\sum_{i=1}^{N-p}\frac1{\alpha_{d+i}}\right)\right),
\]
since $\alpha_k\leq n$.
Omit all factors belonging to a vanishing boundary term.

\paragraph{Kernel-filter cost.}
For the filter in step 5 of Algorithm~\ref{algo: quantumMayerestimation},
choose an even polynomial $P$, bounded by $1/2$ on $[-1,1]$,
that approximates one half of the kernel projector: $P(0)$ is within
$\eta/2$ of $1/2$, and $|P(x)|\leq\eta/2$ on the nonzero spectrum
of $\Delta_{d,p}/\alpha$.
It has degree
\[
 \mathscr D_P=O\left(\frac{\alpha}{\gamma_d^p}
                         \log\frac2\eta\right).
\]
Lemma~\ref{lemma: qsvt} contributes at most
$4\mathscr D_P\sqrt{2N\nu}$ from the input error. Taking
\[
 \nu\leq\frac1{2N}\left(\frac{\eta}{8\mathscr D_P}\right)^2
\]
therefore yields a block $Q$ with $\|Q-\Pi_0/2\|\leq\eta$.
Here we multiply the filter block on both sides by $\Pi_d^K$
and denote the selected block by $Q$, so $Q=\Pi_d^K Q\Pi_d^K$.
Thus basis states outside $C_d^K$, on which the extended Laplacian
acts as zero, do not contribute to the estimated kernel dimension.
These two selections add only a constant number of membership
calls per filter call.
Rotation-synthesis errors can be included by reducing the other error
tolerances by a constant factor. The resulting precision required for
each boundary block-encoding changes only logarithmic factors in its
gate cost. The filter cost is
\[
 T_{\rm filt}=\widetilde O(Nn\alpha/\gamma_d^p),\qquad
 Q_{\rm filt}=\widetilde O(\alpha/\gamma_d^p),
\]
where $T$ counts gates outside membership oracles and $Q$
counts their calls.
The reference register is acted on by the identity and requires
no additional filter gates.

\paragraph{State-preparation cost.}
With $f=|S_d^K|>0$, the simplex-state preparation of
Lemma~\ref{lemma: uniformsuperposition}, followed by the $n$ CNOT gates
that produce $|\Omega\rangle$ in Eq.~\eqref{eq:mayer-correlated-simplex-state},
costs
\[
 T_{\rm prep}=\widetilde O\left(
 dn\sqrt{\frac{\binom n{d+1}}{f}}\right),\qquad
 Q_{\rm prep}=\widetilde O\left(
 \sqrt{\frac{\binom n{d+1}}{f}}\right).
\]
Here $T_{\rm prep}$ counts elementary gates outside the membership
oracles, including the copying CNOT gates, and $Q_{\rm prep}$ counts
membership-oracle calls for one coherent preparation of
$|\Omega\rangle$. These are per-preparation costs, before the repetitions
required by amplitude estimation; inverse preparation has the same costs.
We take $d\geq1$ as in the main theorem; for $d=0$, use
$(d+1)n$ instead of $dn$.
The preparation error includes the flag and all work registers,
not just the first data register.

\par
\endgroup

\begingroup
\color{black}
\subsection{Repetitions required for relative error}
\label{sec:mayer-repetitions}
\paragraph{Success probability and preparation errors.}
Write $f=|S_d^K|>0$, $w=\beta_{d,p}/f$, and let $\Pi_0$ project
onto $\ker\Delta_{d,p}$. Let $Q$ be the actual encoded block, with
$\|Q-\Pi_0/c\|\leq\eta_Q$, where $c\geq1$ is a known constant.
For the target state $|\Omega\rangle$ in
Eq.~\eqref{eq:mayer-correlated-simplex-state}, let the normalized
coherently prepared state satisfy
$\||\widetilde\Omega\rangle-|\Omega\rangle\|\leq\eta_\Omega$.
The error $\eta_Q$ includes both filter approximation and
block-encoding implementation errors. The zero block-ancilla event has
probability
\begin{equation}
 q=\|(Q\otimes I)|\widetilde\Omega\rangle\|^2,\qquad
 |q-w/c^2|\leq2(\eta_Q+\eta_\Omega).
 \label{eq:mayer-readout-probability}
\end{equation}
Indeed, $\|Q\|,\|\Pi_0/c\|\leq1$, and the vectors
$(Q\otimes I)|\widetilde\Omega\rangle$ and
$(\Pi_0/c\otimes I)|\Omega\rangle$ differ in norm by at most
$\eta_Q+\eta_\Omega$.

\paragraph{Amplitude estimation.}
Amplitude estimation of this event, using the preparation circuit,
its inverse, and the zero-ancilla reflection, estimates $\sqrt q$
to additive error $\tau$ using $O(1/\tau)$ calls at constant
success probability. Repetition boosts the success probability
with logarithmic overhead.
Choose $\eta_Q+\eta_\Omega\leq\tau^2/8$. Then
$|\sqrt q-\sqrt w/c|\leq\tau/2$, so the estimate $\widehat a$
obeys $|\widehat a-\sqrt w/c|\leq2\tau$.
Consequently, $\widehat w=c^2\widehat a^2$ satisfies
\begin{equation}
 |\widehat w-w|\leq4c\tau\sqrt w+4c^2\tau^2.
 \label{eq:mayer-readout-error}
\end{equation}
For $0<\delta<1$ and $w>0$, choosing
$\tau\leq\delta\sqrt w/(16c)$ suffices for relative error $\delta$.
The amplitude precision is proportional to $\delta\sqrt w$, not
$\sqrt w/\delta$. The stricter systematic-error tolerance above is
sufficient and affects only logarithmic precision factors in the
preceding constructions.

\paragraph{Unknown Betti number and the zero case.}
The unknown $w$ need not be supplied: halve $\tau$ successively,
use the corresponding systematic-error tolerance, and stop when the
amplitude confidence interval is sufficiently narrow relative to
its positive lower endpoint. Allocating failure probabilities across
rounds gives $\widetilde O(1/(\delta\sqrt w))$ calls for $w>0$.
A terminating test for $\beta_{d,p}=0$ additionally assumes
$w=0$ or $w\geq w_{\min}>0$ and stops refining the amplitude precision
once it reaches $O(\delta\sqrt{w_{\min}})$; no relative-error guarantee at zero
is intended. Integrality supplies the universal, possibly
exponentially small threshold $w_{\min}=1/\binom{n}{d+1}$.

\subsection{Total gate and query complexity}
\label{sec:mayer-total-cost}
We now combine the per-call costs in Section~\ref{sec:mayer-per-call}
with the repetition count in Section~\ref{sec:mayer-repetitions},
as outlined in Eq.~\eqref{eq:mayer-cost-structure}.
If $T_{\rm filt}$ and $T_{\rm prep}$ denote the costs of a filter
call and a coherent state preparation at the required precision,
including the reflections used in amplitude estimation, the total cost
of estimating $w$ to relative error $\delta$ is
\begin{equation}
 \widetilde O\left(
 \frac{T_{\rm filt}+T_{\rm prep}}{\delta\sqrt w}\right).
 \label{eq:mayer-readout-cost}
\end{equation}
The factor $\delta^{-1}\sqrt{f/\beta_{d,p}}$ counts the amplitude-estimation
repetitions; each repetition uses the filter and state-preparation
circuits whose costs were established above.
The corresponding total membership-query count is
\begin{equation}
 Q_{\rm total}
 =\widetilde O\left[
 \left(\frac{\alpha}{\gamma_d^p}
       +\sqrt{\frac{\binom n{d+1}}{f}}\right)
       \frac1{\delta\sqrt w}\right].
 \label{eq:mayer-total-queries}
\end{equation}
If each membership oracle, including controlled use and inverse,
costs at most $C_K$ gates, the full gate bound is
\begin{equation}
 G_{\rm total}\leq G_{\rm nonoracle}+C_K Q_{\rm total},
 \label{eq:gate-query-accounting}
\end{equation}
where $G_{\rm nonoracle}$ counts gates outside the oracles.
For a flag (clique) complex with a fixed classical edge list,
Eq.~\eqref{eq:flag-membership-cost} gives
$C_K=\widetilde O(|E|+n)$ for the explicit membership circuit.
Theorem~\ref{thm: mainresult} bounds their sum with the query count,
$G_{\rm nonoracle}+Q_{\rm total}$, in the unit-cost oracle model.
Adding the query count does not change the asymptotic scaling of the
displayed main-text upper bound. A known zero Laplacian needs neither
filtering nor inversion.

The direct output is $\widehat w$. An unnormalized estimate
$f\widehat w$ also requires $f$ to be known, or estimated to compatible
relative accuracy with its cost included.
\par
\endgroup


\section{Block-encoding the persistent Mayer Laplacian $ \Delta_{d,p}^{a,b}$}
\label{sec: blockencodingpersistentMayerLaplacian}
\begingroup
\color{black}
We prove Lemma~\ref{lemma: blockencodingpersistentLaplacian}
using the Schur-complement construction of
\cite{hayakawa2022quantum,memoli2022persistent}.
Write
\[
 H_-=(\partial_d^{p,(a)})^\dagger\partial_d^{p,(a)},\qquad
 \Gamma=\partial_{d+N-p}^{N-p,(b)}
             (\partial_{d+N-p}^{N-p,(b)})^\dagger.
\]
The normalization factors are
$a_-=\prod_{i=0}^{p-1}\alpha_{d-i}^2$ and
$a_+=\prod_{i=1}^{N-p}\alpha_{d+i}^2$.
If $d<p$, omit $H_-$ and set $a_-=0$.
All blocks below act on the common degree-$d$ register and are defined
to be zero outside their respective input and output subspaces.
Let $P_a$ select $C_d^{K_a}$ and $P_b$
select $C_d^{K_b}$, and put $P_\perp=P_b-P_a$. Then
\begin{equation}
 \begin{aligned}
 \Delta_1&=P_a\Gamma P_a,&
 \Delta_2&=P_a\Gamma P_\perp,\\
 \Delta_3&=P_\perp\Gamma P_a,&
 \Delta_4&=P_\perp\Gamma P_\perp.
 \end{aligned}
 \label{eq:persistent-selected-blocks}
\end{equation}
Membership tests implement these selections. In particular, each
$\Delta_i$ inherits normalization $a_+$ and one use of the
$\Gamma$ encoding, plus selection costs.
By Eq.~\eqref{eq:boundary-power-source-selection}, the lower power
uses only source selection $\Pi_d^{K_a}$, while the upper power
uses only $\Pi_{d+N-p}^{K_b}$. Together with $P_a,P_b$, this requires
only $O_d^{K_a},O_d^{K_b},O_{d+N-p}^{K_b}$, omitting the last
when its degree exceeds $n-1$. No lower-degree or intermediate-degree
membership queries are needed.
On $C_d^{K_a}$ the persistent Laplacian is
\[
 \Delta_{d,p}^{a,b}
 =H_-+\Delta_1-\Delta_2\Delta_4^+\Delta_3.
\]

\subsection*{Normalized gap and inverse}
For $\Delta_4\ne0$, choose a supplied lower bound
\begin{equation}
 0<\gamma_\Delta\leq
 \min\left\{\frac12,\lambda_{\min}^{+}(\Delta_4/a_+)\right\}.
 \label{eq:persistent-inverse-gap}
\end{equation}
The superscript $+$ means the smallest strictly positive eigenvalue.
This is the gap of the block $\Delta_4/a_+$ that is inverted in the
Schur-complement formula, not that of
the Schur complement or of the persistent Laplacian.

We use the odd-polynomial pseudoinverse construction and its
robustness bound from \cite{gilyen2019quantum}
(Theorem~41 and Lemma~22), in the following normalized form.
\begin{lemma}[Normalized pseudoinverse encoding]
\label{lemma: invert}
Suppose $A\geq0$, $\|A\|\leq a$, and
$\operatorname{spec}(A/a)\subseteq\{0\}\cup[\gamma,1]$,
where $0<\gamma\leq1/2$. A normalized input-block error $e$
allows an encoded block $Z$ satisfying
\[
 \left\|Z-\frac{\gamma a}{2}A^+\right\|
 \leq \rho_{\rm pol}+8D\sqrt e,\qquad
 D=O\left(\frac1\gamma\log\frac{2}{\gamma\rho_{\rm pol}}\right).
\]
It uses $O(D)$ calls to the input encoding and its inverse,
with a constant number of additional ancillas.
\end{lemma}
\begin{proof}
An odd polynomial bounded by one approximates $\gamma/(2x)$
on $[\gamma,1]$ and vanishes at zero. Real singular value
transformation applied to the adjoint encoding gives
$(\gamma/2)(A/a)^+=(\gamma a/2)A^+$.
The polynomial approximation and input perturbation errors add;
the displayed constant in the robustness bound is conservative.
\end{proof}

Applied to $A=\Delta_4$, this gives normalization
\begin{equation}
 b_{\rm inv}=\frac{2}{\gamma_\Delta a_+},\qquad
 Z\approx\Delta_4^+/b_{\rm inv}.
 \label{eq:persistent-inverse-normalization}
\end{equation}
A normalized error $\rho$ in $Z$ is an absolute
pseudoinverse-encoding error $b_{\rm inv}\rho$.

\subsection*{Products and error bounds}
Multiplying the encodings of $\Delta_2$, $\Delta_4^+$,
and $\Delta_3$ from Eqs.~\eqref{eq:persistent-selected-blocks}
and \eqref{eq:persistent-inverse-normalization}, using
Lemma~\ref{lemma: product}, gives normalization
$a_+^2b_{\rm inv}=2a_+/\gamma_\Delta$.
A signed linear combination using Lemma~\ref{lemma: linearcombination}
therefore gives
\begin{equation}
 \alpha'=a_-+a_++\frac{2a_+}{\gamma_\Delta}
 \label{eq:persistent-alpha}
\end{equation}
for $\Delta_{d,p}^{a,b}$.

To track errors, work with the normalized encoded blocks.
Let $e_-$ bound the error for $H_-/a_-$, let $e$ bound
the errors for all $\Delta_i/a_+$, and let $\rho$ bound
the inverse-block error in Eq.~\eqref{eq:persistent-inverse-normalization}.
All actual encoded blocks have operator norm at most one; the ideal
normalized inverse has norm at most $1/2$. Replacing the three factors
one at a time and applying the triangle inequality gives error at most
$2e+\rho$. Thus the absolute error
of the persistent-Laplacian encoding is at most
\begin{equation}
 a_-e_-+a_+e+\frac{2a_+}{\gamma_\Delta}(2e+\rho).
 \label{eq:persistent-error-bound}
\end{equation}
For a desired normalized error $h$, choose
$e_-\leq h/3$, $\rho_{\rm pol}\leq h/6$, and
\[
 e\leq\min\left\{\frac h3,\left(\frac{h}{48D}\right)^2\right\}.
\]
Then $\rho\leq h/3$ and
Eq.~\eqref{eq:persistent-error-bound} is at most $\alpha'h$.
For an absolute target error $\epsilon$, take $h=\epsilon/\alpha'$.
Errors in the selection and linear-combination circuits can be
included by reducing these error tolerances by a constant factor.
This replaces the need to solve a nonlinear error equation;
the required input precision is inverse polynomial in the displayed
parameters, and hence contributes only logarithmic overhead when
the costs of the individual block-encodings depend logarithmically on precision.

\subsection*{Cost of the persistent-Laplacian encoding}
Let $T_-$ be the cost of encoding $H_-/a_-$ and let $T_+$
be the cost of encoding $\Gamma/a_+$, including the selections
in Eq.~\eqref{eq:persistent-selected-blocks}. Include the ancilla
reflections, controlled calls, and linear-combination overhead
in these per-call costs. At the precisions just specified, the cost
of one call to the persistent-Laplacian block-encoding is
\begin{equation}
 \widetilde O\left(T_-+\frac{T_+}{\gamma_\Delta}\right).
 \label{eq:persistent-block-cost}
\end{equation}
The same accounting applies to membership queries:
if encoding $H_-/a_-$ and $\Gamma/a_+$ requires $Q_-$ and $Q_+$
membership queries, respectively,
the count is $\widetilde O(Q_-+Q_+/\gamma_\Delta)$,
not a constant for the full construction.
Logarithms of inverse errors and gaps are suppressed.

\paragraph{The case $\Delta_4=0$.}
If $\Delta_4$ is known to be zero, positivity of $\Gamma$ implies
$\Delta_2=\Delta_3=0$. The construction then omits the inverse,
uses $\alpha'=a_-+a_+$, and costs
$\widetilde O(T_-+T_+)$. No inverse-gap promise is needed.
An upper boundary that vanishes by degree is omitted altogether.
If the entire persistent Laplacian is known to vanish, its
normalized kernel dimension is one and no spectral filter is needed.
\par
\endgroup

\section{Complexity of quantum algorithm for estimating persistent Betti number $\beta_{d,p}^{a,b}$}
\label{sec: complexitypersistentBettinumber}
\begingroup
\color{black}
Assume the nontrivial case $\Delta_4\ne0$.
Let $\gamma_{d,p}^{a,b}>0$ be a lower bound on the
smallest positive eigenvalue of the unnormalized persistent
Laplacian, independent of the inverse-gap bound $\gamma_\Delta$.
Its normalized gap is $\gamma_{d,p}^{a,b}/\alpha'$.
A polynomial filter approximating one half of the kernel projector,
bounded by $1/2$ on $[-1,1]$, has degree
\[
 D_{\rm filt}
 =O\left(\frac{\alpha'}{\gamma_{d,p}^{a,b}}
             \log\frac{2}{\eta}\right).
\]
Here its polynomial error is at most $\eta/2$.
By Lemma~\ref{lemma: qsvt}, normalized input error
$h\leq(\eta/(8D_{\rm filt}))^2$ makes the implementation
error at most $\eta/2$. The preceding construction supplies
this precision by Eq.~\eqref{eq:persistent-error-bound}.
The resulting encoded block approximates
$\Pi_0^{a,b}/2$ to error $\eta$ after multiplying by $P_a$ on both sides.
This excludes basis states outside $C_d^{K_a}$, on which the extended
operator acts as zero;
the additional selection queries are included in the filter cost.

Let $T_{\rm prep}$ be the cost of preparing the state in
Eq.~\eqref{eq:mayer-correlated-simplex-state} with $K$ replaced by $K_a$,
using Lemma~\ref{lemma: uniformsuperposition} and the copying CNOT gates,
at the required precision, including the reflections used in amplitude estimation.
Put $f=|S_d^{K_a}|>0$ and $w=\beta_{d,p}^{a,b}/f>0$.
Choose the filter and preparation errors according to the
success-probability error bound in Eq.~\eqref{eq:mayer-readout-probability}.
In that analysis, $c$ specifies the target filter block $\Pi_0/c$.
Here the filter approximates $\Pi_0^{a,b}/2$, so $c=2$ and the ideal
success probability is $w/4$. Equation~\eqref{eq:mayer-readout-cost}
therefore gives
\begin{equation}
 \widetilde O\left[
 \left(
 \frac{\alpha'}{\gamma_{d,p}^{a,b}}
       \left(T_-+\frac{T_+}{\gamma_\Delta}\right)
 +T_{\rm prep}\right)\frac{1}{\delta\sqrt w}\right].
 \label{eq:persistent-total-cost}
\end{equation}
The qualifications for testing $\beta_{d,p}^{a,b}=0$ and estimating
the normalization are the same as in the amplitude-estimation analysis
for ordinary Mayer Betti numbers in Appendix~\ref{sec: comlexityanalysis}.
The filtering factor multiplies the block-construction cost:
there is only one explicit $1/\gamma_\Delta$ before
substituting the $\gamma_\Delta$-dependent $\alpha'$.

Lemma~\ref{lemma: Nboundaryoperator} gives the boundary-encoding costs
entering $T_-$ and $T_+$, while Lemma~\ref{lemma: uniformsuperposition}
gives the state-preparation cost $T_{\rm prep}$. The following bounds
count elementary gates outside the membership oracles; all other
boundary-encoding circuit costs are included. Up to logarithmic
precision factors, these costs are
\[
 \begin{aligned}
 T_-&=\widetilde O(pn),&
 T_+&=\widetilde O((N-p)n),\\
 T_{\rm prep}&=\widetilde O\left(
 dn\sqrt{\frac{\binom{n}{d+1}}{f}}\right).
 \end{aligned}
\]
The membership-query counts for encoding $H_-/a_-$ and $\Gamma/a_+$
are both $O(1)$,
including selection queries. Hence the persistent
algorithm uses
\begin{equation}
 \widetilde O\left[
 \left(\frac{\alpha'}{\gamma_{d,p}^{a,b}}
       \left(1+\frac{1}{\gamma_\Delta}\right)
       +\sqrt{\frac{\binom n{d+1}}{f}}\right)
       \frac1{\delta\sqrt w}\right]
 \label{eq:persistent-total-queries}
\end{equation}
membership calls. Multiply this count by the largest
implementation cost of a required membership oracle to
include oracle gates.
The looser bounds
$T_\pm=\widetilde O(n^2\sum_{\pm}1/\alpha_k)$ also hold,
where each sum includes its respective boundary factors. Since
\[
 \alpha'\leq\frac{3(a_-+a_+)}{\gamma_\Delta},
 \qquad
 \sum\frac1{\alpha_k}\leq\frac{N}{\alpha_{\min}},
\]
they yield the coarser main-text bound with
$\gamma_\Delta^{-2}$, for both $d\geq p$ and $d<p$.
Adding the query count in Eq.~\eqref{eq:persistent-total-queries}
does not change this asymptotic upper bound, yielding the unit-cost operation
count in Theorem~\ref{thm: persistent-mainresult}.
In the latter case $a_-=T_-=0$.
Fixed $N$, polynomial costs for the constituent block-encodings and
state preparation, polynomial normalization factors,
inverse-polynomial gaps, and an inverse-polynomial lower bound on
$w=\beta_{d,p}^{a,b}/|S_d^{K_a}|>0$
therefore suffice for polynomial time.
If $\Delta_4=0$, replace the block cost and normalization by
the zero-block expressions above; the inverse-gap factor disappears.
\par
\endgroup

\section{ ANALYSIS OF MAYER BETTI NUMBERS }
\label{sec: analyzingMayerBetti}
As emphasized throughout the work, one of the crucial factors for our quantum algorithm to achieve polynomial scaling in $n$ is the magnitude of Mayer Betti number $\beta_{d,p}$. It needs to be of the same magnitude as the number of simplices $|S_d^K|$ (for which our algorithm achieves the best complexity), or at worst the ratio $\frac{\beta_{d,p} }{ |S_d^K|} $ needs to be $\Omega( \frac{1}{\text{poly } n} )$ (for which our algorithm still has polynomial complexity). While it is known that in the simplicial homology, the Betti numbers are usually small for most complexes, especially in the dense regime, it is interesting that Mayer Betti numbers can be large for a wide range of complexes. In what follows, we attempt to give a rigorous proof for this statement, based on the algebraic nature of the Mayer homology. 

Recall that $\beta^p_d := \dim H_{d,p} = \dim \ker \partial^p - \dim ( \text{Im} \partial^{N-p})$. Since $\partial^p: C_d \longrightarrow C_{d-p}$, $\dim \ker \partial^p \geq |S_d^K| - |S_{d-p}^K|$. Similarly, $\dim ( \text{Im} \partial^{N-p}) \leq |S_{d+N-p}| $, so we have that:
\begin{align}
    \beta_{d,p} \geq |S_d^K| - |S_{d-p}^K| - |S_{d+N-p}^K|
\end{align}
If $ |S_d^K| - |S_{d-p}^K| - |S_{d+N-p}^K| < 0$, then by definition, $\beta_{d,p} \geq 0 > |S_d^K| - |S_{d-p}^K| - |S_{d+N-p}^K|$, which is trivial. Due to this, a more precise mathematical representation of the above inequality is as follows:
\begin{align}
     \beta_{d,p} \geq \max \{  0, |S_d^K| - |S_{d-p}^K| - |S_{d+N-p}^K|\}
\end{align}
\begingroup
\color{black}
For $|S_d^K|>0$, this gives
\begin{equation}
    \frac{\beta_{d,p}}{|S_d^K|}
    \geq 1-\frac{|S_{d-p}^K|}{|S_d^K|}
             -\frac{|S_{d+N-p}^K|}{|S_d^K|}.
    \label{eq:mayer-density-rank-bound}
\end{equation}
In particular, a sufficient condition for a constant normalized signal is
\begin{equation}
    \frac{|S_{d-p}^K|+|S_{d+N-p}^K|}{|S_d^K|}
    \leq 1-c
    \quad\Longrightarrow\quad
    \frac{\beta_{d,p}}{|S_d^K|}\geq c,
    \label{eq:mayer-density-margin}
\end{equation}
where $c>0$ is independent of $n$. Merely bounding the two ratios by
unspecified constants does not ensure this positive margin.

\subsection*{General and sufficient conditions for large (normalized) Mayer Betti numbers}
Fix $N$ and $1\leq p\leq N-1$. For $p\leq d$ and
$d+N-p\leq n-1$, write
\begin{align}
    |S_d^K|&=p_1\binom{n}{d+1},&
    |S_{d-p}^K|&=p_2\binom{n}{d-p+1},&
    |S_{d+N-p}^K|&=p_3\binom{n}{d+N-p+1},
\end{align}
with $0<p_1\leq1$ and $0\leq p_2,p_3\leq1$. These are densities of simplices
at the specified degrees, not densities of maximal simplices. Degrees
outside the chain complex contribute zero to
Eq.~\eqref{eq:mayer-density-rank-bound}.
The binomial ratios satisfy
\begin{align}
    \frac{|S_{d-p}^K|}{|S_d^K|}
    &=\frac{p_2}{p_1}\prod_{j=0}^{p-1}
        \frac{d+1-j}{n-d+j}
    \leq\frac{p_2}{p_1}\left(\frac{d+1}{n-d}\right)^p,\\
    \frac{|S_{d+N-p}^K|}{|S_d^K|}
    &=\frac{p_3}{p_1}\prod_{j=0}^{N-p-1}
        \frac{n-d-1-j}{d+2+j}
    \leq\frac{p_3}{p_1}
        \left(\frac{n-d-1}{d+2}\right)^{N-p}.
\end{align}
To bound the upper-degree contribution by $1/4$, it suffices to impose
\begin{equation}
    \frac{p_3}{p_1}\leq
    \frac14\left(\frac{d+2}{n-d-1}\right)^{N-p}.
    \label{18}
\end{equation}
Keeping this explicit bound, rather than just its asymptotic order,
gives
\begin{equation}
    \frac{\beta_{d,p}}{|S_d^K|}
    \geq\frac34-\frac{p_2}{p_1}
                  \left(\frac{d+1}{n-d}\right)^p.
    \label{eq:mayer-density-quarter-bound}
\end{equation}

\paragraph{Sublinear degrees.}
Suppose $p_2/p_1=O(1)$ and Eq.~\eqref{18} holds. If
$d=\Theta(n^{1-\mu})$ with fixed $0<\mu<1$, the subtracted term in
Eq.~\eqref{eq:mayer-density-quarter-bound} is $O(n^{-p\mu})$.
If $d=\Theta(n/\log n)$, it is $O((\log n)^{-p})$.
In both cases the normalized Mayer Betti number is at least
$3/4-o(1)$, and hence is $\Omega(1)$.
The right-hand side of Eq.~\eqref{18} has order
$n^{-\mu(N-p)}$ and $(\log n)^{-(N-p)}$, respectively, but its
explicit upper bound must still be satisfied.
For $p_2/p_1=\Theta(1)$, the exact lower-degree ratio also gives the
following examples (with $d\geq p$).
\begin{table}[H]
    \centering
   \begin{tabular}{|c|c|}
   \hline
   $d$  &  $ \frac{|S_{d-p}^K|}{|S_d^K|} $ \\
   \hline
$\Theta(1)$  & $ \Theta(n^{-p} )$ \\
\hline
   $\sqrt{n}$  & $\Theta ( n^{-p/2} )$  \\
   \hline
   $ n^{2/3}$ & $\Theta ( n^{-p/3}  )$ \\
   \hline
   $n^{9/10}$ & $\Theta (n^{-p/10} )$ \\
   \hline
\end{tabular}
    \caption{ Asymptotically tight bounds on $ \frac{|S_{d-p}^K|}{|S_d^K|} $ based on concrete values of $d$.  }
    \label{tab:placeholder}
\end{table}

\paragraph{Linear degrees.}
Suppose $d/n\to1/\zeta$ for a fixed $\zeta>2$ and $p_2/p_1=O(1)$.
Then
\begin{equation}
    \frac{p_2}{p_1}\left(\frac{d+1}{n-d}\right)^p
    =\frac{p_2/p_1}{(\zeta-1)^p}+o(1).
\end{equation}
Unlike the sublinear cases, this term need not tend to zero.
Under Eq.~\eqref{18}, a sufficient finite-$n$ condition is
\begin{equation}
    \frac{p_2}{p_1}\left(\frac{d+1}{n-d}\right)^p
    \leq\frac34-c,
    \qquad 0<c<\frac34,
    \label{eq:mayer-linear-density-margin}
\end{equation}
for a constant $c$ independent of $n$. This yields
$\beta_{d,p}/|S_d^K|\geq c$.
Equivalently, a strict asymptotic margin
$\limsup_{n\to\infty}(p_2/p_1)/(\zeta-1)^p<3/4$
is sufficient for some such $c$ and all sufficiently large $n$.
For example, $p_2=p_1$ and $\zeta=3$ give a lower bound
$3/4-2^{-p}-o(1)$, which is positive for every $p\geq1$.
An arbitrary constant value of $p_2/p_1$ does not suffice.
\par
\endgroup

\begingroup
\color{black}
These are sufficient conditions on the simplex counts of a given complex,
not a guarantee that the density parameters can be prescribed independently.
Adding a simplex also requires all of its faces, which can change the counts
at lower degrees. To exhibit a family where the required conditions can
be checked directly, we construct an explicit cone in
Appendix~\ref{sec: provingconecomplex} and establish its normalized Mayer
Betti bounds, spectral-gap bound for $p=N-1$ and $m\equiv1\pmod N$, and efficient
membership test.
\par
\endgroup

\begingroup
\color{black}
If $|S_d^K|/\binom{n}{d+1}$ is inverse-polynomial or larger, the number
of degree-$d$ simplices is superpolynomial but subexponential for
$d=\Theta(n^{1-\mu})$, with fixed $0<\mu<1$, and for
$d=\Theta(n/\log n)$. It is exponential for $d/n\to1/\zeta$ with
fixed $\zeta>2$. Any method that explicitly lists all these simplices
must spend at least $\Omega(|S_d^K|)$ time on that enumeration alone.
In contrast, for fixed $N$, the quantum algorithm has polynomial
scaling under the stated density, normalized-signal, spectral-gap,
and efficient-oracle assumptions, with inverse-polynomial relative
precision. This is a comparison with explicit enumeration-based
methods, not a lower bound against all classical algorithms for
estimating the same Betti number.
\par
\endgroup

In addition, our quantum algorithm achieves the best efficiency when $\frac{\beta_{d,p}}{|S_d^K|} = \Omega(1)$. However, even when this ratio is  $\Omega( \frac{1}{\text{poly(n)}})$, our quantum algorithm, which has complexity scaling as $\sqrt{ \frac{|S_d^K| }{ \beta_{d,p}}}$, still obtains polynomial scaling (under the premise that the gap grows at worst inverse-polynomially). This suggests that for a wider range of complexes other than those satisfying the conditions above, the polynomial-time quantum algorithm is still possible.

\subsection*{Numerical studies}
\begingroup\color{black}
We use the Costa--Farber multi-parameter random simplicial complex
model~\cite{costa2014random}, constructed as follows.
\begin{enumerate}
    \item \textbf{Vertices.} Retain all $n$ vertices ($P_0=1$).
    \item \textbf{Edges.} Include each edge independently with probability $P_1$.
    \item \textbf{Higher-dimensional simplices.} Proceed in increasing
    dimension $i=2,3,\ldots$. Conditional on the lower-dimensional complex,
    include each $i$-simplex whose boundary faces are all present
    independently with probability $P_i$. For example, a triangle can
    be filled only when all three of its edges are present.
\end{enumerate}
The parameters $P_i$ are conditional inclusion probabilities, not the
realized simplex densities $p_1,p_2,p_3$ used in the preceding analysis.\footnote{\textcolor{black}{Setting $P_i=1$ for all $i\geq2$ gives the clique complex of an Erd\H{o}s--R\'enyi random graph (up to any imposed dimension cutoff).}}
Varying the higher-dimensional probabilities allows us to explore
different topological structures even when all edges are present.
\par\endgroup

\textcolor{black}{We numerically examine Mayer Betti numbers and positive
Laplacian gaps for the specified Costa--Farber models and finite system
sizes. These experiments illustrate their dependence on simplex counts
and vertex counts; the general sufficient conditions are established
separately by the preceding analysis.}

\begin{figure}[htbp]
    \centering
    \includegraphics[width=0.7\linewidth]{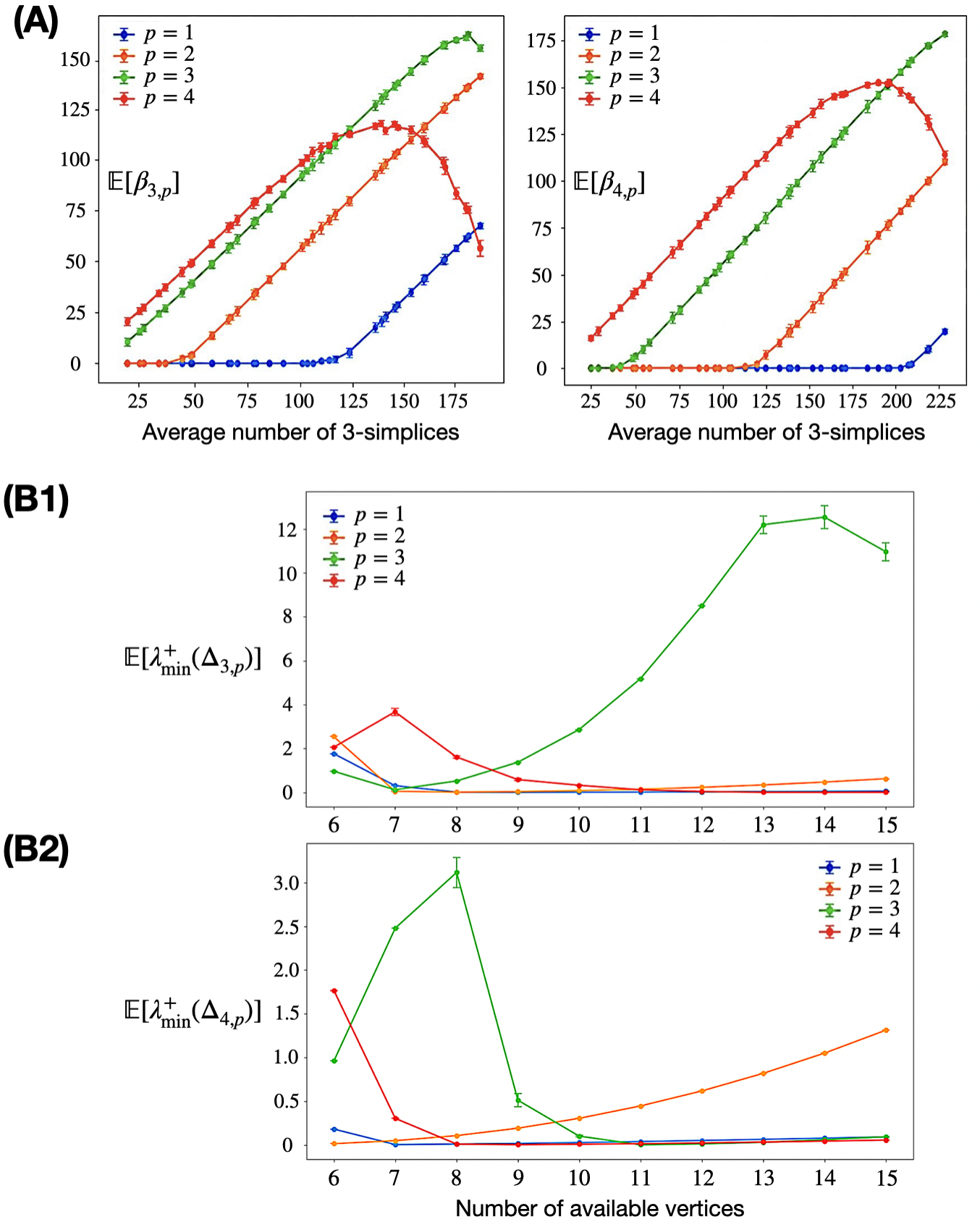}
    \caption{\textbf{(A)} Plot showing the Mayer Betti numbers $\beta_{3,p}, \beta_{4,p} $ versus the number of simplices. In this simulation, we use the Costa-Farber random complex model with the following inclusion probability: $P_1 = P_2 = 1.0$, and for all $i=3,...,n$, we have $P_i$ varied from $0.1$ to $0.90$ with a step size of $\frac{1}{40}$. Both simulations use $n=10$ (number of vertices) and $N=5$. \textbf{(B1-B2)} Plots showing the (averaged) spectral gap of Mayer Laplacian $\Delta_{d,p}$ for $d=3,4$ versus number of vertices. The model is the Costa-Farber random complex with the following simplex inclusion probability: $P_1= 1.0, P_2= 1.0, P_3 = 1.0, P_4 = 1.0, P_5= 0.7, P_6 = 0.7$ and $P_i = 0$ for all $i > 6$. The figures show that the gap of $\Delta_{3,p}$ and $\Delta_{4,p}$, in this particular Costa-Farber model, all tend to be larger as the number of vertices increase (at least for many value of $p$).  }
    \label{fig: MayerBettivsSimplices}
\end{figure}
\begin{figure}[htbp]
    \centering
    \includegraphics[width= 0.7\linewidth]{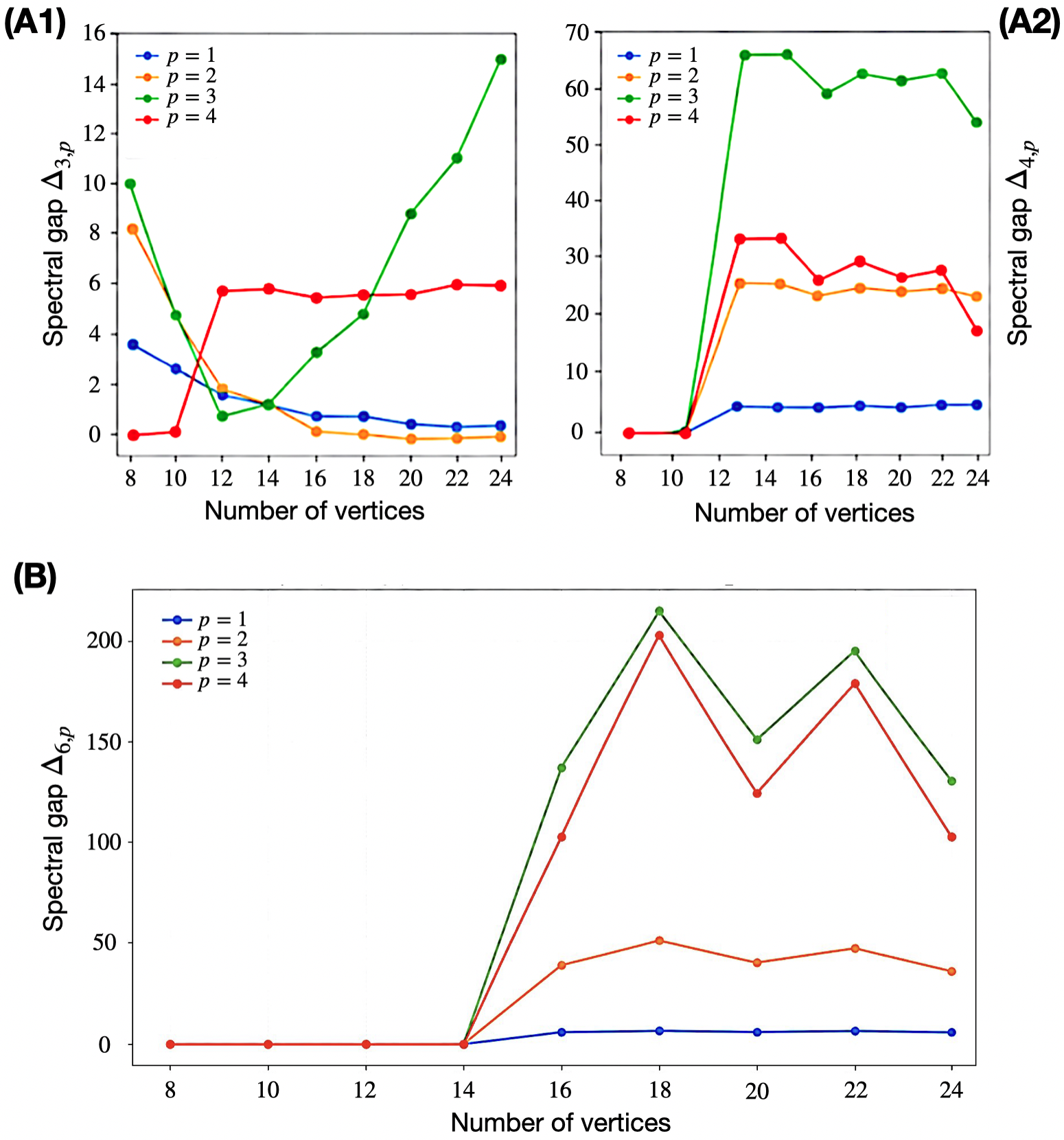}
    \caption{\textbf{Plot between the (averaged) spectral gap of the Mayer Laplacian versus the number of vertices in the Costa-Farber random complex model.} \textbf{(A)} We choose $P_1 =0.9$, and for all $2 \leq i \leq \max_{\dim}$, $P_i = 0.6$, and we plot the spectral gap $\Delta_{3,p}$ \textbf{(A1)} or $\Delta_{4,p}$ \textbf{(A2)} versus the number of vertices $n$. \textbf{(B)} We choose $P_1 =0.6$, and for all $2 \leq i \leq \max_{\dim}$, $P_i = 1.0$, and we plot $\Delta_{6,p}$ versus the number of vertices $n$. The number of vertices is varied as indicated, and $\max_{\dim} = 9, N = 5$ for all cases.   }
    \label{fig: packfigure1}
\end{figure}

\begingroup\color{black}
Figures~\ref{fig: MayerBettivsSimplices} and \ref{fig: packfigure1}
illustrate the behavior of Mayer Betti numbers and spectral gaps in
the sampled models and size ranges. The observed large Betti numbers
are consistent with the possibility of large signals established by
our sufficient conditions, but do not establish typicality across
dense or sparse complexes. Likewise, the finite-size gap trends do not
prove an asymptotic gap lower bound or exclude eventual gap closure.
The cone construction below provides a separate analytic gap bound
for $p=N-1$ and $m\equiv1\pmod N$.
\par\endgroup


\section{Cone-shaped complex $K$ }
\label{sec: provingconecomplex}
In the main text, we have provided a general way to construct the cone-shaped complex $K$, which is as follows. Fix $N \geq 3$ and let $m \geq N-1$. Let $n = 3m$ be the number of vertices. Pick one vertex, denoted as $a$, while the remaining $n-1= 3m-1$ vertices form the set $B$. Let $L_m = \text{Skel}_{m-1} (\Delta^{3m-2})$, meaning that $L_m$ contains every subset of $B$ having at most $m$ vertices. Now form the cone:
\begin{align}
    K = a*\text{Skel}_{m-1} (\Delta^{3m-2}) =  L_m \cup \{  \{a\} \cup \sigma : \sigma \in L_m  \}
\end{align}
Below, we prove certain properties of this cone-shaped complex, including the large value of Mayer Betti numbers and large spectral gap for a value of $p$. We also describe how to obtain the quantum oracle $O_d^K$ for this kind of complex. 

\subsection*{Proof for $K$ having large normalized Mayer Betti numbers.} Consider Mayer homology of the order $d = m-1$. In the following, we will prove that $\frac{\beta_{d,p}}{|S_d^K|} = \Omega(1)$ for all $p$. 
\begin{itemize}
    \item \textbf{(The sectors $1 \leq p \leq N-2$)} Since $p \leq N-2$, then $d+N-p = m-1 + N-p \geq m+1$. 
    But $\dim K = m$, so $S_{d+N-p}(K_m) = \{0\}$, meaning that there is no $(d+N-p)$-simplex. Equivalently, $|S_{d+N-p}| = 0$. At the same time, 
    $$|S_d^{K}| = \binom{n}{d+1} = \binom{3m}{m}, |S_{d-p}^{K}| = \binom{n}{d-p+1} = \binom{3m}{m-p} $$
    So we have that:
    \begin{align}
        \frac{|S_{d-p}^{K}|}{|S_d^{K}|} = \frac{\binom{3m}{m-p}}{\binom{3m}{m} } = \prod_{j=0}^{p-1} \frac{m-j}{2m+j+1}
    \end{align} 
    Since $ \frac{m-j}{2m+j+1} <  \frac{m}{2m+1}  $ for any $j= 1,2,...,p-1$, plus that $ \frac{m}{2m+1} < \frac{m}{2m} = \frac{1}{2}$, we have that $\frac{m-j}{2m+j+1} <\frac{1}{2} $ for all $j$ in the above. So the ratio above is bounded as:
\begin{align}
  \frac{|S_{d-p}^{K}|}{|S_d^{K}|} =   \prod_{j=0}^{p-1} \frac{m-j}{2m+j+1} < \frac{1}{2^p}
  \label{24}
\end{align}

Thus we have the nontrivial lower bound on the normalized Mayer Betti number:
\begin{align}
     \frac{\beta_{d,p}}{|S_d^{K}|} \geq 1-  \frac{|S_{d-p}^{K}|}{|S_d^{K}|}  > 1- \frac{1}{2^p}
\end{align}
\item (\textbf{The remaining sector $p=N-1$)} For $p=N-1$,  we have $d +N-p = d+1 =m$. In this case, we have the following:
\begin{align}
    \frac{|S_{d-p}^K|}{|S_d^K|} = \frac{\binom{3m}{ m-N+1} }{ \binom{3m}{m}}, \frac{|S_{d+N-p}^K|}{|S_d^K|} = \frac{\binom{3m-1}{m}}{\binom{3m}{m} }
\end{align}
The first term is upper bounded by $2^{-(N-1)}$ simply replacing $p = N-1$ in Eq.~\ref{24}. The second term can be computed exactly as:
\begin{align}
   \frac{\binom{3m-1}{m}}{\binom{3m}{m} } = \frac{ \frac{(3m-1)!}{ m!(2m-1)!}  }{ \frac{(3m)!}{m! (2m)!}} = \frac{2m}{3m} = \frac{2}{3}
\end{align}
So we have the following nontrivial lower bound:
\begin{align}
    \frac{\beta_{d,p}}{|S_d^{K}|}  \geq 1- \frac{1}{2^{N-1}} - \frac{2}{3} = \frac{1}{3} - \frac{1}{2^{N-1}}
\end{align}
For example, for $N =3$, this bound is $\frac{1}{12}$. 
\end{itemize}
\begingroup\color{black}
Thus, at $d=m-1$, this family has
$\beta_{d,p}/|S_d^K|=\Omega(1)$ for every $1\leq p<N$ and
$|S_d^K|/\binom{n}{d+1}=1$.
Together with the efficient membership test and the gap bound proved
below for $p=N-1$ and $m\equiv1\pmod N$, these facts give polynomial
quantum running time for these values of $p$ and $m$, for fixed $N$ and
inverse-polynomial relative precision. They establish compatibility
of the algorithm's promises, not a separation from classical methods.
\par\endgroup

\subsection*{Proof for $K$ having large spectral gap when $p=N-1$ and $m \equiv 1$ (mod $N$) } 
\label{sec: proofforconeK}
In the following, we prove that the complex $K$ defined above has a large spectral gap for certain values of $p$ and $m$. First, we recall that the Mayer Laplacian is defined as:
\begin{align}
    \Delta_{d,p}= \left(\partial_d^p\right)^\dagger \partial_d^p +  \partial^{N-p}_{d+N-p} \left(\partial^{N-p}_{d+N-p}\right)^\dagger
\end{align}
where
\begin{align}
    \partial_d^p &:= \partial_{d-p+1} \cdots \partial_{d-1} \partial_d \\
    \partial^{N-p}_{d+N-p} &:= \partial_{d+1} \cdots \partial_{d+N-p-1}\partial_{d+N-p}.
\end{align}
At $d=m-1$, it was shown above that for $1 \leq p \leq N-2$, then $d+N-p \geq m+1$, so that $|S_{d+N-p}^K| = 0$ and thus 
\begin{align}
    \Delta_{d,p} = \left(\partial_d^p\right)^\dagger \partial_d^p.
\end{align}
While for $p= N-1$, $d+N-p = m$, so 
\begin{align}
    \Delta_{d,p}=  \Delta_{m-1,p}= (\partial_{m-1}^{N-1})^\dagger \partial_{m-1}^{N-1} +  \partial_{m} (\partial_{m})^\dagger
\end{align}
In the following, we consider $p=N-1$ specifically and show that when $m \equiv 1$ (mod $N$), the gap of $\Delta_{m-1,N-1}$ is inverse-polynomial. For the purpose of simplifying the notation, we define $A \equiv \partial_{m-1}^{N-1}, \  T \equiv   \partial_m^K$. Then we have
\begin{align}
    \Delta_{m-1,N-1} = A^\dagger A + T T^\dagger.
\end{align}
Because:
\begin{align}
    AT = \partial_{m-1}^{N-1} \partial_m = \partial_m^N = 0
\end{align}
by definition, we have:
\begin{align}
    (A^\dagger A) (T T^\dagger) = 0
\end{align}
It implies that:
\begin{align}
    \text{Im } (A^\dagger A) \perp \text{Im } (T T^\dagger)
\end{align}
\textcolor{black}{Their positive-eigenvalue subspaces are therefore orthogonal.} So, the image space of $  \Delta_{m-1,N-1}$ can be decomposed as:
\begin{align}
    \text{Im } (\Delta_{m-1,N-1} ) = \text{Im } (A^\dagger A) \oplus \text{Im } (T T^\dagger)
\end{align}
This implies that the \textcolor{black}{smallest positive eigenvalue} of $\Delta_{m-1,N-1}$ is:
\begin{align}
    \textcolor{black}{\gamma_{d,p}} = \min \{  \sigma_{\min}^2 (A), \sigma^2_{\min} (T)  \}
    \label{i16}
\end{align}
where $ \sigma_{\min}(.)$ means the smallest nonzero singular value. In the following, we will analyze and establish the lower bound of $\sigma_{\min}^2 (A),  \sigma^2_{\min} (T)   $, which in turn reveal the lower bound of $ \textcolor{black}{\gamma_{d,p}} $. \\

\paragraph{Lower bound of $\sigma^2_{\min} (T) $.} Recall that $T \equiv \partial_m$. \begingroup\color{black}
Each top simplex is $\{a\}\cup S$ for a unique $m$-vertex subset
$S$ of the base $L_m$. Projecting $T=\partial_m$ onto the rows
indexed by faces that do not contain $a$ keeps exactly the term
obtained by deleting $a$ from each column. This square submatrix is
diagonal, with entries of modulus one. Consequently,
\begin{align}
    \|Ty\|^2\geq\|y\|^2
    \quad (y\in C_m^K),
    \qquad \sigma_{\min}^2(T)\geq1.
\end{align}
This argument does not require the apex to be first in the vertex order.
\par\endgroup

\paragraph{Lower bound of $\sigma_{\min}^2(A)$.} Recall that $A \equiv \partial_{m-1}^{N-1}$. In addition, let $N$ be the Mayer-homology constant defined as above ($\partial^N = 0$) and $q := \exp( \frac{2\pi i}{N} )$ (in Section \ref{sec: mayerhomology}, we defined this term by $\xi$). In the following, we aim to \textcolor{black}{construct a generalized inverse $R$ satisfying $ARA=A$ and bound its operator norm by a polynomial in $n$}.

For $S \subseteq \{0,1,...,n-1\}$, define:
\begin{align}
    r_S(j) = |\{   s\in S: s< j  \}|
\end{align}
\textcolor{black}{Let $0\leq k\leq m-2$, so both degrees belong to the full skeleton.} Define an operator $U_k: C_{k}(K) \longrightarrow C_{k+1} (K)$ as follows:
\begin{align}
   U_k \ket{S} = \sum_{j \notin S} q^{r_S(j) - j} \ket{S \cup \{ j\} }
\end{align}
\textcolor{black}{Every nonzero entry of $U_k$ has modulus one. A $k$-simplex has $k+1$ vertices.} Each column has $n-k-1$ nonzero entries, while each row has $k+2$ entries. Therefore:
\begin{align}
    \|U_k\| \leq \sqrt{(n-k-1)(k+2)} \leq n+1.
    \label{i30}
\end{align}
To proceed, we have the following proposition:
\begin{proposition}
\label{prop: qcommutationrelation}
\textcolor{black}{For $2\leq k\leq m-1$}, the following key $q$-commutation relation holds:
    \begin{align}
    \partial_{k} U_{k-1} - q U_{k-2} \partial_{k-1} = -q [2k-n]_q I_{C_{k-1}^K} 
\end{align}
where for any integer $r$:
\begin{align}
    [r]_q = \frac{1-q^r}{1-q}
\end{align}
\end{proposition}
\noindent
\textbf{Proof.} See Appendix \ref{sec: qcommutationrelation}. \\

From this commutation relation, we have:
\begin{align}
    \begin{cases}
         \partial_{k} U_{k-1} - q U_{k- 2} \partial_{k-1} = -q [2k-n]_q I_{C_{k-1}^K}  \\
          \partial_{k-1} U_{k- 2} - q U_{k-3} \partial_{k-2} = -q [2(k-1)-n]_q I_{C_{k-2}^K} 
    \end{cases}
\end{align}
Multiplying the first equation by $U_{k-2}$, we have:
\begin{align}
     \partial_{k} U_{k-1} U_{k-2} - q U_{k-2} \partial_{k-1} U_{k-2} = -q [2k-n]_q  U_{k-2}
\end{align}
By replacing $ \partial_{k-1} U_{k-2}  = q U_{k-3} \partial_{k-2}  -q [2(k-1)-n]_q I_{C_{k-2}^K}   $, we have:
\begin{align}
     \partial_{k}  U_{k-1} U_{k-2} - qU_{k-2} (  q U_{k-3} \partial_{k-2}  -q [2(k-1)-n]_q I_{C_{k-2}^K}  ) =  -q [2k-n]_q  U_{k-2}
\end{align}
which is equivalent to:
\begin{align}
      \partial_{k} U_{k-1} U_{k-2} - q^2 U_{k-2} U_{k-3} \partial_{k-2} =  -q [2k-n]_q  U_{k-2} - q^2 [2(k-1)-n]_q U_{k-2}
\end{align}
Multiplying this equation with $U_{k-3}$, we obtain:
\begin{align}
      \partial_{k} U_{k-1} U_{k-2} U_{k-3} - q^2 U_{k-2} U_{k-3} \partial_{k-2} U_{k-3} =  -(q [2k-n]_q+ q^2 [2(k-1)-n]_q )U_{k-2} U_{k-3}
\end{align}
Again, from the $q$-commutation relation, we have:
\begin{align}
\begin{split}
     &\partial_{k-2} U_{k-3} - q U_{k-4} \partial_{k-3} = -q[ 2(k-2) - n]_q I_{C_{k-3}^K} \\
    &\longrightarrow  \partial_{k-2} U_{k-3} = q U_{k-4} \partial_{k-3}-q[ 2(k-2) - n]_q I_{C_{k-3}^K}
\end{split}
\end{align}
So we have:
\begin{align}
\begin{split}
     & \partial_{k} U_{k-1} U_{k-2} U_{k-3} - q^2 U_{k-2} U_{k-3} ( q U_{k-4} \partial_{k-3} -q[ 2(k-2) - n]_q I_{C_{k-3}^K}) \\
     &=  -( q [2k -n]_q+ q^2 [2(k-1)-n]_q )U_{k-2} U_{k-3} \\
     \longrightarrow  & \partial_{k} U_{k-1} U_{k-2} U_{k-3} - q^3 U_{k-2} U_{k-3}U_{k-4} \partial_{k-3}  = -( q [2k -n]_q+ q^2 [2(k-1)-n]_q  + q^3 [2(k-2) -n]_q  ) U_{k-2} U_{k-3}
\end{split}
\end{align}
\textcolor{black}{For $1\leq l\leq k-1$ and $k\leq m-1$, iterating gives}
\begin{align}
     \partial_k U_{k-1} U_{k-2} \cdots U_{k-l} - q^l U_{k-2} U_{k-3} \cdots U_{k-l-1} \partial_{k-l} =  - \big(  \sum_{i=1}^l q^i [2(k-i+1) -n]_q  \big) U_{k-2} U_{k-3} \cdots U_{k-l}  
\end{align}
Suppose that for \textcolor{black}{$1\leq l\leq k-1$}, there is $y\in C_{k-l}^K$ that satisfies $\partial_{k-l} y = 0$. So from the equation above, by multiplying both sides by $y$, we have:
\begin{align}
    \partial_k U_{k-1} U_{k-2} \cdots U_{k-l} y  = - \big(  \sum_{i=1}^l q^i [2(k-i+1) -n]_q  \big) U_{k-2} U_{k-3} \cdots U_{k-l}   y
\end{align}
In the following, we give a direct calculation of $ \sum_{i=1}^l q^i [2(k-i+1) -n]_q $. By definition, we have the following:
\begin{align}
    [2(k-i+1) -n]_q = \frac{1-q^{2(k-i+1)-n} }{1-q} \longrightarrow q^i [2(k-i+1) -n]_q   =  q^i \frac{1-q^{2(k-i+1)-n} }{1-q} = \frac{q^i - q^{2(k+1) - i -n}}{1-q}
\end{align}
So:
\begin{align}
     \sum_{i=1}^l q^i [2(k+1-i) -n]_q = \sum_{i=1}^l  \frac{q^i - q^{2(k+1) - i -n}}{1-q} = \sum_{i=1}^l \frac{  q^i - q^{2(k+1)-n} q^{-i}  }{1-q}
\end{align}
It is well-known that:
\begin{align}
\begin{split}
    \sum_{i=1}^l q^i& = \frac{q(1-q^l)}{1-q} \\
    \sum_{i=1}^l q^{2(k+1)-n} q^{-i}    &= q^{2(k+1)-n} \frac{(q^{-l}-1)}{1-q} = q^{2(k+1)-n -l} \frac{1-q^l}{1-q}
\end{split}
\end{align}
So we have that:
\begin{align}
     \sum_{i=1}^l q^i [2(k-i+1) -n]_q = \frac{q(1-q^l)}{(1-q)^2} -  q^{2(k+1)-n -l} \frac{1-q^l}{ ( 1-q)^2  } = \frac{1-q^l}{1-q} \big( q \frac{1 - q^{2(k+1)-1 - n-l} }{1-q}  \big)
\end{align}
By definition, we have:
\begin{align}
     \sum_{i=1}^l q^i [2(k+1-i) -n]_q  = q [l]_q  [ 2(k+1-l) + l  -n-1]_q 
\end{align}
Therefore, we have the following relation:
\begin{align}
      \partial_k U_{k-1} U_{k-2} \cdots U_{k-l} \  y  = -q [l]_q  [ 2(k+1-l) + l  -n-1]_q  U_{k-2} U_{k-3} \cdots U_{k-l}  \ y
      \label{i47}
\end{align}
The equation above defines a recursive relation. By replacing $k = k-1$ and $l= l-1$ (because we still use $y$ that satisfies $\partial_{k-l} y = 0$), we have:
\begin{align}
    \partial_{k-1} U_{k-2}U_{k-3} ... U_{k-l} y = -q [l-1]_q  [ 2(k+1-l) + (l-1)  -n-1]_q  U_{k-3}U_{k-4} ... U_{k-l} y
\end{align}
Multiplying both sides of Eq.~\ref{i47} by $\partial_{k-1}$ to the left, we have:
\begin{align}
\begin{split}
    \partial_{k-1} \partial_k U_{k-1}U_{k-2} ... U_{k-l} \ y &= -q [l]_q  [ 2(k+1-l) + l  -n-1]_q  \partial_{k-1} U_{k-2}U_{k-3} ... U_{k-l} \ y \\
    &= (-q [l]_q  [ 2(k+1-l) + l  -n-1]_q ) (  -q [l-1]_q  [ 2(k+1-l) + (l-1)  -n-1]_q )  U_{k-3} U_{k-4} ... U_{k-l} \ y \\
    &= (-1)^2 q^2 [l]_q [l-1]_q[ 2(k+1-l) + l  -n-1]_q  [ 2(k+1-l) + (l-1)  -n-1]_q U_{k-3} U_{k-4} ... U_{k-l} \ y 
\end{split}
\end{align}
Continuing this procedure for a total of $l-1$ times, we obtain:
\begin{align}
    \partial_{k-l+1} ... \partial_{k-1} \partial_k U_{k-1} ... U_{k-l} \ y=  \Gamma_{l}  \Ibb y
\end{align}
where:
\begin{align}
    \Gamma_l := \prod_{j=1}^l (-q) [j]_q [2(k+1-l) + j -n-1]_q  
\end{align}
By choosing $l$ such that $l = p = N-1$ and $k = d = m-1$, we have $\partial_d^p \equiv \partial_{d-p+1} ... \partial_d = \partial_{k -l+1 } ... \partial_{k-1} \partial_k$. So we have:
\begin{align}
   \partial_{d -l+1 } ... \partial_{d-1} \partial_k U_{d-1}... U_{d-N+1} y = \partial_d^p  U_{d-1}... U_{d-N+1} y  =  \Gamma_{l}  \Ibb y
   \label{i52}
\end{align}
Furthermore, if $m = Nt +1$ for some $t \in \Zbb$ and $N$ is the Mayer-constant defined earlier, then we have the following.
\begin{align}
    k+1-l = m - (N-1) = N(t-1)+2
\end{align}
which implies that 
\begin{align}
    2(k+1-l) + j - n-1 = 2N(t-1) +4 +j - 3m -1 = j + 2N(t-1) - 3(Nt+1) +3  = j - Nt -2N = j - N(t+2) 
\end{align}
Because $q^N = 1$, so $[j -N(t+2) ]_q = [j]_q$, so we have:
\begin{align}
    \Gamma_l =  (-q)^{N-1} \prod_{j=1}^{N-1} [j]_q [j]_q
\end{align}
We have the following proposition:
\begin{proposition}
      Let $ \Gamma_l$ be defined as above. Then: 
      \begin{align}
          | \Gamma_l| \geq 1
      \end{align}
\end{proposition}
\noindent
\textbf{Proof.} It can be seen, as $\Gamma_l$ is the product, we have:
\begin{align}
    | \Gamma_l|  = |   (-q)^{N-1} \prod_{j=1}^{N-1} [j]_q [j]_q| = |   (-q)^{N-1}| \prod_{j=1}^{N-1} |  [j]_q |\cdot |  [j]_q |
\end{align}
By definition, we have:
\begin{align}
    [j]_q = \frac{1-q^j}{1-q} = \frac{1- \exp( 2\pi i \frac{j}{N} )}{1- \exp(  2\pi i\frac{1}{N} )}
\end{align}
So:
\begin{align}
\begin{split}
      |  [j]_q | &= \left| \frac{1- \exp( 2\pi i \frac{j}{N} )}{1- \exp(  2\pi i\frac{1}{N} )}   \right|  \\
    &= \left|  \frac{\sin( \pi \frac{j}{N} )}{ \sin ( \pi \frac{1}{N})}   \right| \geq 1
\end{split}
\end{align}
for all $j=1,2,...,N-1$. We also have that $q := \exp(  \frac{2\pi i}{N} )$, so $|-q^{n-1}| =1$. Combining everything, we have:
\begin{align}
    |\Gamma_l|  \geq 1
\end{align}
which means that the proposition is proved. $\blacksquare$ \\

\noindent
\begingroup\color{black}
Return to $d=m-1$, $p=l=N-1$, with $m=Nt+1$ and $t\geq1$.
For $x\in C_d^K$, set $y=Ax=\partial_d^{N-1}x\in C_{d-N+1}^K$.
The condition needed in Eq.~\eqref{i52} is
\begin{align}
    \partial_{d-N+1}y
    =\partial_{d-N+1}\partial_d^{N-1}x
    =\partial_d^Nx=0.
\end{align}
The product $U_{d-1}\cdots U_{d-N+1}$ only uses indices from
$m-N$ to $m-2$, so every raising map lies in the full skeleton.
No map from degree $m-1$ to degree $m$ is used.
\par\endgroup
From Eq.~\eqref{i52}, we have:
\begin{align}
    \partial_d^{N-1} U_{d-1} ...U_{d-N+1} y =   \Gamma_l \Ibb y
\end{align}
which implies that:
\begin{align}
    \partial_d^{N-1} U_{d-1} ...U_{d-l}  \partial_d^{N-1} x = \Gamma_l \partial_d^{N-1} x  \longleftrightarrow  \partial_d^{N-1} \big(  \Gamma_l^{-1} U_{d-1} ... U_{d-N+1} \big) \partial_d^{N-1} = \partial_d^{N-1}
\end{align}
Define $ R \equiv \Gamma_l^{-1} U_{d-1} ... U_{d-N+1}$ for simplicity. Then we have:
\begin{align}
    \|R\| = \| \Gamma_l^{-1} U_{d-1} ... U_{d-N+1} \| \leq \| \Gamma_l^{-1} \| \cdot \| U_{d-1} \|  ... \|U_{d-N+1}\| \leq (n+1)^{N-1}
    \label{i65}
\end{align}
where we have used the inequality in Eq.~\ref{i30}.

To proceed, let $z \perp \ker (\partial_d^{N-1})$, we have:
\begin{align}
    \partial_d^{N-1} ( R \partial_d^{N-1}  z - z) =   \partial_d^{N-1} R \partial_d^{N-1}  z -  \partial_d^{N-1} z = 0.
\end{align}
Therefore, $R  \partial_d^{N-1} z - z \in \ker ( \partial_d^{N-1})$. Let $P$ be the projector onto $\ker(\partial_d^{N-1})^{\perp}$, we have:
\begin{align}
    P ( R  \partial_d^{N-1} z - z) = 0  \leftrightarrow  P  R  \partial_d^{N-1} z = Pz = z
\end{align}
Consequently, due to triangle inequality:
\begin{align}
    \|z\| \leq \| P\| \cdot \|R\| \cdot \|\partial_d^{N-1}z \| = \|R\| \cdot \|\partial_d^{N-1}z \| \leq \textcolor{black}{(n+1)^{N-1}} \|\partial_d^{N-1}z \| \longleftrightarrow 1 \leq (n+1)^{N-1} \frac{\|\partial_d^{N-1}z \| }{ \|z\|}
\end{align}
as $P$ is the projector, so $\|P\|=1$, and we also use the inequality from Eq.~\ref{i65}.  By definition, we have the following:
\begin{align}
    \sigma_{\min} (\partial_d^{N-1} )  := \textcolor{black}{\min_{0\ne z\perp\ker A}} \frac{\|\partial_d^{N-1}z \| }{ \|z\|} \geq \frac{1}{(n+1)^{N-1}}
\end{align} 
Given that we have defined earlier $ A := \partial_d^{N-1}$ and $d= m-1$, the above result in:
\begin{align}
    \sigma_{\min}^2 (A) \equiv \sigma_{\min}^2 (\partial_{m-1}^{N-1}) \geq \frac{1}{(n+1)^{2N-2}}
\end{align}
From Eq.~\ref{i16}, we have:
\begin{align}
     \textcolor{black}{\gamma_{d,p}} = \min \{  \sigma_{\min}^2 (A), \sigma^2_{\min} (T)  \} \geq \frac{1}{(n+1)^{2N-2}}
\end{align}
\textcolor{black}{Therefore, the smallest positive eigenvalue of
$\Delta_{d,p}=\Delta_{m-1,N-1}$ has an inverse-polynomial lower bound
for fixed $N$.}

\subsection*{Building the oracle for cone-shaped complex }

\begingroup
\color{black}
For the cone $K=a*L_m$, membership can be tested directly from its
definition. Every vertex subset of size at most $m$ belongs to $K$,
whereas a subset of size $m+1$ belongs to $K$ exactly when it contains
the apex $a$. No larger subset belongs to $K$. Thus, for a nonempty
vertex subset $S$,
\begin{equation}
    S\in K \quad\Longleftrightarrow\quad
    |S|\leq m \quad\text{or}\quad
    \bigl(|S|=m+1\ \text{and}\ a\in S\bigr).
    \label{eq:cone-membership-predicate}
\end{equation}
In particular, checking the edges and imposing only a dimension cutoff
would incorrectly accept $m$-simplices that do not contain the apex.

To implement $O_d^K$ on the $n$-bit representation of $S$, proceed
as follows.
\begin{enumerate}
 \item First, compute the Hamming weight $w=|S|$ in a work register.
 \item Next, test Eq.~\eqref{eq:cone-membership-predicate} using
 $w$ and the apex bit. For the degree-$d$ oracle, also require
 $w=d+1$.
 \item Finally, write the resulting Boolean value into the output
 qubit and uncompute the comparison and weight registers.
\end{enumerate}
Reversible counting and comparison give a
polynomial-size circuit, without an edge database. At the target
degree $d=m-1$, all subsets of weight $m$ are accepted; at degree $m$,
the apex condition selects precisely the top simplices of $K$.
\par
\endgroup

\begingroup\color{black}
\section{Realistic resource estimation }
\label{sec: realisticresourceestimation}
In practice, all quantum operations are built from elementary gates, such as CNOT, single-qubit rotation, Toffoli gate, etc. Given that the practical quantum algorithm for persistent Mayer homology is provided in Algo.~\ref{algo: practicalpipeline}, which includes an explicit construction of the oracle $O_d^K \  \forall \ d$, in the following we examine how many gates would be required for quantum computers to handle real-world problems. \\

This appendix uses the alternative LCU construction and normalization
procedure in Sections~\ref{sec:boundary-alternative-lcu}
and~\ref{sec:boundary-alternative-amplification}, with the gate-count
assumptions stated below. It is separate from the complexity analysis
of the direct incidence-state construction in
Appendix~\ref{sec: comlexityanalysis}.

\noindent
\textbf{Gate complexity analysis when including gate cost for the oracles.} We consider the entire algorithm when accounting for the gate cost of building the oracle. The algorithm begins with the construction of the block-encoding of the $N$-boundary operator $\partial_d/n$. For the alternative construction in Section~\ref{sec:boundary-alternative-lcu}, the following resource accounting assumes $n$ gates plus a single use of each $O_d^K$ and $O_{d-1}^K$. Since the cost (number of Toffoli gates) required to build $O_d^K$ is $3|E| + 2\log d$, the gate cost to build the block-encoding of $ \partial_d/n$ is $6|E| + n + 2\log\big( d (d-1) \big)$. The next step is to obtain the block-encoding of $\partial_d/\alpha_d$, by using the amplification procedure in Section~\ref{sec:boundary-alternative-amplification} with the polynomial of degree $\mathcal{O}( \frac{n}{\alpha_d} )$ (we are ignoring the logarithmic factor). According to \cite{gilyen2019quantum}, the minimum degree for this polynomial is also $\Omega ( \frac{n}{\alpha_d} )$. Therefore, the number of elementary gates is of order $ \frac{n}{\alpha_d}  \Big( 6|E| + n + 2\log ( d (d-1) ) \Big) $. The gate cost for building the block-encoding of $\frac{1}{\alpha}\Delta_{d,p}$, which is the sum of the gate cost of block-encoding $\{ \frac{1}{\alpha_{d-i}}\partial_{d-i} \}_{i=0}^p, \{ \frac{1}{\alpha_{d+i}}\partial_{d+i} \}_{i=d+1}^{N-p}$, is thus of order:
\begin{align}
    n  \Big( 6|E| + n + 2 \sum_{i=0}^{p} \log ( (d-i) (d-i-1) ) +  2 \sum_{i=1}^{N-p} \log ( (d+i) (d+i-1) ) \Big)  \big(  \sum_{i=0}^{p-1} \frac{1}{\alpha_{d-i}} +  \sum_{i=1}^{N-p}\frac{1}{\alpha_{d+i}}\big)
\end{align}
This value is upper bounded by
\begin{align}
     n  \Big( 6|E| + n + 2\log ( d^{2p} ) +  2  \log ( (d+N-p)^{2(N-p)} ) \Big)  \left(  \sum_{i=0}^{p-1} \frac{1}{\alpha_{d-i}} +  \sum_{i=1}^{N-p}\frac{1}{\alpha_{d+i}}\right).
\end{align}
From the block-encoding of $\frac{1}{\alpha}\Delta_{d,p} $, we need to use Lemma \ref{lemma: qsvt} with a polynomial of degree $\mathcal{O}( \frac{\alpha}{\gamma_{d,p}} \log \frac{1}{\epsilon}  ) $ to obtain the block-encoding of $ P( \frac{1}{\alpha}\Delta_{d,p})$. The degree of this polynomial is shown in \cite{gilyen2019quantum} to be optimal $\Omega( \frac{\alpha}{\gamma_{d,p}}  \log \frac{1}{\epsilon}    ) $. Therefore, the number of elementary gates would be of order (we are ignoring the logarithmic dependence on inverse of error tolerance)
\begin{align}
      n  \left( 6|E| + n + 2\log ( d^{2p} ) +  2  \log ( (d+N-p)^{2(N-p)} ) \right)  \left(  \sum_{i=0}^{p-1} \frac{1}{\alpha_{d-i}} +  \sum_{i=1}^{N-p}\frac{1}{\alpha_{d+i}}\right) \frac{\alpha}{\gamma_{d,p}}
\end{align}
which was shown earlier to be upper bounded by
\begin{align}
     \frac{n}{\gamma_{d,p} }  \Big( 6|E| + n + 2\log ( d^{2p} ) +  2  \log ( (d+N-p)^{2(N-p)} ) \Big) N^2 \frac{ (\alpha_{\max})^{2p} + (\alpha_{\max})^{2(N-p)}  }{\alpha_{\min}}.
\end{align}
We remind that $\alpha_{\max / \min} \equiv \max/\min \{  \{ \alpha_{d-i}\}_{i=0}^{p-1}, \{  \alpha_{d+i}\}_{i=1}^{N-p}  \}$ and $\alpha_{d-i} = \sqrt{  (d+1-i) (n-d+i) }, \alpha_{d+i} = \sqrt{  (d+1+i) (n-d-i)} $. Earlier, when analyzing the algorithm complexity, we have used the bound $(n+1)$ for $\alpha_{d-i},\alpha_{d+i}$, which leads to the polynomial scaling in $n$ in the worst case. However, here, we still use $  \sqrt{  (d+1-i) (n-d+i) }$, as they can be much smaller than $n$ for certain values of $d,i$, as we will see below. 

The next step is to prepare the state $ \ket{\Tilde{\Phi_d^K}}\ket{1}  $ (Lemma \ref{lemma: uniformsuperposition}). The preparation algorithm is introduced in \cite{gilyen2019quantum}, which is based on oblivious amplitude amplification. The procedure uses $\Omega( \sqrt{\frac{\binom{n}{d+1}}{|S_d^K|}} \log \frac{1}{\epsilon}  ) $ numbers of calls to $O_d^K$ and $dn$ another $1,2$-qubit gates, so the number of elementary gates is of the order
\begin{align}
     \sqrt{\frac{\binom{n}{d+1}}{|S_d^K|}} (dn + 3|E| + 2 \log d).
\end{align}
Again, we are ignoring the logarithmic dependence on the inverse of error. The final step is estimating the overlaps
$$ \bra{\Tilde{\Phi_d^K}}\bra{1}  P( \frac{1}{\alpha}\Delta_{d,p}) \ket{\Tilde{\Phi_d^K}}\ket{1}  $$
which incurs the total number of gates by a factor of $\Omega \left(\frac{1}{\epsilon} \right)$. So the total number of gates is of the order
\begin{align}
\begin{split}
    \Big(   \frac{n}{\gamma_{d,p} }     \Big( 6|E| + n + 2\log ( d^{2p} ) +  2  \log ( (d+N-p)^{2(N-p)} ) \Big) N^2  \frac{ (\alpha_{\max})^{2p} + (\alpha_{\max})^{2(N-p)}  }{\alpha_{\min}}  + \\
    2    \sqrt{\frac{\binom{n}{d+1}}{|S_d^K|}} (dn + 3|E| + 2 \log d)  \Big) \frac{1}{\epsilon}.
\end{split}
\end{align}
By replacing $\epsilon =  \delta \sqrt{\frac{|S_d^K|}{\beta_{d,p}}}$, we have the order of the number of gates
\begin{align}
\begin{split}
     \Big(  \frac{n}{\gamma_{d,p} }    \Big( 6|E| + n + 2\log ( d^{2p} ) +  2  \log ( (d+N-p)^{2(N-p)} ) \Big) N^2 \frac{ (\alpha_{\max})^{2p} + (\alpha_{\max})^{2(N-p)}  }{\alpha_{\min}}  +  \\
     2    \sqrt{\frac{\binom{n}{d+1}}{|S_d^K|}} (dn + 3|E| + 2 \log d) \Big) \frac{1}{\delta} \sqrt{\frac{|S_d^K|}{|\beta_{d,p}|} }.
\end{split}
\end{align}
Again, we remark that the discussion above is in the case where $d > p$. When $d \leq p$, the Mayer Laplacian only contain one term $\Delta_{d,p} = \partial_{d+N-p}^{N-p}\left(\partial_{d+N-p}^{N-p}\right)^\dagger$. The complexity remains largely the same, except that there is no factor $\alpha_{\max}^{2p}$, so the complexity is
\begin{align}
\begin{split}
    \Biggl(  \frac{n}{\gamma_{d,p} }    \left( 6|E| + n + 2\log ( d^{2p} ) +  2  \log [ (d+N-p)^{2(N-p)} ] \right) N^2 \frac{ \alpha_{\max}^{2(N-p)}  }{\alpha_{\min}}  +  \\
     2    \sqrt{\frac{\binom{n}{d+1}}{|S_d^K|}} (dn + 3|E| + 2 \log d) \Biggl) \frac{1}{\delta} \sqrt{\frac{|S_d^K|}{|\beta_{d,p}|} }.
\end{split}
\end{align}

\noindent
\textbf{Toffoli complexity.} Inspired by \cite{berry2024analyzing}, if we consider the Toffoli gates only, then we simply note that the oracle $O_d^K$ requires $3|E| + 2\log d$ Toffoli gates. Since most of the subroutines used in the algorithm require the oracle usage, the gate complexity above is essentially the number of Toffoli gates. The only difference is that in the gate count for the preparation of $  \ket{\Tilde{\Phi_d^K}}\ket{1} $, we don't need to include those $dn$ gates, but rather only the  $\Omega( \sqrt{\frac{\binom{n}{d+1}}{|S_d^K|}} \log \frac{1}{\epsilon}  ) $  calls to the oracle $O_d^K$ that requires $6|E| + n$ Toffoli gates. As a result, the number of Toffoli gates required is of the order:
\begin{align}
\begin{split}
      \Big(  \frac{n}{\gamma_{d,p} }    \Big( 6|E| + n + 2\log ( d^{2p} ) +  2  \log ( (d+N-p)^{2(N-p)} ) \Big) N^2 \frac{ \alpha_{\max}^{2p} + \alpha_{\max}^{2(N-p)}  }{\alpha_{\min}}  + 
       \\
       2    \sqrt{\frac{\binom{n}{d+1}}{|S_d^K|}} ( 3|E| + 2 \log d) \Big) \frac{1}{\delta} \sqrt{\frac{|S_d^K|}{ |\beta_{d,p}|}}
\end{split}
\end{align}
\par\endgroup

\section{Proof of Proposition \ref{prop: qcommutationrelation}}
\label{sec: qcommutationrelation}
\begingroup\color{black}
Let $2\leq k\leq m-1$ and $|S|=k$. All vertex additions below
remain in the full skeleton.
\paragraph{The first term $\partial_kU_{k-1}$.}
For $i\in S$ and $j\notin S$, the
coefficient of $\ket{(S\setminus\{i\})\cup\{j\}}$ in
$\partial_kU_{k-1}\ket S$ is
\begin{equation}
    q^{r_S(j)-j+r_S(i)+\mathbf{1}_{\{j<i\}}}.
\end{equation}
\paragraph{The second term $qU_{k-2}\partial_{k-1}$.}
The coefficient of the same basis state in
$qU_{k-2}\partial_{k-1}\ket S$ is
\begin{equation}
    q^{1+r_S(i)+r_S(j)-\mathbf{1}_{\{i<j\}}-j}.
\end{equation}
These agree because $i\ne j$. Thus the off-diagonal terms cancel
in the difference, although they need not vanish separately.
The remaining diagonal coefficient is
\begin{equation}
    \sum_{j\notin S}q^{2r_S(j)-j}
    -q\sum_{i\in S}q^{2r_S(i)-i}.
\end{equation}
\paragraph{Evaluation of the remaining diagonal sum.}
To evaluate this sum, put $a_t=2r_S(t)-t$ for $0\leq t\leq n$.
Then $a_0=0$, $a_n=2k-n$, and
\begin{equation}
    a_{t+1}=
    \begin{cases}
        a_t-1,&t\notin S,\\
        a_t+1,&t\in S.
    \end{cases}
\end{equation}
In both cases the contribution of vertex $t$ is
$\frac{q}{q-1}(q^{a_t}-q^{a_{t+1}})$. Telescoping gives
\begin{align}
    \sum_{j\notin S}q^{2r_S(j)-j}
    -q\sum_{i\in S}q^{2r_S(i)-i}
    &=\frac{q}{q-1}(1-q^{2k-n})\\
    &=-q[2k-n]_q.
\end{align}
Therefore
\begin{equation}
    \partial_kU_{k-1}-qU_{k-2}\partial_{k-1}
    =-q[2k-n]_q I_{C_{k-1}^K},
\end{equation}
as claimed. $\blacksquare$
\par\endgroup

\section{COMBINATORIAL FORMULAS FOR MAYER BOUNDARIES AND LAPLACIANS}
\label{sec: mayercombinatorics}

Let $K$ be a globally ordered simplicial complex and let $\xi$ be a
primitive $N$th root of unity.  On an ordered simplex, the Mayer boundary is
\begin{equation}
  \partial_d[v_0\cdots v_d]
  =\sum_{i=0}^d \xi^i
  [v_0\cdots\widehat v_i\cdots v_d].
  \label{eq:mayer-boundary-ordered}
\end{equation}
For $r\geq 1$, write
\begin{equation}
  \partial_d^r
  :=\partial_{d-r+1}\cdots\partial_{d-1}\partial_d
  :C_d^K\longrightarrow C_{d-r}^K,
  \qquad \partial_d^0:=I_{C_d^K}.
  \label{eq:iterated-mayer-boundary}
\end{equation}
Write
\begin{equation}
  [k]_\xi:=1+\xi+\cdots+\xi^{k-1},
  \qquad
  [k]_\xi!:=\prod_{j=1}^k[j]_\xi,
  \qquad [0]_\xi!:=1.
  \label{eq:quantum-factorial}
\end{equation}

\begin{definition}[Higher Mayer incidence coefficient]
For a $d$-simplex $\sigma$, an arbitrary simplex $\tau$, and $r\geq0$, define
\begin{equation}
  \Minc{\sigma}{\tau}{r}
  :=\bra{\tau}\partial_d^r\ket{\sigma}.
  \label{eq:mayer-incidence}
\end{equation}
In particular, $\Minc{\sigma}{\tau}{r}=0$ unless $\tau$ is a face of
$\sigma$ obtained by deleting exactly $r$ vertices.
\end{definition}

For $r=1$, $N=2$, and $\xi=-1$, this recovers the usual oriented incidence
number $[\sigma:\tau]\in\{0,\pm1\}$.  For $N>2$, it is instead a
cyclotomic coefficient determined by the global vertex order.

\begin{lemma}[Power normal form]
\label{lem:power-normal-form}
Let $\sigma=[v_0\cdots v_d]$, and let $\tau$ be obtained by deleting the
vertices in the original positions $i_1<\cdots<i_r$.  Then
\begin{equation}
  \Minc{\sigma}{\tau}{r}
  =[r]_\xi!\,
  \xi^{\sum_{j=1}^r i_j-\binom{r}{2}}.
  \label{eq:power-normal-form}
\end{equation}
In particular, $\partial_d^N=0$, while $[r]_\xi!\neq0$ for
$1\leq r<N$.
\end{lemma}

\begin{proof}
Fix a deletion order $\pi\in S_r$, namely
\begin{equation}
  v_{i_{\pi_1}}
  \longrightarrow
  v_{i_{\pi_2}}
  \longrightarrow\cdots\longrightarrow
  v_{i_{\pi_r}}.
  \label{eq:deletion-order}
\end{equation}
When the vertex originally at $i_{\pi_t}$ is deleted, its current position is
\begin{equation}
  i_{\pi_t}-\#\{s<t:i_{\pi_s}<i_{\pi_t}\}.
\end{equation}
Because $i_1<\cdots<i_r$, the second term counts the non-inversions of
$\pi$.  Hence the sum of the current deletion positions is
\begin{equation}
  \sum_{j=1}^r i_j-\binom{r}{2}+\operatorname{inv}(\pi),
\end{equation}
where
\begin{equation}
  \operatorname{inv}(\pi)
  :=\#\{(a,b):1\leq a<b\leq r,\ \pi_a>\pi_b\}.
\end{equation}
Summing over the possible deletion orders gives
\begin{align}
  &\xi^{\sum_j i_j-\binom{r}{2}}
  \sum_{\pi\in S_r}\xi^{\operatorname{inv}(\pi)}
  =
  \xi^{\sum_j i_j-\binom{r}{2}}[r]_\xi!,
\end{align}
where we used the standard inversion generating function
\begin{equation}
  \sum_{\pi\in S_r}t^{\operatorname{inv}(\pi)}
  =\prod_{k=1}^r(1+t+\cdots+t^{k-1})=[r]_t!.
\end{equation}
Since $[N]_\xi=0$, this proves $\partial_d^N=0$.  Primitivity gives
$[j]_\xi=(1-\xi^j)/(1-\xi)\neq0$ for $1\leq j<N$.
\end{proof}

\begin{remark}
The proof requires only that $\xi$ be a primitive $N$th root of unity; it
does not require $N$ to be prime.
\end{remark}

For a face $\tau\subseteq\sigma$ obtained by deleting $q$ vertices, define
\begin{equation}
  e_\sigma(\tau)
  :=\sum_{j=1}^q i_j-\binom{q}{2},
  \label{eq:deletion-exponent}
\end{equation}
where the $i_j$ are the deleted positions.  Let $1\leq p<N$ and set
$q=N-p$.  The Mayer Laplacian is
\begin{equation}
  \Delta_{d,p}
  :=(\partial_d^p)^\dagger\partial_d^p
  +\partial_{d+q}^q(\partial_{d+q}^q)^\dagger.
  \label{eq:mayer-laplacian-combinatorial}
\end{equation}
Terms involving a negative chain degree or a degree above the dimension of
$K$ are understood to vanish.

\begin{proposition}[Mayer-Laplacian matrix elements]
\label{prop:matrix-elements}
For $d$-simplices $\sigma$ and $\sigma'$, the down term is
\begin{align}
  \bra{\sigma'}(\partial_d^p)^\dagger\partial_d^p\ket{\sigma}
  &=|[p]_\xi!|^2
  \sum_{\substack{\tau\in S_{d-p}^K:\\
                   \tau\subseteq\sigma\cap\sigma'}}
  \xi^{e_\sigma(\tau)-e_{\sigma'}(\tau)},
  \label{eq:down-matrix-element}
\end{align}
and the up term is
\begin{align}
  \bra{\sigma'}\partial_{d+q}^q(\partial_{d+q}^q)^\dagger\ket{\sigma}
  &=|[q]_\xi!|^2
  \sum_{\substack{\omega\in S_{d+q}^K:\\
                   \sigma\cup\sigma'\subseteq\omega}}
  \xi^{e_\omega(\sigma')-e_\omega(\sigma)}.
  \label{eq:up-matrix-element}
\end{align}
\end{proposition}

\begin{proof}
Insert a resolution of the identity between each boundary power and its
adjoint, and apply Lemma~\ref{lem:power-normal-form}.  The complex conjugates
have the displayed orientation because $|\xi|=1$.
\end{proof}

These formulas are the Mayer analogue of lower and upper adjacency for an
ordinary combinatorial Laplacian.  Each common face or containing simplex
contributes a root-of-unity phase, so multiple contributions may cancel.

For a $d$-simplex $\sigma$, define its $q$-step upper degree by
\begin{equation}
  \deg_q^\uparrow(\sigma)
  :=\#\{\omega\in S_{d+q}^K:\sigma\subseteq\omega\},
\end{equation}
and for a $(d-p)$-simplex $\tau$ define
\begin{equation}
  \deg_{p,d}^\uparrow(\tau)
  :=\#\{\sigma'\in S_d^K:\tau\subseteq\sigma'\}.
\end{equation}

\begin{corollary}[Diagonal degree formula]
\label{cor:diagonal-degree}
For every $d$-simplex $\sigma$,
\begin{equation}
  \bra{\sigma}\Delta_{d,p}\ket{\sigma}
  =|[p]_\xi!|^2\binom{d+1}{p}
  +|[q]_\xi!|^2\deg_q^\uparrow(\sigma).
  \label{eq:diagonal-degree}
\end{equation}
The first term is set to zero when $d-p<0$.
\end{corollary}

\begin{proof}
On the diagonal, all phase differences vanish.  The down term counts the
faces obtained from $\sigma$ by deleting $p$ vertices, and the up term
counts the simplices obtained by adding $q$ vertices to $\sigma$.
\end{proof}

\begin{corollary}[Support and norm bounds]
\label{cor:sparsity}
The number of columns that may be nonzero in the $\sigma$ row satisfies
\begin{equation}
  s_{d,p}(\sigma)
  \leq
  \sum_{\substack{\tau\subseteq\sigma\\\dim\tau=d-p}}
  \deg_{p,d}^\uparrow(\tau)
  +\deg_q^\uparrow(\sigma)\binom{d+q+1}{q}.
  \label{eq:local-sparsity}
\end{equation}
Let
\begin{equation}
  D_\downarrow
  :=\max_{\tau\in S_{d-p}^K}\deg_{p,d}^\uparrow(\tau),
  \qquad
  D_\uparrow
  :=\max_{\sigma\in S_d^K}\deg_q^\uparrow(\sigma),
\end{equation}
with an empty maximum interpreted as zero.  Then
\begin{equation}
  s_{d,p}
  \leq
  \binom{d+1}{p}D_\downarrow
  +D_\uparrow\binom{d+q+1}{q},
  \label{eq:global-sparsity}
\end{equation}
and
\begin{align}
  \norm{\Delta_{d,p}}
  &\leq
  |[p]_\xi!|^2\binom{d+1}{p}D_\downarrow
  +|[q]_\xi!|^2D_\uparrow\binom{d+q+1}{q}.
  \label{eq:norm-bound}
\end{align}
\end{corollary}

\begin{proof}
For the second term, a supported column must share with $\sigma$ at least one
face obtained by deleting $p$ vertices.  Summing the number of candidate
$d$-simplices over these faces gives the first term.  For the first term, each
$(d+q)$-simplex containing $\sigma$ contains
$\binom{d+q+1}{q}$ candidate $d$-faces.  This may overcount duplicates and
ignores cyclotomic cancellation, so it is an upper bound.  Applying the
triangle inequality before summing a row gives the same bounds weighted by
the appropriate quantum factorials.  Hermiticity and
$\norm{A}_2\leq\sqrt{\norm{A}_1\norm{A}_\infty}$ imply
Eq.~\eqref{eq:norm-bound}.
\end{proof}

\end{document}